\documentclass[11pt]{article}

\usepackage[utf8]{inputenc}
\usepackage[T1]{fontenc}
\usepackage{amsmath,amssymb,amsthm}
\usepackage{hyperref}
\usepackage{geometry}
\usepackage{authblk}
\usepackage[numbers,sort&compress]{natbib}
\usepackage{booktabs}
\hypersetup{
	colorlinks=true,
	linkcolor=blue,
	citecolor=blue,
	urlcolor=blue
}

\usepackage{nicefrac} 
\usepackage{thmtools} 
\usepackage{mathtools}

\newcommand{\sat}{\mathit{sat}}
\newcommand{\phragmen}{Phragm\'{e}n}

\newcommand{\calR}{\mathcal{R}}
\newcommand{\elec}{\mathcal{E}}

\newcommand{\NP}{{\normalfont \textsf{NP}}}
\newcommand{\coNP}{{\normalfont \textsf{coNP}}}

\newcommand{\owascore}{\mathit{OWA}}
\newcommand{\thielescore}{\mathit{Thiele}}

\theoremstyle{plain}
\newtheorem{theorem}{Theorem}
\newtheorem{proposition}[theorem]{Proposition}
\newtheorem{claim}{Claim}
\newtheorem{definition}{Definition}
\newtheorem{example}{Example}

\usepackage{environ}
\makeatletter
\newcommand{\problemtitle}[1]{\gdef\@problemtitle{#1}}
\newcommand{\probleminput}[1]{\gdef\@probleminput{#1}}
\newcommand{\problemquestion}[1]{\gdef\@problemquestion{#1}}
\NewEnviron{problem}{
	\problemtitle{}\probleminput{}\problemquestion{}
	\BODY
	\par\addvspace{.4\baselineskip}
	\noindent
	\begin{tabular}{@{\hspace{0pt}}lp{0.752\columnwidth}}
		\multicolumn{2}{@{\hspace{0pt}}l}{\textsc{\@problemtitle}} \\
		\textbf{Input:} & \@probleminput \\
		\textbf{Question:} & \@problemquestion
	\end{tabular}
	\par\addvspace{.4\baselineskip}
}
\makeatother

\newenvironment{claimproof}{\par\noindent\textit{\underline{Proof}:}}{\leavevmode\unskip\penalty9999 \hbox{}\nobreak\hfill\quad\hbox{$\diamond$}}

\title{Free-Riding in Multi-Issue Decisions}

\author[1]{Martin Lackner}
\author[2]{Jan Maly}
\author[3]{Oliviero Nardi\thanks{Corresponding Author: oliviero.nardi@tuwien.ac.at}}

\affil[1]{DPKM, WU Wien, Vienna, Austria. \texttt{lackner@dbai.tuwien.ac.at}}
\affil[2]{DPKM, WU Wien and DBAI, TU Wien, Vienna, Austria. \texttt{jan.maly@wu.ac.at}}
\affil[3]{DBAI, TU Wien, Vienna, Austria. \texttt{oliviero.nardi@tuwien.ac.at}}

\date{}

\begin{document}
	
	\maketitle
	
	\begin{abstract}
		Voting in multi-issue domains allows
		for compromise outcomes that satisfy all voters to some extent.
		Such fairness considerations, however, open the possibility
		of a special form of manipulation: free-riding.
		By untruthfully opposing a popular opinion in one issue, 
		voters can receive increased consideration in other issues.
		We study under which conditions this is possible and show that even weak fairness considerations enable free-riding. 
		Additionally, we study free-riding from a computational and experimental point of view. 
		Our results show that free-riding in multi-issue domains is often possible,
		but comes at a non-negligible individual risk for voters. Thus, the allure of free-riding
		is smaller than one could intuitively assume.
	\end{abstract}
	
	\section{Introduction}
	
	In our increasingly digital and interconnected world, many new collective decision problems
	arise, with applications ranging from social media to the fine-tuning of
	AI systems through human feedback \cite{BakkerCSTCBMGAB22,AlamdariEP24}.  
	In contrast to traditional collective decisions like political elections,
	these new domains often require much more complex and interlinked decisions, e.g., 
	by requiring the same voter (evaluator) to make many different decisions about the desired behavior of an
	AI system. In practice, such decisions are often made using simple majoritarian aggregation 
	rules, which leads to underrepresentation of minority opinions and bias.\footnote{This issue is discussed in the context of reinforcement learning from human feedback, e.g., in a recent survey \cite[in Section~3.2.1]{casper2023open}, in the context of participatory budgeting~\cite{Fairstein2023Participatory}, and the aggregation of moral judgments~\cite{ChandakGP24}.}
	Recent research in the area of Computational Social Choice has shown that we can 
	often leverage the fact that several different issues need to be decided to achieve a
	fairer comprise outcome \cite{Lackner2020,abcsurvey,LacknerM23,ChandakGP24,Peters24}.
	The combinatorial richness of a sequence of decisions opens the opportunity for incorporating a more diverse set of opinions, and to respect a broader range of preferences.

However, by striving for fairness across multiple issues, we open the door to a specific, simple form of manipulation: \emph{free-riding}.
We define free-riding as untruthfully opposing a necessarily winning
candidate. That is, if there is a very popular (maybe unanimous) alternative
for a certain issue, it typically does not change the outcome
if one voter does not approve of it. 
Under the assumption
that multi-issue voting rules try to establish fairness based on satisfaction,
it will consequently give this voter additional consideration in other issues.
Consequently, presented with a popular alternative that is certain to win, a voter may be tempted to misrepresent their preferences and by that artificially lowers their (calculated) satisfaction.
As we show in our paper, this form of manipulation is possible almost universally in multi-issue voting.

Thus, it may appear that free-riding is a form of simple and risk-free manipulation.
The main contribution of our paper is to challenge this intuition. While free-riding is 
indeed often a successful form of manipulation, it is far from trivially beneficial for free-riding individuals.
For our analysis, we mostly focus on the setting where issues are decided sequentially, because free-riding occurs very naturally in this setting. However, 
we also consider the case where issues are decided simultaneously, i.e., where voters 
present their preferences for all issues at the same time, and all issues are decided in one shot.
Within both scenarios, we consider voting rules based on order-weighted averages, short: OWA \cite{yager1988ordered,amanatidis2015multiple}, and
on Thiele score, inspired by multiwinner voting \cite{Thie95a,abcsurvey}.
We obtain the following results:

\begin{itemize}
\item First, we show via axiomatic analysis that essentially every reasonable voting rule is susceptible to free-riding,
both in the sequential and the non-sequential case.
The only notable exception is the utilitarian rule, which maximizes the sum of utilities and which is thus a fully majoritarian rule.
\item For sequential OWA and Thiele rules, we observe an interesting phenomenon. Here,
it may be that free-riding in an issue leads to a \emph{lower} satisfaction in subsequent
issues.
Thus, free-riding for these voting rules is not risk-free.
\item Moreover, we show that it is a computationally hard task to determine
 whether free-riding is beneficial in the long run.
 This is in contrast to simply determining whether free-riding is possible.
 Note that the computation of long-term consequences requires full preference information about all (subsequent) issues; the decision to free-ride requires only information about a single issue.
Hence, free-riding is either inherently risky (in the case of incomplete information) or determining its eventual consequences is at least computationally difficult (in the case of complete information).
\item Voters may still decide to free-ride if the risk is small enough.
To study this question, we complement our theoretical analysis with numerical simulations to quantify this risk.
In case that a voter can free-ride only in one issue (e.g., because it is difficult to recognize popular alternatives), our simulations show that the risk of free-riding is indeed significant.
Nonetheless, positive outcomes of free-riding are more likely than negative outcomes; the exact risk depends largely on the voting rule in use.
In case that a voter can repeatedly free-ride (i.e., using every free-riding opportunity), the risk of free-riding becomes negligible. 
This is because positive outcomes of free-riding are more likely, and hence repeated free-riding makes net-negative outcomes very unlikely.
\item Finally, we consider the non-sequential setting. We study global optimization rules, which decide all issues at once. We show that, even when the winner of an issue is already known, it remains computationally hard to determine whether free-riding is possible (without changing the winner).
Thus, in this setting, even single-issue free-riding is computationally difficult
and requires full information. 
Hence, we conclude that for global optimization rules, the problem of free-riding is essentially the same as the more general problem of strategic voting, since both require full information and high computational power.
\end{itemize}

In general, we conclude that free-riding in multi-issue voting is not as simple and risk-free as one could intuitively assume. Our results paint a nuanced picture of where free-riding can easily occur and where it is difficult to free-ride.

\subsection{Structure}
Our paper is structured as follows. In Section~\ref{sec:model}, we introduce our model for multi-issue elections and free-riding. This is followed by our theoretical contribution related
to the possibility and risk of free-riding (Section~\ref{sec:strategyproof}) and regarding the computational complexity of free-riding (Section~\ref{sec:complexity}). 
We present numerical simulations in Section~\ref{sec:experiments}.
While the majority of this paper is based on sequential voting rules (deciding one issue after the other), we consider the viewpoint of global optimization rules in Section~\ref{sec:global}.
Finally, in Section~\ref{sec:discussion}, we summarize our results and discuss directions for future work.
We have moved especially long or technical proofs to Appendix~\ref{app:complexity} (computational complexity) and~\ref{app:global} (global optimization rules).

\medskip
This work is an extended version of a previous conference publication \cite{aamas-freeriding}.

\subsection{Related Work}

Our work falls in the broad domain of \emph{voting in combinatorial domains}~\cite{Lang2016VotinginCombinatorial}.
In contrast to many works in this field \cite[e.g.,][]{brams1998paradox,ahn2012combinatorial,conitzer2012paradoxes,brams1997voting,lang2009sequential}), we assume that voters' preferences
are separable (i.e., independent) between issues.
Separability is a very natural assumption in our sequential setting, as voters may be only asked to state their preferences after preceding issues have already been decided.

Our work is most closely related to papers on multiple referenda.
In this setting, \citet{amanatidis2015multiple} study the computational complexity of
OWA voting rules, including questions of strategic voting.
In a similar model, \citet{barrot2017manipulation} consider questions
of manipulability: how does the OWA vector impact the susceptibility
to manipulation?
In contrast to our paper, these two papers do not consider free-riding.
We discuss more technical connections between these papers and ours later in the text.

Another related formalism is \emph{perpetual voting}~\cite{Lackner2020}, which essentially
corresponds to voting on multi-issue decisions in sequential order. 
In this setting, issues are chronologically ordered, i.e., decided one after the other.
The work of \citet{Lackner2020} and its follow-up by Lackner and Maly~\cite{lm-pv,LacknerM23} do not consider strategic issues. \citet{Kozachinskiy2025optimal} further study this setting, providing bounds on the dissatisfaction of voters.
Moreover, \citet{BulteauHPRT21} move to a non-sequential (offline) model of perpetual voting and therein study proportional representation; we briefly consider this model in Section~\ref{sec:global}. In the same model, \citet{Elkind2024TemporalElections} study proportionality and strategyproofness for the utilitarian and egalitarian rules. \citet{ChandakGP24} and \citet{Elkind2025Verifying} study proportionality in both the sequential and offline settings. Finally, note that perpetual voting and its offline variant have also been referred to as \emph{temporal voting}; see the survey by \citet{Elkind2024TemporalSurvey}.

A third related formalism is that of \emph{public decision making}, which has been studied both in the offline~\cite{conitzer2017fair} and sequential~\cite{freeman2017fair} formulations. As in our model, public decision making considers $k$~issues and for each one alternative has to be chosen.
This model is more general than ours in that it allows arbitrary additive utilities (whereas we consider only binary utilities, i.e., approval ballots).
Our works differ in that \citet{conitzer2017fair} and \citet{freeman2017fair} focus on fairness properties, whereas our focus is on strategic aspects.
Fairness considerations in public decision making have further been explored by \citet{skowron2022proportional} and \citet{Kahana2024Leximin}.

Our model is also related to \emph{multi-winner voting} \cite{abcsurvey,Faliszewski2017MultiwinnerVoting}.
The main difference is that instead of selecting $k$~candidates from the same set of candidates, we have individual candidates for each of the $k$~issues.
In our paper, we adapt the class of Thiele rules from the multi-winner setting to ours. 
This class has been studied extensively, both axiomatically \cite{Aziz2017Justifiedrepresentationin,Sanchez-Fernandez2017Proportional,jet-consistentabc,sanchez2019monotonicity} and computationally \cite{azi-gas-gud-mac-mat-wal:c:multiwinner-approval,skowron2016finding,god-bat-sko-fal:c:2d,brill2022approval}.
The concept of free-riding has also been considered for multi-winner elections~\cite{hylland1992proportional,schulze2004free}.
In particular, in the approval setting, authors have studied this
by looking at \emph{subset-manipulation}, i.e., manipulating by submitting only
a subset of one's truly approved candidates~\cite{pet:prop-sp,botan2021manipulability}.
We note that, while this notion is related to ours in spirit,
it is technically distinct: there is no requirement for the untruthfully-disapproved candidates to be in the truthful outcome. Subset-manipulation has also been studied in the setting of \emph{fair mixing}~\cite{brandl2021distribution}, a generalization of multi-winner voting where the output is a probability distribution over candidates.

In multi-winner voting, there is also substantial literature on the relationship between fairness (often proportionality) and strategyproofness \cite[e.g.,][]{pet:prop-sp,ijcai/LacknerSkowron-multiwinner-strategyproofness,lee2018representing,KVVBE20:strategyproofness}. \citet{delemazure2022strategyproofness} study strategyproofness  in the related model of \emph{party-approval multi-winner voting}, which is a special case of our model (where alternatives and preferences are constant across issues).


To conclude, we note that free-riding is a very general phenomenon and has been widely studied in the economic literature on public goods \cite{groves1977optimal,samuelson1954pure} and philosophy \cite{tuck2008free}. It has also been considered in more technical domains, such as free-riding in memory sharing \cite{friedman2019fair} and in peer-to-peer networks \cite{adar2000gnutella}.

\section{The Model}\label{sec:model}

As is customary, we write $[k]$ to denote $\{1,\dots, k\}$.

We study a form of multi-issue decision making, where for each issue there are two or more possible options (candidates) available.
We mostly focus on the sequential case where issues are decided one after the other.
Furthermore, we assume that for each issue each voter submits an approval ballot, i.e., a subset of candidates that she likes.
Formally,
$k$ denotes the number of issues and $C_1,\dots,C_k$ the respective sets of candidates.
Let $N=[n]$ denote the set of voters.
We write $A_i(v)\subseteq C_i$ for the approval ballot of voter~$v$ concerning issue~$i$.
In combination, we call such a triple $\elec=\langle (C_i)_{i=1}^k,\, N,\, (A_i)_{i=1}^k\rangle$ an \emph{election}.
If $k$ is clear from the context, we write $\bar C$ for $(C_i)_{i=1}^k$ and 
$\bar A$ for $(A_i)_{i=1}^k$.

A \emph{partial outcome} of an election is an $\ell$-tuple $\bar w = (w_1,\dots,w_\ell)$ with $w_i\in C_i$ and $\ell\in[k]$.
If $\ell=k$, then we refer to $\bar w$ simply as an \emph{outcome}.
Given an election~$\elec$ and a (possibly partial) outcome $\bar w$ of length $\ell$,
the \emph{satisfaction} of voter $v\in N$ with $\bar w$ is $\sat_\elec(v, \bar w)=|\{i\in [\ell]: w_i\in A_i(v)\}|$.
In other words, the satisfaction of a voter is the number of issues that were decided in this voter's
favour.\footnote{Note that different notions of satisfaction are possible; for instance, we could assume that voters
have fine-grained preferences over issues and candidates. However, our simple model is a natural starting
point, and we leave the investigation of different notions of satisfaction for future work.}
Furthermore, we write $s_\elec(\bar w)=(s_1,\dots,s_n)$ to denote the $n$-tuple of satisfaction scores $\left(\sat_\elec(v, \bar w)\right)_{v\in N}$ sorted in increasing order, i.e., $s_1\leq s_2\leq\cdots\leq s_n$.
If the election~$\elec$ is clear from the context, we omit it in the notation.

Given an election $\elec$ and a partial outcome for the first $\ell<k$ issues, a \emph{sequential voting rule} returns the winner of the $(\ell+1)$-th issue. A straightforward example is the \emph{utilitarian rule} which returns, in each issue, the candidate with highest support (assuming some fixed tiebreaking among candidates). Note that this rule completely ignores the outcome of previous issues. We now introduce the two main classes of voting rules that we study, both of which include the utilitarian rule as a special case. 

\subsection{OWA rules}

OWA voting rules for multi-issue domains were proposed by \citet{amanatidis2015multiple} and are based on \emph{ordered weighted averaging} operators \cite{yager1988ordered}. We study them mostly in their sequential formulation. 
An OWA voting rule is defined by an family of vectors $\{\alpha^n\}_{n=1}^\infty$, where
each $\alpha^n=(\alpha_1,\dots,\alpha_n)$ has length $n$ and satisfies $\alpha_1>0$ and $\alpha_j \geq 0$ for $j\in[n]$.
Given an election with $n$ voters, the score of a (possibly partial) outcome $\bar w$ subject to $\alpha^n$ is 
\[\owascore_{\alpha^n}(\bar w)=\alpha^n\cdot s(\bar w),\]
where $\cdot$ is the scalar (dot) product.
In issue $i$, given a partial outcome $\bar w$ for the first $i-1$ issues, the OWA rule returns the candidate $c\in C_i$ with maximum $\owascore_{\alpha^n}((\bar w,\, c))$ (where $(\bar w,\,c)$ is the $\ell$-tuple formed by appending $c$ to $\bar w$).
If more than one candidate achieves the maximum score, we use a fixed tiebreaking order among outcomes.
We typically omit the superscript of $\alpha^n$, as $n$ is clear from the context.
Note that the utilitarian rule corresponds to $\alpha^n=(\nicefrac 1 n,\dots,\nicefrac 1 n)$. Another notable rule is the \emph{egalitarian rule}, which corresponds to $\alpha^n=(1,0,\dots,0)$. This rule selects, at each step, the candidate that maximizes the satisfaction of the least-satisfied voter. Note that there are typically many such candidates, and thus this rule is rather indecisive. One can thus consider the \emph{leximin} rule, which -- among outcomes that are optimal for the least satisfied voter -- maximizes the satisfaction of the second-least satisfied voter, then the third-least, etc.

\begin{example}\label{ex:running}
Consider an election with 100 voters and 4 issues with the same three candidates, $\{a,b,c\}$. There are 66 voters that approve $\{a\}$ in all issues, 33 voters that approve $\{b\}$ in all issues, and one voter approves always $\{c\}$.
The utilitarian rule selects the outcome $\bar w_1=(a,a,a,a)$ as $a$ is the most approved candidate in each issue, achieving a total satisfaction of $\sum_{v\in N}\sat(v, \bar w_1)=4\cdot 66$.
The leximin rule selects $\bar w_2=(a,b,c,a)$. Indeed, in the first round, all alternatives give a satisfaction of $0$ to (at least) $34$ voters; alternative $a$ is the alternative that maximizes the satisfaction of the $35$th least-satisfied voter. In the second round, all alternatives give satisfaction $0$ to the least-satisfied voter; alternative $b$ is the only alternative where the second least-satisfied voter has satisfaction of $1$, and is hence selected. In the third round, $c$ is the only alternative giving a satisfaction of at least $1$ to everyone. The last round selects $a$, similarly as the first round. On the other hand, the egalitarian rule is highly indecisive in this election, as many different outcomes are optimal; depending on the tiebreaking rule, it can return, e.g., $\bar w_1=(a,a,a,a)$, $\bar w_3=(b,b,b,b)$ or even the questionable $\bar w_4=(c,c,c,c)$.
\end{example}

More precisely, the leximin rule is based on the leximin ordering $\succ$. Given two outcomes $\bar w$ and $\bar w'$ with $s(\bar w)=(s_1,\dots,s_n)$ and $s(\bar w')=(s_1',\dots,s_n')$, $\bar w \succ \bar w'$ if there exists an index $j\in [n]$ such that $s_1=s_1'$, $\dots$, $s_{j-1}=s_{j-1}'$ and $s_{j} > s_{j}'$.
At issue $i$, the \emph{leximin rule} returns a candidate $c\in C_i$ such that $(\bar w,\,c)$ is maximal w.r.t. $\succ$. One can show that this rule corresponds to the OWA rule defined by the vector $\alpha^n=(1,\nicefrac{1}{kn},\nicefrac{1}{k^2n^2},\dots)$.

\newcommand{\leximinOWA}{The OWA rule defined by $\alpha=(1,\frac{1}{kn}, \frac{1}{k^2n^2}, \dots)$ is equivalent to the leximin rule.}
\begin{proposition}
\leximinOWA\label{prop:leximinOWA}
\end{proposition}

\begin{proof}
Assume that $\bar w \succ \bar w'$, i.e., for $s(\bar w)=(s_1,\dots,s_n)$ and $s(\bar w')=(s_1',\dots,s_n')$, there exists an index $j\in [n]$ such that $s_1=s_1'$, $\dots$, $s_{j-1}=s_{j-1}'$ and $s_{j} > s_{j}'$. Then 
\begin{multline*}
\owascore_\alpha(\bar w)- \owascore_\alpha(\bar w') =
\alpha\cdot s(\bar w) - \alpha\cdot s(\bar w') = \\
(\underbrace{s_{j} - s'_{j}}_{\geq 1}) \cdot \frac{1}{(kn)^{j-1}} + \sum_{\ell=j+1}^n (\underbrace{s_{\ell} - s'_{\ell}}_{\geq -k}) \cdot \frac{1}{(kn)^{\ell-1}} \geq\frac{1}{(kn)^{j-1}} - k\sum_{\ell=j+1}^n  \frac{1}{(kn)^{\ell-1}} \geq \\
\frac{1}{(kn)^{j-1}} - k (n-1) \frac{1}{(kn)^{j}} >
0.
\end{multline*}
This argument is symmetric in $\bar w$ and $\bar w'$, so we have shown that $\bar w \succ \bar w'$ iff $\owascore_\alpha(\bar w)- \owascore_\alpha(\bar w')>0$. Thus, a maximal element with respect to $\succ$ achieves a maximum $\owascore_\alpha$-score and vice versa.
%
\end{proof}

\subsection{Thiele rules}

The second class we study is based on Thiele methods (introduced by \citet{Thie95a}, see the book by \citet{abcsurvey}).
While Thiele methods are a class of multi-winner voting rules, they can be adapted to our setting straightforwardly.
A voting rule in the Thiele class is defined by a function~$f\colon \mathbb N \to \mathbb{R}_{\geq 0}$ satisfying $f(1)>0$ and $f(\ell)\geq f(\ell+1)$ for all $\ell\in \mathbb N$.
The $f$-Thiele rule assigns a score of 
\[\thielescore_f(\bar w) = \sum_{v\in N} \sum_{\ell=1}^{\sat(v, \bar w)} f(\ell)\]
to a (partial) outcome $\bar w$.
In each issue $i$, given a partial outcome $\bar w$, the rule returns the candidate $c\in C_i$ with maximum $\thielescore_f((\bar w,\, c))$.
Intuitively, these are weighted approval rules for which the
weight assigned to each voter only depends on her satisfaction.
Note that the utilitarian rule corresponds to $f_{\mathit{util}}(\ell)=1$.
The egalitarian rule does not appear in this class, whereas, for fixed $n$ and $k$, leximin is equivalent to the Thiele method with
$f_{\mathit{lex}}(\ell)=\nicefrac{1}{(kn)^{\ell-1}}$ (among other quickly decreasing functions).
The most important Thiele rule is $f(\ell)=\nicefrac 1 \ell$, which is called \emph{Proportional Approval Voting} in the multi-winner setting. We also refer to this Thiele rule as PAV.

\begin{example}
Continuing Example~\ref{ex:running}, we see that for PAV in the first round $a$ has maximum score of $66$. In the second round, both partial outcomes $(a,\,a)$ and $(a,\,b)$ achieve the same score of $66+33=99$. The final outcome is either $\bar w_5=(a,b,a,a)$ or $\bar w_6=(a,a,b,a)$, depending on tiebreaking.
Note that PAV is more majoritarian than leximin as it essentially ignores the single voter approving $\{c\}$.
\end{example}

\subsection{Global-Optimization-Based Rules}

In Section~\ref{sec:global}, we will consider the setting where voters report their preferences over the issues at the same time, and all issues are decided simultaneously. In this setting, we look at voting rules based on a global optimization objective. That is, they select some outcome maximizing some score function. Given a sequential rule $\calR$, we refer to its optimization-based counterpart (computed as explained below) as opt-$\calR$.

The optimization-based rule corresponding to the OWA rule $\alpha^n$ simply selects the complete outcome maximizing $\owascore_{\alpha^n}(\bar w)$, whereas the rule corresponding to the Thiele rule $f$ selects the outcome maximizing $\thielescore_f(\bar w)$. Again, we assume a fixed tiebreaking order over the outcomes. For example, opt-egalitarian selects an outcome $\bar w$ maximizing $\min_{v\in N}\sat(v, \bar w)$. This rule is \NP-hard to compute \cite{amanatidis2015multiple}. In contrast, the opt-utilitarian rule is identical to its sequential formulation, and thus is poly-time computable. In this setting, we will often consider opt-leximin, which again is part of the opt-OWA family with $\alpha^n=(1,\nicefrac{1}{kn}, \nicefrac{1}{k^2n^2}, \dots)$ (see Proposition~\ref{prop:leximinOWA}). This rule selects any outcome $\bar w$ that is maximal w.r.t. the leximin ordering $\succ$, as previously defined.

We note that in this paper we focus on the sequential case, unless explicitly noted.

\subsection{Free-Riding}

In this paper, we study a specific form of strategic manipulation
called \emph{free-riding}. Intuitively, this means that a voter misrepresents
her preferences on an issue where her favorite candidate wins also without her 
support.
If the used voting rule takes the satisfaction of voters into account (as most OWA and Thiele methods do), such a manipulation can increase the voter's influence on other issues. 

\begin{example}
Consider an election with three voters and two issues.
The first issue is uncontroversial: all voters approve candidate~$a$.
The second issue is highly controversial: all voters approve different candidates ($A_2(1)=\{x\}$, $A_2(2)=\{y\}$, $A_2(3)=\{z\}$). If PAV (with alphanumeric tiebreaking) is used to determine the outcome,
it could select, e.g., the outcome $(a,\,x)$.
This leaves voters~$2$ and~$3$ less satisfied than voter~$1$. 
Both of them could free-ride to improve their satisfaction. Consider voter~$2$.
If voter~$2$ changes her ballot on the first issue to another candidate, the outcome changes to
$(a,\,y)$, as, according to the ballots, $a$ yields a score of $2$ in the first round, and $(a,\,y)$ of $3$
in the second round, which is optimal.
As voter~$2$'s true preferences are positive towards~$a$, this manipulation was successful.
\end{example}

In the following,
given an election $\elec$ and a rule $\calR$ such that $\calR(\elec)=(w_1,\dots,w_k)$,
we indicate $w_i$ as $\calR(\elec)_i$.

\begin{definition}\label{def:freeriding}
Consider an election $\elec = \langle (C_i)_{i=1}^k, N, (A_i)_{i=1}^k \rangle$, a voter $v \in N$
and a voting rule $\calR$.
Let $\calR(\elec)=(w_1,\dots,w_k)$.
We say that voter $v$ can free-ride in election $\elec$ on issues $I\subseteq[k]$ if there
exists another election $\elec^*=\langle (C_i)_{i=1}^k, N, (A^*_i)_{i=1}^k\rangle $ that
only differs from $\elec$ in the approvals of $v$ for issues in $I$ such that,
for all $i\in I$,
$w_i\in A_i(v)$, $w_i\not\in A_i^*(v)$ and $\calR(\elec^*)_i=w_i$.
In this case, we also say that $v$ can free-ride in $\elec$ via $\elec^*$.
\end{definition}

Usually, we say a voter can manipulate if she can achieve a higher satisfaction
by misrepresenting her preferences. In contrast, Definition~\ref{def:freeriding} makes no 
assumptions about the satisfaction of the free-riding voter. Instead, we only require
that the manipulator can misrepresent her preference in an issue without changing the 
outcome of the issue. This might lead to the same, a higher or lower satisfaction
for the manipulator. This distinction will be crucial when talking about the risk of
free-riding.

We say that \emph{a voting rule $\calR$ can be manipulated by free-riding} if there exists an election $\elec$,
a voter $v$ and an election $\elec^*$ such that $v$ can perform free-riding in $\elec$ via $\elec^*$
and $\sat_{\elec}(v, \calR(\elec)) < \sat_{\elec}(v, \calR(\elec^*))$.

Finally, note that sometimes we will consider a more general notion of free-riding, called \emph{generalized} free-riding, where we lift the constraint that the outcome on the issues where free-riding occurs remains exactly the same. 
The two notions are rather similar, and most of our results hold for both models, but we are able to get stronger results for the generalized case. This only plays a role in Section~\ref{sec:complexity}, where we study computational complexity. We defer its formal definition to that section.

\section{Possibility and Risk of Free-Riding}\label{sec:strategyproof}

In this section, we ask two main questions:
The first question is for which voting rules free-riding is actually possible.
We approach this
question from two angles: First, we identify a small set of mild axioms and show
that any rule satisfying these axioms is susceptible to free-riding.
Secondly, we focus on the families of sequential rules we introduced in Section~\ref{sec:model} and
show that free-riding is possible for essentially all rules in these families.
These two approaches complement each other, as they do not cover the same set of voting rules.

The second question is under which conditions free-riding may lead to a \emph{decrease in satisfaction} of the free-riding voter.
We analyze for which rules free-riding entails such a risk.

\subsection{Possibility of Free-Riding}\label{sec:axiomatic}

We first identify a multi-issue voting rule where free-riding is not possible: the utilitarian rule.
Observe that with the utilitarian rule the outcomes for different issues do not influence each other.
Thus, artificially decreasing one's satisfaction brings no advantage with
this rule. Therefore, the utilitarian rule cannot be manipulated by free-riding.

\begin{proposition}\label{prop:util-cannot-manipulated}
The utilitarian rule cannot be manipulated by free-riding.
\end{proposition}

Next, we show that for essentially all other reasonable rules, free-riding is possible. In order to formulate axioms that are as mild as possible, we will focus on a class of elections with a 
particularly simple structure. We claim that in such elections it is easy to argue about how
a rule should behave.
\begin{definition}[Simple elections]
For every $k\geq 1$, $\ell\in[k]$ and $n\geq 2$, an \emph{$\ell$-simple $(k,n)$-election} satisfies the following:
  \begin{itemize}
    \item there are $k$ issues and $n$ voters;
    \item in each issue there are two candidates: $C_1=\dots=C_k=\{a,b\}$;
    \item every voter~$v<n$ approves $\{a\}$ in every issue;
    \item voter~$n$ approves $\{b\}$ in $\ell$ issues, and $\{a\}$ in the remaining $k-\ell$.
  \end{itemize}
\end{definition}
An example $2$-simple $(6,4)$-election is
\begin{align*}
A_1 &= (\{a\},\,\{a\},\,\{a\},\,\{a\},\,\{a\},\,\{a\}),\\
A_2 &= (\{a\},\,\{a\},\,\{a\},\,\{a\},\,\{a\},\,\{a\}),\\
A_3 &= (\{a\},\,\{a\},\,\{a\},\,\{a\},\,\{a\},\,\{a\}),\\
A_4 &= (\{b\},\,\{a\},\,\{a\},\,\{a\},\,\{b\},\,\{a\}).
\end{align*}
Note that for every number of issues $k\in\mathbb{N}$ and of voters $n\geq3$, there is exactly one $k$-simple $(k,n)$-election (namely, the election where voter~$n$ always approves only $b$) and exactly $k$ different $1$-simple $(k,n)$-elections. 

Given this definition, we introduce the following axioms for multi-issue voting rules.\footnote{Observe that all the axioms we propose apply to both sequential and optimization-based rules.}

\begin{itemize}
  \item \textbf{Near-unanimity}: For $n\geq 3$, the complete outcome of every $1$-simple $(k,n)$-election must be $(a,\,\ldots,\,a)$.
  \item \textbf{Incentive for minorities}: For $n\geq3$, there must be a number of issues $k$ such that the complete outcome of the $k$-simple $(k,n)$-election is a permutation of $(a,\,\ldots,\,a,\,b)$.
\end{itemize}

The two axioms are logically independent; the utilitarian rule satisfies near-unanimity but not incentive for minorities,
and vice versa for the constant rule always selecting $(b,\,a,\,a,\,\ldots)$. Both rules are also immune to free-riding.

Note that a weaker formulation of incentive for minorities could require for $b$ to win \emph{at least} once instead of \emph{exactly} once. We have two arguments in favor of our stronger formulation. First, every sequential rule that satisfies the weaker variant of this axiom must also satisfy our stronger notion. Secondly, one can interpret the additional requirement of the stronger notion (where $b$ must win exactly once, instead of at least once) as a continuity condition: if given $k-1$ issues voter~$n$ does not deserve any representation (satisfaction~0), then they should not suddenly deserve to decide more than one issue when there is exactly one additional issue. Nevertheless, one could adopt this weaker notion instead and get results analogous to what we obtain in this section by slightly strengthening the other axioms. We do so in Appendix~\ref{app:possibility}.

First, we show that any rule that satisfies these two axioms is susceptible to free-riding.

\begin{theorem}
  Any rule that satisfies near-unanimity and incentive for minorities can be manipulated by free-riding.\label{thm:many-issues}
\end{theorem}

\begin{proof}
  Consider a rule that satisfies the conditions of the theorem. Consider any $n\geq3$ and let $k$ be the number of issues whose existence is prescribed by incentive for minorities. Focus on the $k$-simple $(k,n)$-election $\elec$ and let $i$ be the unique issue where $b$ is the outcome for $\elec$. Consider now the $1$-simple $(k,n)$-election $\elec^\prime$ where voter~$n$ approves of $b$ in issue $i$. By near-unanimity, the outcome must be $a$ in all issues. But then voter~$n$ can free-ride in $\elec^\prime$ via $\elec$ in issues $[k]\setminus\{i\}$, completing the proof.
\end{proof}

Essentially, this result shows that as soon as we want to give \emph{some} (but not all) power to minorities, free-riding becomes unavoidable.

Observe that the proof of Theorem~\ref{thm:many-issues} requires the manipulator to free-ride in multiple issues. If we additionally enforce the following axioms, then free-riding can occur even if we restrict the manipulator to only free-ride in one issue.

\begin{itemize}
  \item \textbf{Issue-wise unanimity}: For every $\ell$-simple $(k,n)$-election, if in some issue all voters approve of $a$, then $a$ must be selected in that issue.
  \item \textbf{Monotonicity}: If the complete outcome of the $k$-simple $(k,n)$-election contains $b$ exactly once, then for all $\ell\leq k$, the complete outcome of any $\ell$-simple $(k,n)$-election $b$ at most once.
\end{itemize}

Note that adding issue-wise unanimity is only a mild requirement, here, as its core idea is already embodied by near-unanimity (although the two are logically independent).

\begin{theorem}
  Any rule that satisfies near-unanimity, incentive for minorities, issue-wise unanimity, and monotonicity can be manipulated by free-riding in one issue.\label{thm:one-issue}
\end{theorem}

\begin{proof}
  Consider a rule that satisfies the conditions of the theorem. Consider any $n\geq3$ and let $k$ be the number of issues whose existence is prescribed by incentive for minorities. Let $\ell$ be the minimal $\ell$ for which there is an $\ell$-simple $(k,n)$-election where $b$ wins exactly once, and let $\elec$ be such an election. Observe that such an $\ell$ must exist (by incentive for minorities), that we must have $\ell>1$ (by near-unanimity), and that for all $(\ell-1)$-simple $(k,n)$-elections $b$ is not in the outcome (by monotonicity). Let $i$ be the issue where $b$ wins in $\elec$; note that in this issue voter~$n$ must approve of $b$ (by issue-wise unanimity). Next, let $j$ be some issue distinct from $i$ where voter~$n$ approves of $b$ (such a $j$ exists as $\ell>1$). Let $\elec^\prime$ be the $(\ell-1)$-simple $(k,n)$-election obtained by letting voter~$n$ approve of $a$ instead of $b$ in issue $j$. Observe that voter~$n$ can free-ride in $\elec^\prime$ via $\elec$ in issue $j$, completing the proof.
\end{proof}

Observe that the previous results do not fully cover the families of rules we have introduced. For example,
the egalitarian rule does not satisfy near-unanimity (the outcome might depend on the tiebreaking).
Clearly, one might be concerned by what failing such mild axioms means for a rule; nonetheless, we now
present a complementary result that shows that the whole set of families we have introduced can be
manipulated via free-riding (excluding the utilitarian rule, as per Proposition~\ref{prop:util-cannot-manipulated}).

\newcommand{\corrfreeriding}{Every sequential Thiele and sequential OWA rule 
except the utilitarian rule can be manipulated by free-riding.}
\begin{theorem}\label{thm:freeriding}
\corrfreeriding
\end{theorem}

\begin{proof}
First, let $\calR$ be a sequential $f$-Thiele Rule different from the utilitarian rule.
Then, there exists a $k$ such that $f(k-1)> f(k)$. 
Consider a $k+1$ issue election with four voters and two candidates $a$ and $b$ 
such that for the first $k$ issues all voters only approve candidate $a$.
Moreover, on issue $k+1$ voters $1$ and $2$ approve $b$ while voter $3$ and $4$
approve $a$. Assume further that $a$ is preferred to $b$ in the tiebreaking order.
Clearly, $a$ wins in the first $k$ issues. Hence in issue $k+1$ all voters have weight $f(k)$
which means both candidates have a score of $2f(k)$. By tiebreaking $a$ wins.
We claim that voter $1$ can manipulate by changing her vote in one of the first $k$ issues
to $\{b\}$. Let $i$ be the issue on which $1$ manipulates.  Then, in issue $i$,
candidate $a$ has a score of $3f(i-1)$ while $b$ has a score of $f(i-1)$.
Now, $f(i-1) > f(k)$ implies that $f(i-1) >0$. Therefore $3f(i-1) > f(i-1)$,
which means $a$ still wins in issue $i$. It is clear that $a$ also wins in the other
issues until $k+1$. In issue $k+1$, $a$ has a score of $2f(k)$ while $b$ has a score of
$f(k) + f(k-1)$. By assumption, this means that $b$ wins on issue $k+1$. Therefore,
voter $1$ did free-ride successfully.

Next, let $\calR$ be a sequential OWA-Rule that is not the utilitarian rule. 
Consider an election with $2$ issues and $k$ voters. In each issue there are $k$ candidates
$a_1, \dots a_k$. In the first issue, voters $1$ and $2$ approve $a_1$. Every other voter $i\in\{3,\dots, k\}$
approves $a_i$. In the second issue voter $1$ approves $a_1$, voter $2$ approves $a_2$ and
all other voters approve both $a_1$ and $a_2$.
We assume that candidates with a lower index are preferred by the tiebreaking,
which is applied lexicographically. Then there
exists a $k$ for which the vector $\alpha$ for $k$ voters satisfies $\alpha_{1} > \alpha_{k}$. 
In the first issue $a_1$ has to be selected, as no other candidate can have a higher OWA score.
Then, as before $(a_1,a_1)$ and $(a_1,a_2)$ lead to the highest possible
score on when looking at the second issue. By tiebreaking, $(a_1,a_1)$ wins.
Now, we claim that voter $2$ can free-ride by approving $a_2$ instead of $a_1$ in the first issue.
After the free-riding, all candidates are tied for the first issue, hence $a_1$ wins by tiebreaking.
However, then following the discussion above, $a_2$ needs to be selected in the second
issue. It follows that voter~$2$ did successfully free-ride.
\end{proof}

Hence, free-riding is essentially unavoidable if we want to guarantee fairer outcomes
using Thiele or OWA rules.

\subsection{The Risk of Free-Riding}

Intuitively, free-riding seems to offer a simple and risk-free way to 
manipulate. However, we observe that for most sequential voting rules, free-riding may lead to a decrease in satisfaction.
First, we can show that this holds for all sequential Thiele rules, except the utilitarian rule.

\begin{proposition}
Let $f:\mathbb{N} \to \mathbb{R}_{\geq 0}$ be a function for which there is
an $i \in \mathbb{N}$ such that $f(i) > f(i+1)$. Then, under the sequential $f$-Thiele rule,
free-riding can reduce the satisfaction of the free-riding voter.
\end{proposition}

\begin{proof}

Consider a sequential $f$-Thiele rule such that $f(i) > f(i+1)$ and consider the following
election with nine voters, $i+4$ issues and candidates $a,\dots, g$ for 
all issues. We assume alphabetic tiebreaking.
The approvals for each issue are given by these tuples:
\begin{align*}
A_1 = \dots = A_{i-1} &= (\{a\},\{a\},\{a\},\{a\},\{a\},\{a\},\{a\},\{a\},\{a\})\\
A_{i} &= (\{a\},\{a\},\{a\},\{b\},\{b\},\{c\},\{d\},\{e\},\{f\})\\
A_{i+1} &= (\{b\},\{a\},\{c\},\{b\},\{b\},\{a\},\{a\},\{a\},\{b\})\\
A_{i+2} &= (\{b\},\{a\},\{c\},\{b\},\{e\},\{a\},\{f\},\{a,b\},\{g\} )\\
A_{i+3} &= (\{b\},\{c\},\{d\},\{e\},\{b\},\{f\},\{a\},\{a\}, \{g\})
\end{align*}
Then, $\{a\}$ is clearly the winner for the first $i-1$ issues. 
Thus, all voters have a satisfaction of $i-1$ before issue $i$ and 
$a$ wins on issue $i$ as it has the most supporters.
In issue $i+1$, $a$ and $b$ increase the Thiele score by $f(i+1) + 3f(i)$,
while $c$ increases the score only by $f(i+1)$. By tiebreaking, $a$ wins again.
Then, in issue $i+2$, $a$ increases the score by $f(i+2) + 2f(i+1)$, $b$ by $f(i) + 2f(i+1)$,
while all other candidates by at most $f(i)$.
As $f(i) > f(i+1) \geq f(i+2)$ (together with the tiebreaking rule if $f(i+1)=0$), it follows that
$b$ wins in issue $i+2$. Finally, in issue $i+3$, $a$ increases the score by $f(i+1) + f(i+2)$, 
$b$ by $f(i) + f(i+2)$ and all other candidates increase the score by at most $f(i)$.
Hence $b$ wins.

Now assume voter~$1$ changes her preferences and free-rides in issue $i$.
It is straightforward to check that $a$ remains the winner for issue $i$,
but winner in issue $i+1$ changes to $b$ while $a$ now wins for issue $i+2$ and
$i+3$.
Therefore, $1$ now additionally approves of the winner on issue $i+1$
but does not approve the winners of issues $i+2$ and $i+3$ any more. Hence, free-riding led to
a lower satisfaction for the free-riding voter.
\end{proof}

\newcommand{\OWAfreeridingdecrease}{Consider a sequential $\alpha$-OWA rule such that there exists an $n\geq 8$ for which $\alpha^n$ is nonincreasing and satisfies $\alpha_{3} > \alpha_{n-2}$. Then, free-riding can reduce the satisfaction of the free-riding voter.}

The same holds for the following large class of sequential OWA rules. 

\begin{proposition}\label{OWAfreeridingdecrease}
\OWAfreeridingdecrease
\end{proposition}

\begin{proof}
Consider an election with $4$ issues and $n$ voters. In each issue the candidate set is a subset of
$\{a_1, \dots a_n\}$. The specific set of candidates is defined as all candidates that receive
at least one approval according to the following description:
\begin{description}
\item[Issue 1:] Voter $3,\dots n-3$ and $n$ approve $a_n$. Every other voter $i$ 
approves~$a_i$.
\item[Issue 2:] Voter $1,2,3$ approve $a_n$, voters $n-2,n-1,n$ approve $a_1$.
Every other voter $i$ approves~$a_i$.
\item[Issue 3:] Voter $1$ and voter $4$ approve $a_4$, voter $n-1$ and $n$ approve $a_{n}$, 
Every other voter $i$ approves~$a_i$.
\item[Issue 4:] Voter $2$ and $3$ approve $a_2$, voter $n-2$ and $n$ 
approve $a_n$ and every other voter $i$ approves $a_i$.
\end{description}
We assume that candidates with a higher index are preferred by the tiebreaking.

Let us determine the outcome of this election.
In the first issue $a_n$ has to be selected, as no other candidate can have a higher OWA score.
This leads to satisfaction vector $(0,0,0,0,1,\ldots,1)$.
In the second issue selecting $a_1$ or $a_n$ both lead to a satisfaction vector of $(0,0,1,\ldots,1,2)$.
Selecting any other candidate leads to a satisfaction vector of $(0,0,0,0,1,\ldots,1,2)$.
Clearly, it is again the case that no candidate can have a higher OWA score than $a_n$.
Hence $a_n$ wins again.

In the third issue, selecting any candidate other than $a_4$, $a_3$, $a_{n-2}$ or $a_{n}$
leads to a satisfaction vector of $(0, 0,1,\ldots,1,2,2)$.
Selecting $a_3$ leads to a satisfaction vector of $(0,0,1,\ldots,1,3)$, while selecting $a_4$
leads to a satisfaction vector of $(0,0,1,\ldots,1,2,2,2)$ and $a_{n-2}$ to $(0,1,\ldots,1,2)$.
Finally, selecting $a_n$ leads to a vector of $(0,1,\ldots,1,2,2)$.
As $\alpha_{2} \geq \alpha_{3} > \alpha_{n-2} \geq \alpha_{n}$ selecting $a_{n}$ leads
to a higher OWA-score than selecting a candidate $a_i$ with $i < n-2$.
Moreover, the OWA-score of $a_{n-2}$ cannot be higher than the score of $a_{n}$.
Hence, $a_{n}$ wins in issue three.

In the fourth issue, selecting any candidate other than $a_2$ or $a_n$
leads to a satisfaction vector of $(0,1,\ldots,1,2,2,2)$.
Selecting $a_2$ leads to a satisfaction vector of  
$(0,1,\ldots,1,2,2,3)$,
while selecting $a_n$ leads to a satisfaction vector of $(1,\ldots,1,2,3)$.
As $\alpha_{1} \geq \alpha_{3} > \alpha_{n-2}$ selecting $a_{n}$ leads
to the highest OWA score.

Now, we claim that voter $n$ can free-ride by approving any other candidate in the first issue.
Indeed, if voter $n$ approves any other candidate it is still the case that no candidate
can have a higher OWA score than $a_n$. Let us consider how the other issues change:

In the second issue selecting $a_1$ leads to a satisfaction vector of $(0,0,1,\ldots,1)$.
Selecting $a_n$ leads to a satisfaction vector of $(0, 0,0,1,\ldots,1,2)$, while
selecting any other candidate leads to a satisfaction vector of $(0,0,0,0,0,1,\ldots,1,2)$.
As $\alpha_{3} > \alpha_{n}$ we know that the OWA score of $a_1$ is higher than that of $a_n$
which is at least as high as the OWA score of every other candidate.

In the third issue, selecting any candidate other than $a_4$, $a_2$ or $a_{n}$
leads to a satisfaction vector of $(0, 0,1,\ldots,1,2)$.
Selecting $a_2$ leads to a satisfaction vector of  $(0,1,\ldots,1)$,
while selecting $a_4$ leads to a satisfaction vector of $(0,1,\ldots,1,2)$.
Finally, selecting $a_n$ leads to a vector of $(0,0,1,\ldots,1,2,2)$.
As $\alpha_{2} \geq \alpha_{3} > \alpha_{n-1}$ selecting $a_{4}$ leads
to a higher OWA-score than selecting a candidate $a_i$ with $i \neq 2,4$.
Moreover, the OWA-score of $a_{2}$ cannot be higher than the score of $a_{4}$.
Hence, $a_{4}$ wins in issue three.

In the fourth issue, selecting any candidate other than $a_2$, $a_4$, or $a_n$
leads to a satisfaction vector of
$(0,1,\ldots,1,2,2)$.
Selecting $a_2$ leads to a satisfaction vector of  
$(1,\ldots,1,2,2)$,
while selecting $a_4$ leads to a satisfaction vector of  
$(0,1,\ldots,1,3)$
and $a_n$ to a satisfaction vector of $(0,1,\ldots,1,2,2,2)$.
As $\alpha_{1} > \alpha_{n-2}$ selecting $a_{2}$ leads
to the highest OWA score.
However,  this decreases the satisfaction of voter $n$ with 
respect to the honest ballots to $2$.
\end{proof}

\newcommand{\EGALfreeridingdecrease}{Free-riding can decrease the satisfaction of the free-riding voter under the sequential egalitarian rule.}

Note that the conditions of Proposition~\ref{OWAfreeridingdecrease} do not include the sequential egalitarian rule, but an equivalent statement holds nonetheless:

\begin{proposition}\label{EGALfreeridingdecrease}
\EGALfreeridingdecrease
\end{proposition}

\begin{proof}
Consider an election with $5$ issues and $5$ voters. In each issue there are $3$ candidates
$a$, $b$ and $c$. We assume that tiebreaking always prefers $a$.
The approval sets are given as follows:

\begin{center}
\begin{tabular}{cccccc}
              & Issue 1   & Issue 2     & Issue 3 & Issue 4  & Issue 5 \\
\midrule
Voter 1       & $\{b\}$   &    $\{a\}$  & $\{b\}$ &  $\{b\}$ &   $\{b\}$\\    
Voter 2       & $\{b\}$   &    $\{a\}$  & $\{a\}$ &  $\{b\}$ &   $\{a\}$\\    
Voter 3       & $\{b\}$   &    $\{a\}$  & $\{a\}$ &  $\{c\}$ &   $\{b\}$\\    
Voter 4       & $\{a\}$   &    $\{a\}$  & $\{b\}$ &  $\{a\}$ &   $\{b\}$\\    
Voter 5       & $\{a\}$   &    $\{a\}$  & $\{b\}$ &  $\{a\}$ &   $\{b\}$\\    
\end{tabular}
\end{center}

Let us determine the winners under the sequential egalitarian rule:
In the first issue, every option leads to a minimal satisfaction of $0$. Therefore, $a$ is winning
by tiebreaking. Then, in the second issue, $a$ must be winning as it raises the minimal satisfaction
to $1$. In the third issue, again, no alternative can increase
the minimal satisfaction and therefore $a$ wins by tiebreaking. This leads to a situation
where every voter except $1$ has a satisfaction of $2$ while $1$ has a satisfaction of $1$.
Hence, in issue four, $b$ must be the winner, as it increases the minimal satisfaction to $2$.
Finally, in the fifth issue, electing $b$ leads to a minimal satisfaction of $3$, which is better
than electing $a$. We observe that voter $1$ has a satisfaction of $3$ in the end.

Now, we claim that voter $1$ can free-ride on issue two. If voter $1$ approves $b$ instead,
then all candidates lead to the same minimal satisfaction of $0$. Hence, $a$ wins by tiebreaking.
If voter $1$ decides to free-ride on issue two, this changes the winners of the following issues
as follows: In issue three, $b$ must now be the winner, as it increases the minimal satisfaction
to $1$. Then, in issue four and five, no candidate increases the minimal satisfaction and hence $a$
wins both issues by tiebreaking. However, this decreases the satisfaction of voter $1$ with 
respect to the honest ballots, to $2$.
\end{proof}

\section{Computational Complexity}\label{sec:complexity}

\newcommand{\freeridingProb}{\textsc{Successful-$\calR$-Free-Riding}}
\newcommand{\genFreeridingProb}{\textsc{Generalized-Successful-$\calR$-Free-Riding}}

In this section, we will study the computational complexity of successful free-riding:
that is, raising one's total satisfaction via free-riding. Overall,
we will show that it is generally hard to do so, even with full information. 
Observe that, due to the performance of, e.g., modern \textsc{SAT}- or \textsc{ILP}-solvers,
computational hardness (in particular \NP-completeness) cannot be seen as 
an unbreakable defense against manipulation. However, the main appeal of free-riding 
is its simplicity. A manipulator that 
is able to solve computationally hard problems (and has full information about other's preferences) has no benefit from restricting
the potential manipulation to free-riding -- such a voter could optimize their satisfaction via arbitrary manipulation.

Observe that 
the outcome of a sequential rule is always
polynomial-time computable: for every round, we can just iterate over all the
candidates involved in that issue and pick the one maximizing the score.

Therefore, it is computationally easy to decide if (some) way of free-riding
is possible: for each issue where a voter approves of the winner $c$, one can
just iterate over all singleton ballots $\{c^\prime\}$ with $c^\prime\neq c$ and check if
under any such ballot $c$ still wins.\footnote{Note that restricting
oneself to singleton ballots is safe under OWA and Thiele rules:
it is easy to see that, if there exists a ballot $A^*_i(v)$ via which
voter~$v$ can free-ride in issue $i$, then she can free-ride in the same issue also via
any ballot in the set $\{\{c^\prime\}\colon c^\prime\in A^*_i(v)\}$.}

However, although voters can easily verify if free-riding is possible,
it might be hard to judge its long-term consequences. If this
is unfeasible, voters might be discouraged from free-riding
(because, as we have shown, it can have negative consequences). Hence,
we study the following problem:


\begin{problem}
  \problemtitle{\freeridingProb}
  \probleminput{An election $\elec=\langle N, \bar A,\bar C\rangle$
  and a voter~$v\in N$.}
  \problemquestion{Is there an election $\elec^*$ such that $v$ can free-ride in $\elec$
  via $\elec^*$ and $\sat_{\elec}(v,\calR(\elec))<\sat_{\elec}(v,\calR(\elec^*))$?}
\end{problem}

First, we show that free-riding is \NP-complete for a large class of
sequential $f$-Thiele rules and the sequential egalitarian rule. 

\newcommand{\seqFreeridingThiele}{\freeridingProb{} is \NP-complete for every sequential $f$-Thiele rule for which there exists a $\ell\in\mathbb N$ such that (i) for all $j,j^\prime\in[\ell]$ it holds $f(j)=f(j^\prime)$ and (ii) $f$ is strictly decreasing on $\mathbb{N}\setminus[\ell-1]$.}

\begin{theorem}\label{thm:seqFreeridingThiele}
\seqFreeridingThiele
\end{theorem}

The conditions of this theorem apply to all functions that are constant up to a certain number $\ell$, and from $\ell$ on become strictly decreasing. 
This is the case, e.g., for the sequential PAV rule.

\begin{proof}[Proof of Theorem~\ref{thm:seqFreeridingThiele}]
Fix an $f$-Thiele rule $\calR$ satisfying the conditions of the theorem. First, notice that \freeridingProb{} is in \NP,
as, given an insincere approval ballot for the manipulator, we can check whether it 
improves her satisfaction in polynomial time and whether it is a case of free-riding.

Now, we show hardness by a reduction from \textsc{3-SAT} \cite{garey1979computers}.
Let $\phi$ be a 3-CNF with $n$ variables and $m$ clauses. We refer to the $j$-th clause as $C_j$.
We assume w.l.o.g.\ that $\phi$ is not satisfied
by setting all variables to false and that each clause contains 
exactly three literals.
Furthermore, let $k \in \mathbb N$ be the smallest $k$ such that $f(k+1)<f(k)$.
As $\calR$ is not the utilitarian rule, such a $k$ surely exists, and is constant in the size of $\phi$ 
(as $f$ does not depend on the number of either the voters or the issues).
Moreover, let $\ell \in \mathbb N$ be the smallest $\ell$ such that $f(k+1)(\ell+1)<f(k)\ell$. 
Again, note that $\ell$ can be large, but does not depend on the instance.
We will construct an instance of \freeridingProb{} with $3n(\ell+1)+5$ voters and $k+3n$ rounds.

For each variable $x_i$, we have three voters $s_i$, $v_i$ and $\bar v_i$, $3\ell$ voters $r^i_1,t^i_1,w^i_1,\dots,r^i_\ell,t^i_\ell,w^i_\ell$.
Furthermore, we have four additional voters $a,u_1,u_2,u_3$. Finally, we have a distinct voter $v$, who will
try to free-ride.

In the first $k-1$ rounds, there are two candidates, $c$ and $\bar c$, and every voter approves of both candidates.
Here, $v$ cannot increase her satisfaction by manipulating, and all voters win in each round.
Thus, the satisfaction of every voter after the first $k-1$ rounds will be exactly $k-1$.

Now, focus on rounds from $k$ to $k+3n-1$. We subdivide this set of rounds into triples; that is,
the first triple is $k,k+1,k+2$, the second $k+3,k+4,k+5$, and so on. We refer to the $j$-th
round of triple $i$ as round $(i, j)$; for example, round $(2, 3)$ corresponds to round $k+5$.
In each the first round of every triple, there is one candidate $c$, plus one candidate $c_{v^*}$ for
every voter $v^*\in N$. In the second and third rounds, we additionally have a voter $\bar c$.
We assume that, if there is a tie, $c$ wins against every voter-candidate, but loses against $\bar c$ every other voter.

For every triple $i$, in round $(i,1)$, voters $a,v,s_i$ approve of candidate $c$. Every other voter approves only of her voter-candidate.
In round $(i,2)$, voters $v,v_i,r^i_1,\dots,r^i_\ell$ and $a,\bar v_i,t^i_1,\dots,t^i_\ell$ vote for $c$ and $\bar c$, respectively.
The rest of the voters vote for their voter-candidate.
In round $(i,3)$, voters $v,t^i_1,\dots,t^i_\ell$ and $w^i_1,\dots,w^i_\ell$ vote for $c$ and $\bar c$, respectively. Again, the rest of the voters vote for their voter-candidate.

First, note that whoever votes for a voter-candidate never wins. We show this by induction. In round $k$, every voter has satisfaction $k-1$, and hence the candidates approved most often will win; this cannot be any voter-candidate, as $c$ is the most approved. Now, suppose this holds up to a round $i$. In round $i+1$, we know that there is at least one voter voting for $c$ that has never won since round $k-1$, as she always voted for her voter-candidate up to round $i$; for example, voters $s_i$, $v_i$ or $w^i_1$ in the cases where round $i+1$ is the first, second or third round of the triple, respectively. Such voters contribute to the score of $c$ with a value of $f(k-1)$. Since no voter-candidate can have a score greater than $f(k-1)$, our claim follows.

Consequently, in every round $(i,1)$, voter $v$ can free-ride by voting for her voter-candidate. We claim that if $v$ free-rides in $(i,1)$ then she wins in round $(i,2)$ and loses in round $(i,3)$; if she does not, the opposite happens. Furthermore, we claim that at the beginning of every triple, $v$ and $a$ have the same satisfaction. We show both claims by induction.

Consider round $(1, 1)$. If $v$ does not free-ride, $c$ wins, and the satisfaction of $a$ and $v$ will be $k$. Now consider round $(1, 2)$. The approval score of both $c$ and $\bar c$ is $f(k+1)+f(k)+f(k)\ell$, and hence $\bar c$ wins by tiebreaking. Thus, in the next round, the approval score of $c$ will be $f(k)\ell+f(k+1)$, and of the score of $\bar c$ will be $f(k)\ell$; hence, $c$ wins. Now, suppose that $v$ free-rides. Then, her satisfaction after round $(1, 1)$ will be $k-1$. Hence, the approval score of $c$ would increase to $2f(k)+f(k)\ell$, making it the winner. Hence, in round $(1, 3)$ the scores of $c$ and $\bar c$ would be $f(k+1)(\ell+1)$ and $f(k)\ell$, respectively. As we assume $f(k+1)(\ell+1)<f(k)\ell$, here we have that $\bar c$ wins, as desired. Observe that in both cases $v$ and $a$ won the same number of rounds. 

Now suppose the claim holds up to triple $i$. Then, let $s$ be the satisfaction of $v$ and $a$ at the beginning of triple $i+1$. Observe that if $v$ does not free-ride in $(i+1, 1)$, then in round $(i+1, 2)$ the approval score of $c$ and $\bar c$ is $f(s+2)+f(k)+f(k)\ell$, and hence $\bar c$ wins by tiebreaking. Thus, in round $(i+1, 3)$, the score of $c$ will be $f(k)\ell+f(s+2)$, and of the score of $\bar c$ will be $f(k)\ell$; hence, $c$ wins. If $v$ does free-ride, then in round $(i+1, 2)$ the approval score of $c$ raises to $f(s+1)+f(k)+f(k)\ell$, making it the winner. Furthermore, in round $(i+1, 3)$, the scores of $c$ and $\bar c$ would be $f(k+1)\ell+f(s+2)$ and $f(k)\ell$, respectively. Clearly, $\bar c$ wins here. Observe again that $v$ and $a$ won the same number of rounds.

Let us move to the final round. Here, there are $m+1$ candidates, namely $c,c_1,\dots,c_m$. Here, each voter $v_i$ votes for $c_j$ if $x_i\in C_j$ (and similarly for $\bar v_i$ and $\bar x_i$). Furthermore, $v,u_1,u_2,u_3$ vote for $c$, $a$ votes for all voters except for $c$, and everyone else votes for all candidates. We can interpret $c$ (resp. $\bar c$) winning in round $(i,2)$ as setting $x_i$ to true (resp. false). We claim that this assignment satisfies $\phi$ if and only if $c$ wins in the final round.

Indeed, let $\alpha$ be the score contributed by the voters who vote for all candidates. Furthermore, let $\beta$ be the score contributed by $v$ or by $a$, which are the same (as shown before). Observe that $u_1,u_2,u_3$ won exactly $k-1$ rounds. Furthermore, $v_i$ won exactly $k$ rounds if $c$ won in round $(i,2)$, and $k-1$ otherwise (and conversely for $\bar v_i$ and $\bar c$). 

Thus, the score of $c$ in the final round will be $\alpha+\beta+3f(k)$. Furthermore, given a clause $C_j$, if all of its literals are unsatisfied, the score of $c_j$ will also be $\alpha+\beta+3f(k)$. By our rule on tiebreaking, here $c_j$ wins. If, on the other hand, some literals in $C_j$ are satisfied, the score of $c_j$ will be at most $\alpha+\beta+2f(k)+f(k+1)$. Hence, if all clauses are satisfied, $c$ wins, as desired.

Now, observe that $v$ can free-ride only in every first round of every triple, but not elsewhere. Indeed, in the first $k-1$ rounds, whatever $v$ votes for will be the winner. Furthermore, in the second and third rounds of every triple, $v$ is either losing (and hence cannot free-ride) or her weight is breaking a tie between $c$ and some other candidate (which means that, if she were to vote for some other candidate instead, $c$ would no longer win). Observe also that, as shown earlier, $v$ will win all the $k-1$ first rounds, plus two rounds per triple (irrespective of whether she free-rides or not). Therefore, the only way that $v$ can raise her satisfaction is by making $c$ win in the last round by forcing a satisfying assignment for $\phi$ by free-riding. It follows that $v$ can free-ride if and only if $\phi$ is satisfiable, and we are done.\end{proof}


We further consider the egalitarian rule.

\begin{restatable}{theorem}{seqFreeridingEgal}
\freeridingProb{} is \NP-complete for the sequential egalitarian rule.\label{thm:seqFreeridingEgal}
\end{restatable}

The proof of Theorem~\ref{thm:seqFreeridingEgal} can be found in Appendix~\ref{app:complexity}.

Now, we will consider a weaker notion of free-riding, called \emph{generalized free-riding}. This allows us to get results for a broader class of rules, in particular for OWA rules. Here, we assume that a voter can free-ride if she can change her ballot so that she still approves of the winning alternative (though this alternative might differ from the truthful winner). That is, we just require that the new winning candidate is still (truthfully) approved by the manipulator. At its core, generalized free-riding is based on the assumption that voters are indifferent between approved candidates. To formally define generalized free-riding, we replace $\calR(\elec^*)_i=w_i$ in Definition~\ref{def:freeriding} with $\calR(\elec^*)_i\in A_i(v)$.

The problem of \genFreeridingProb{} is then defined analogously to \freeridingProb{}.

\newcommand{\seqFreeridingThieleGen}{\genFreeridingProb{} is \NP-complete for every sequential $f$-Thiele rule distinct from the utilitarian rule such that $f(i)>0$ holds for every $i\in\mathbb N$.}

\begin{theorem}\label{thm:seqFreeridingThieleGen}
\seqFreeridingThieleGen
\end{theorem}

\begin{proof}
Fix an $f$-Thiele rule $\calR$ satisfying the conditions of the theorem. First, notice that \genFreeridingProb{} is in \NP,
as we can guess an insincere approval ballot for the manipulator and check whether it 
improves her satisfaction (and is an instance of generalized free-riding) in polynomial time.

Now we show hardness by a reduction from \textsc{3-SAT}. In the following,
recall that, in every round $i$, a voter $v$ gives each of her 
approved candidates an extra score of $f(\sat(v,\bar w_1^{i-1})+1)$, where $\bar w_{1}^{i-1}=w_1,\dots w_{i-1}$, and the candidate
with the highest score wins.

Then, let $\phi$ be a 3-CNF with $n$ variables and $m$ clauses.
We assume w.l.o.g.\ that $\phi$ is not satisfied
by setting all variables to true and that each clause contains 
exactly three literals.
Furthermore, let $k \in \mathbb N$ be the smallest $k$ such that $f(k+1)<f(k)$.
As $\calR$ is not the utilitarian rule, such a $k$ indeed exists and is constant in the size of $\phi$. 
We construct a \genFreeridingProb{}
instance with $k+n$ rounds as follows:
For every variable $x_i$ there are $4$ voters $v^1_i$, $v^2_i$, $\bar v^1_i$ and $\bar v^2_i$. 
Furthermore, we add nine voters $v_0^1$, $v_0^2$, $w$, $u_1,\dots, u_6$. 
Finally, we add another voter~$v$, who will be the distinguished voter that tries
to manipulate.
In the first $k-1$ rounds, there are two candidates $c$ and $\bar c$.
In the following $n$ rounds, there are three candidates $c_0$, $c$ and $\bar c$ plus one candidate
$c_{v^*}$ for every $v^* \in N$.
In round $k+n$, we have $m+1$ candidates $c, c_1, c_2 \dots, c_m$
plus one candidate $c_{v^*}$ for every $v^* \in N$.
We assume that if ties need to be broken between $c_0$ and another candidate,
then $c_0$ is selected and if a tie between $c$ and a candidate 
other than $c_0$ needs to be broken, then $c$ wins.

In the first $k-1$ rounds, all voters approve both candidates. Hence, $v$ cannot increase her satisfaction by manipulating, and all voters win in each round.
Thus, the satisfaction of every voter after the first $k-1$ rounds will be exactly $k-1$.

We continue with $n$ rounds such that, in round $i$, $v^1_0$ 
and $v_0^2$ approve $c_0$, $v^1_i$ and $v_i^2$ approve $c$, $\bar v^1_i$ and $\bar v^2_i$
approve $\bar c$, $v$ approves both $c$ and $\bar c$ and $w$ approves $c, \bar c$ and $c_0$.
Finally, all other voters $v^* \in N$ only approve their candidate
$c_{v^*}$ except in round $k$, where $u_1$ and $u_2$ additionally approve $c_0$, $c$ and $\bar c$.

Then, in round $k+n$, $v$ and $u_1, \dots, u_6$ approve $c$.
Furthermore, for every candidate $c_i$ with $1 \leq i \leq m$,
$\bar v^1_j$ and $\bar v^2_j$ approve $c_i$ if and only if variable $x_j$ appears
positively in $C_i$ and $v^1_j$ and $v_j^2$ approve $c_i$
if and only if variable $x_j$ appears negatively in $C_i$.
Additionally, $w$ approves $c_1, \dots, c_m$.
All other voters approve only their candidate.

We claim that in rounds $k$ to $k+n-1$ \emph{(i)} $c$ wins if $v$ approves $c$ and $\bar c$ (i.e., $v$ does not misrepresent her preferences),
\emph{(ii)}
either $c$ 
or $\bar c$ becomes the winner if $v$ approves only one of these candidates, and
\emph{(iii)} $c_0$
wins if $v$ approves neither $c$ nor $\bar c$.
We show the claim by induction.
Up until round $k-1$, all voters have gathered the same satisfaction $k-1$, and hence in round $k$ each voter contributes to
the score of their approved candidates with the same value of $f(k)$. The only candidates that are approved
by more than one voter are $c_0$, $c$ and $\bar c$, where
$c_0$ is approved by five voters, while $c$ and $\bar c$ are 
both approved by six voters. Therefore, by our assumption
about tiebreaking, $c$ is the winner in round $k$. 
Furthermore, if $v$ misrepresents her preferences 
and votes only for either $c$ (resp.\ $\bar c$), then 
$c$ (resp.\ $\bar c$) is the unique winner in round $k$.
Finally, if $v$ approves neither $c$ nor $\bar c$, then $c_0$
wins by our assumption on tiebreaking. Observe that
$v$ is not allowed to make $c_0$ win,
by the definition of generalized free-riding.

Now assume the claim holds for rounds $k, \dots i-1$.
Then, in round $i$ the only candidates that are approved
by more than one voter are again $c_0$, $c$ and $\bar c$.
To be more precise, $c_0$ is approved by $v^1_0, v_0^2$ and $w$, $c$ is approved by $v_i^1$, $v_i^2$, $w$ and $v$
and $\bar c$ by $\bar v_i^1$, $\bar v_i^2$, $w$ and $v$.
Crucially, $v_0^1$, $v_0^2$, $v^1_i$, $v^2_i$, $\bar v^1_i$ and $\bar v^2_i$
have not won in any of rounds $k, \dots i-1$, as they never approved of $c$ or $\bar c$ in these rounds;
hence, their satisfaction at this point is still $k-1$.
Now, let $s_v^{i-1}$ and $s_w^{i-1}$ be the satisfaction of $v$ and $w$ up to round round $i-1$, respectively.
Then, both $c$ and $\bar c$ have a score of $2f(k) + f(s_v^{i-1}+1) + f(s_w^{i-1}+1)$,
$c_0$ has score $2f(k) + f(s_w^{i-1}+1)$ whereas all other alternatives can have
a maximal score of $f(k)$.
As we know that $f(s_v^{i-1}+1) > 0$,
this implies, by our assumption
about tiebreaking, that $c$ is the winner in round $i$. 
Furthermore, as before, if $v$ misrepresents her preferences,
she can make $c$ resp.\ $\bar c$ the unique winner
and if she approves neither $c$ nor $\bar c$, then $c_0$
wins by our assumption on tiebreaking. Observe again that $c_0$ cannot win here,
by definition of generalized free-riding.

We can interpret the winners in the rounds $k, \dots k+n-1$ as 
a truth assignment $T$ by setting $x_i$ to true if $c$ wins in round $i$
and to false if $\bar c$ wins in round $i$ (observe that $c_0$ can never win by the previous arguments).
Then, we claim that $c$ wins in round $k+n$ if and only if
$C_j$ is satisfied by this truth assignment:
All voter-candidates $c_{v^*}$ are approved by at most one voter with satisfaction $k-1$,
and hence have a score of at most $f(k)$.
The satisfaction of $v^1_j$ and $v^2_j$ is $k-1$ if $x_j$ is set to false in $T$ and
$k$ if $x_j$ is set to true in $T$.
Similarly, the satisfaction of $\bar v^1_j$ and $\bar v^2_j$ is $k-1$ if $x_j$ is set to true in $T$ and
$k$ if $x_j$ is set to false in $T$.
Finally, the satisfaction of $w$ is $k+n-1$.
Hence, the approval score of $c_i$ is $6f(k) + f(k+n)$ if all literals 
in $C_j$ are set to false and at most $4f(k) + 2f(k+1) + f(k+n)$ if
at least one literal is set to true.
On the other hand, the satisfaction of $u_1$ and $u_2$ is $k$, the satisfaction of 
$u_3, \dots, u_6$ is $k-1$ and the satisfaction of $v$ is $k+n-1$. 
Hence, the approval score of $c$ is $4f(k) + 2f(k+1) + f(k+n)$.
As we assumed $f(k+1) < f(k)$, we get that $4f(k) + 2f(k+1) + f(k+n) < 6f(k) + f(k+n)$.
Therefore, if there is a clause $C_i$ for which no literal is set to true,
then $c_i$ has a higher approval score than $c$ and hence, $c$ is not a winner 
in round $k+n$. On the other hand, if for every clause at least one literal 
is set to true, then $c_1, \dots, c_m$ have at most the same score as $c$ and $c$ wins by
tiebreaking.

Now, the honest ballot of $v$ leads to the truth assignment in which
every variable is set to true by tiebreaking. By assumption, this assignment does not 
satisfy $\phi$ and hence $c$ does not win in round $k+n$. 
By construction, the satisfaction of $v$ equals $k+n-1$ in this case.
Moreover, as $v$ cannot manipulate in the first $k-1$ rounds,
the only way that $v$ can gain more satisfaction
is by forcing the winners in rounds $k, \dots k+n-1$ to form a 
satisfying truth assignment without allowing $c_0$ to win any round.
Hence, $v$ can manipulate via generalized free-riding if and only if
$\phi$ is satisfiable. 
\end{proof}

\begin{restatable}{theorem}{seqFreeridingOWAGen}
\genFreeridingProb{} is \NP-complete for every sequential $\alpha$-OWA rule such that, for all $n$, $\alpha=(\alpha_1,\dots,\alpha_n)$ is nonincreasing and $\alpha_1>\alpha_n$.\label{thm:seqFreeridingOWAGen}
\end{restatable}

Observe that we give three different reductions, depending on the form of the OWA vector. We present here one case; the others can be found in Appendix~\ref{app:complexity}.

\begin{proof}[Proof of Theorem~\ref{thm:seqFreeridingOWAGen}]
Fix an $\alpha$-OWA rule $\calR$ satisfying the conditions of the theorem. First, notice that \genFreeridingProb{} is in \NP,
as we can guess an insincere approval ballot for the manipulator and check whether it 
improves her satisfaction (and is an instance of generalized free-riding) in polynomial time.

Next, we show hardness by a reduction from \textsc{3-SAT}. 
Let $\phi$ be a 3-CNF with $n$ variables and $m$ clauses. We refer to the $j$-th clause as $C_j$.
We assume w.l.o.g.\ that $\phi$ is not satisfied
by setting all variables to false and that each clause contains 
exactly three literals.
We will construct an instance of \genFreeridingProb{} with $2n+5$ voters.
More specifically, there are two voters $v_i$ and $\bar v_i$ for each variable $x_i$, four 
voters $u_1,\dots,u_4$, plus one distinguished voter $v$ who will try to manipulate. 

Given the weight vector $\alpha=(\alpha_1,\dots,\alpha_{2n+5})$, we distinguish three (not necessarily exclusive) cases: (1) $\alpha_2>\alpha_{n+5}$, (2) $\alpha_{n+1}>\alpha_{2n+5}$, and (3) $\alpha_{2}=\alpha_{2n+5}$. Since $\alpha_1>\alpha_{2n+5}$, at least one case must be true. In the following,
we will give a reduction for the first case ($\alpha_2>\alpha_{n+5}$). The other cases are similar; we discuss them in Appendix~\ref{app:complexity}.

We construct an instance with $n+2$ rounds as follows.
In the first $n+1$ rounds, there are two candidates: $c$ and $\bar c$. In
the last round, there are $m+1$ candidates, namely $c,c_1,\dots,c_m$. We assume that, in the case of ties, $c$ always loses.

In the first round, everybody votes for $\bar c$ except for $v$ and $u_1$, who vote for $c$. Here, $\bar c$ wins by tiebreaking, and $v$ cannot manipulate.

In each round $i$ with $i\in\{2,\dots,n+1\}$, voter $v_i$ votes for candidate $\bar c$, voter $\bar v_i$ for candidate $c$; everyone else votes for both candidates. We claim that $v$ can manipulate in every such round $i$ and force the win of either $c$ or $\bar c$. We show so by induction. In round $2$,
suppose that $v$ votes only for $c$ (the case where she votes for $\bar c$ is analogous). Then, if $c$ were to win, the satisfaction vector would be of form $(1,1,1,2,\dots,2)$ (everyone but $v_1$ wins). On the other hand, if $\bar c$ wins, then it would be of form $(0,1,1,2,\dots,2)$ (everyone but $\bar v_1$ and $v$ win). Hence, $c$ wins. Observe that if $v$ votes truthfully, $\bar c$ wins by tiebreaking. Now, suppose this holds up to a round $i$. Before round $i+1$, $v$ has won $i-1$ rounds (all but the first one), whereas voters $u_1,\dots,u_4$, as well as any pair of voters $v_j,\bar v_j$ (with $j\geq i$) have won $i$ rounds. Furthermore, for every pair $v_j,\bar v_j$ (with $j<i$), exactly one voter won $i-1$ rounds while the other $i$ rounds. Suppose again that $v$ votes for $c$ (the case where she votes for $\bar c$ is similar). Then, observe that, if $c$ or $\bar c$ win,s the satisfaction vectors (excluding $v$) would be completely symmetric (and every voter would have at least a score of $i$); however, if $\bar c$ wins, $v$ would have a satisfaction of $i-1$, whereas if $c$ wins, she would get a satisfaction of $i$. Hence, since $\alpha_1>0$, here $c$ wins. Observe again that if $v$ votes truthfully, then $\bar c$ wins.

Consider the final round. Up to here, $v$ and $u_1$ have won $n$ rounds (they lost the first round), while $u_2,u_3,u_4$ have won $n+1$ rounds. Furthermore, every voter $v_i$ has won $n$ rounds if $c$ won in round $i+1$ and $n+1$ times otherwise (and conversely for $\bar v_i$ and $\bar c$). In this round, voters $u_1,\dots,u_4$ approve of all candidates but $c$, voter $v_i$ (resp. $\bar v_i$) approves of $c$ and every candidate $c_j$ such that $x_i\not\in C_j$ (resp. $\bar x_i\not\in C_j$). Finally, voter $v$ approves of $c$. Observe that we can interpret $c$ winning in round $i+1$ as setting $x_i$ to true, and $\bar c$ winning as setting $x_i$ to false. We claim that $c$ wins in the last round if and only if this assignment satisfies $\phi$. To see this, consider that, if $c$ were to win, the satisfaction vector would be:
\begin{equation*}
  (n,\underbrace{n+1,\ldots,n+1}_{n+4\ \text{times}},\underbrace{n+2,\ldots,n+2}_{n\ \text{times}}).
\end{equation*}
Let us call this vector $s$. Consider now a candidate $c_j$ and its corresponding clause $C_j$. If all three of its literals are unsatisfied, then the corresponding voters all have satisfaction $n+1$. Hence, if $c_j$ were to win in this case, the satisfaction vector would again be exactly $s$. By our assumptions on tiebreaking, here $c_j$ would win against $c$. Furthermore, suppose that either one, two, or three of the literals have been satisfied. Then, the vectors are, respectively:
\begin{align*}
  &(n,n,\underbrace{n+1,\ldots,n+1}_{n+2\ \text{times}},\underbrace{n+2,\ldots,n+2}_{n+1\ \text{times}})\\
  &(n,n,n,\underbrace{n+1,\ldots,n+1}_{n\ \text{times}},\underbrace{n+2,\ldots,n+2}_{n+2\ \text{times}})\\
  &(n,n,n,n,\underbrace{n+1,\ldots,n+1}_{n-2\ \text{times}},\underbrace{n+2,\ldots,n+2}_{n+3\ \text{times}})\\
\end{align*}
Let these vectors be $s_1$, $s_2$ and $s_3$, respectively. One can show that if $\alpha_{2}>\alpha_{n+5}$ the dot product between $\alpha$ and each of these three vectors would be strictly lower than the dot product between $\alpha$ and $s$. For example:
\begin{multline*}
    s\cdot\alpha >  s_1\cdot\alpha
    \iff \alpha_1n+\left(\sum_{i=2}^{n+5}\alpha_i(n+1)\right)+\left(\sum_{i=n+6}^{2n+5}\alpha_i(n+2)\right) >\\
    (\alpha_1+\alpha_2)n+\left(\sum_{i=3}^{n+4}\alpha_i(n+1)\right)+\left(\sum_{i=n+5}^{2n+5}\alpha_i(n+2)\right) \iff \\
    (\alpha_2+\alpha_{n+5})(n+1) > \alpha_2n+\alpha_{n+5}(n+2)
    \iff \alpha_{2} > \alpha_{n+5}.
\end{multline*}
The other two cases are similar. Hence, if $C_j$ is satisfied, candidate $c_j$ cannot win against $c$. Consequently, if all clauses are satisfied, candidate $c$ wins.

Now, if $c$ wins in the last round, then the satisfaction of $v$ would be $n+1$; if $c$ loses, it would be $n$. Notice also that $v$ cannot raise her satisfaction by manipulating in the final round. Furthermore, if $v$ always submits her true preferences, then by tiebreaking $\bar c$ would win in every round $i$ with $i\in\{2,\dots,n+1\}$. By assumption, this would not satisfy $\phi$, and hence $c$ would not win in the last round. Therefore, $v$ has an incentive to manipulate in these rounds to try and choose a satisfying assignment for $\phi$. It follows that $v$ can manipulate via generalized free-riding if and only if $\phi$ is satisfiable, so we are done.
\end{proof}

\section{Numerical Simulations}\label{sec:experiments}

\graphicspath{{figures/}}

So far, we have seen that sequential Thiele and OWA rules are generally susceptible to free-riding.
However, we have also seen that free-riding can be detrimental to free-riders, i.e., their satisfaction can decrease.
In this section, we use numerical simulations\footnote{The source code as well as additional plots are available at \url{https://github.com/martinlackner/free-riding}.} to shed more light on the \emph{possibility, effect, and risk} of free-riding with sequential rules.
We answer three main questions:
\begin{enumerate}
  \item How often do voters have the possibility to increase their (final) satisfaction by free-riding? How many instances contain issues in which free-riding leads to a lower satisfaction for the free-rider?
  \item What is the average risk of free-riding? That is, what is the likelihood of free-riding resulting in a negative outcome?
  \item Is there a difference in the effect of free-riding if it is done by a voter belonging to a majority or minority?
\end{enumerate} 
To be able to give robust answers, we study these questions in a range of different models.

\subsection{Model setup and parameter values}
We assume that voters and candidates are points in a $2$-dimensional space; this
is known as the 2d-Euclidean model \cite{spatial,elk-fal-las-sko-sli-tal:c:2d-multiwinner,bre-fal-kac-nie2019:experimental_ejr,god-bat-sko-fal:c:2d}.
For each issue, we sample candidate points from a uniform distribution on the $[-1,1]\times[-1,1]$ square. That is, candidates are different in each issue, independent from each other and across issues.
In contrast, voters' points are the same for all issues. We consider three distributions for voters:\footnote{The unbalanced and many groups scenarios have been proposed by \citet{ChandakGP24}. We use the same parameter values to increase comparability, even though their election model is different from ours.}
\begin{itemize}
  \item \emph{square}: A uniform distribution on the $[-1,1]\times[-1,1]$ square, i.e., voters and candidates are independently drawn from the same distribution.
  \item \emph{many groups}: Voters are split into four groups, the first three comprising each of 20\%, the last of 40\% of voters. Each group is centered around a different point, namely one of $(\pm 0.5, \pm 0.5)$. Voters' $x$ and $y$ coordinates are drawn independently from a normal distribution with standard deviation $0.2$.
  \item \emph{unbalanced}: Here there are two disjoint groups, one comprising of 20\% of voters, the other of $80\%$. As before, the voters' locations are sampled from a normal distribution, with centers $(0.5,0.5)$ and $(-0.5,-0.5)$, respectively, and standard deviation $0.1$.
\end{itemize}

We transform these coordinates into approval ballots as follows:
A voter approves the closest candidate as well as any candidate that is similarly close (within 1.5 times the distance).

Our default choice are multi-issue elections with $n=20$ voters, $k=20$ issues, and $5$ candidates per issue. Later on we investigate the effect of this parameter choices.
Finally, we sample $2000$ elections for each voter distribution and choice of parameters.

\subsection{Considered voting rules and forms of free-riding}

In our experiments, we consider a subclass of Thiele methods and a subclass of OWA rules.
For better comparison, we parameterize both classes with a parameter~$x$ (albeit this parameter has a different interpretation in both classes).
We consider $\alpha$-OWA rules with
\begin{align*}
\alpha_x&=(\underbrace{1,\dots,1}_{n-x\text{ many}},\frac{1}{kn}, \frac{1}{k^2n^2}, \dots) & \text{for } x \in \{0, 1,\dots, n-1\}.
\end{align*}
Note that for decreasing $x$, more and more voters receive full consideration (weight~1).
For $x=0$, all voters receive full consideration and we obtain the utilitarian rule.
In contrast, for $x=n-1$, only the least satisfied voter receives full consideration; this is the sequential leximin rule (cf. Proposition~\ref{prop:leximinOWA}). 

Further, we consider $f$-Thiele rules with
\begin{align*}
  &f_x(i) = \frac{1}{i^x} & \text{ for } x\in \{0, 0.25, 0.5, 0.75, \dots\}.
\end{align*}  
Note that for $x=0$ this is the utilitarian rule, for $x=1$ it is sequential PAV, and for increasing $x$ it approaches the sequential leximin rule, since $f_x$ diminishes quickly.

We distinguish two ways in which free-riding can occur. We speak of \emph{single-issue free-riding} if a voter may free-ride once in a given election. In contrast, with \emph{repeated free-riding} a voter free-rides whenever it is possible in a given election. Single-issue free-riding models situations where free-riding occurs rarely, e.g., if there is limited information available about expected outcomes, or if there is a high social cost for free-riding. 
In contrast, repeated free-riding models situations where free-riding is a viable strategy and and a free-riding voter can always exploit free-riding opportunities.

Note that we always assume that there is only one free-rider among all voters. We leave the study of free-riding group effects as a topic for future work.

\subsection{Results}
\subsubsection{Possibility and risk of free-riding}\label{sec:results-possibility}

  \begin{figure*}
  \centering
  \includegraphics[width = 0.49\textwidth]{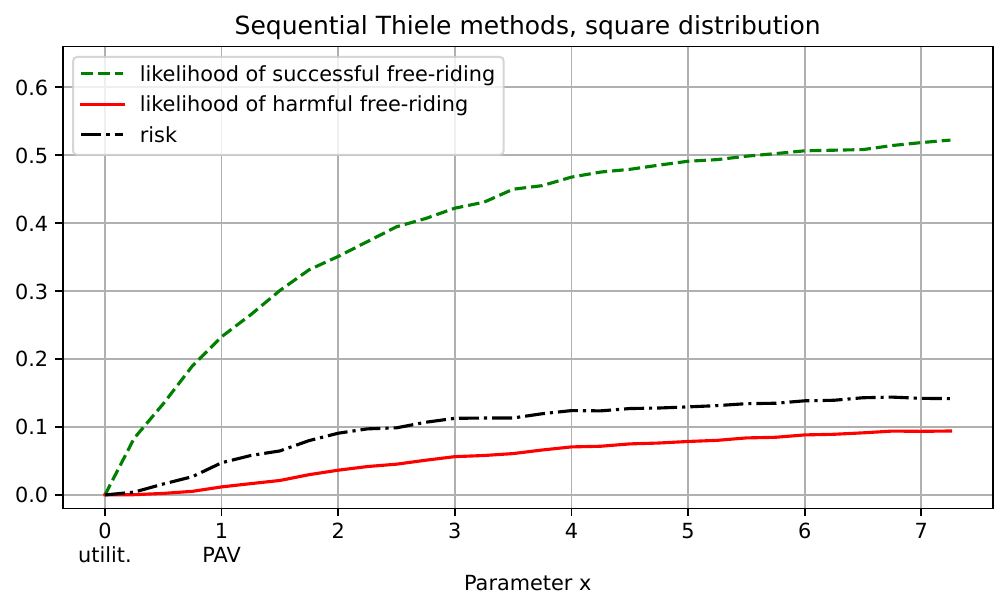}
  \includegraphics[width = 0.49\textwidth]{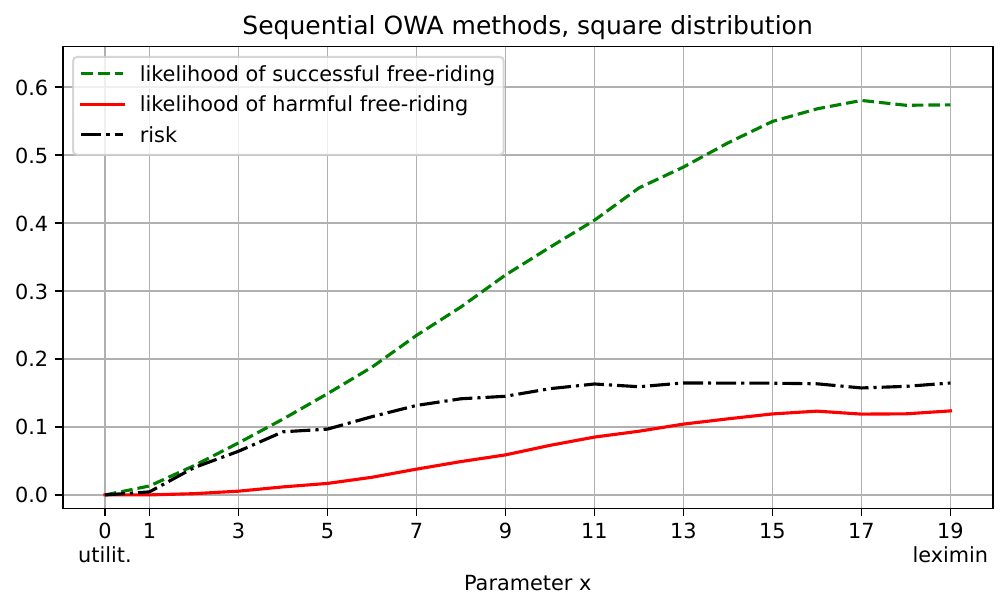} 
  
  \includegraphics[width = 0.49\textwidth]{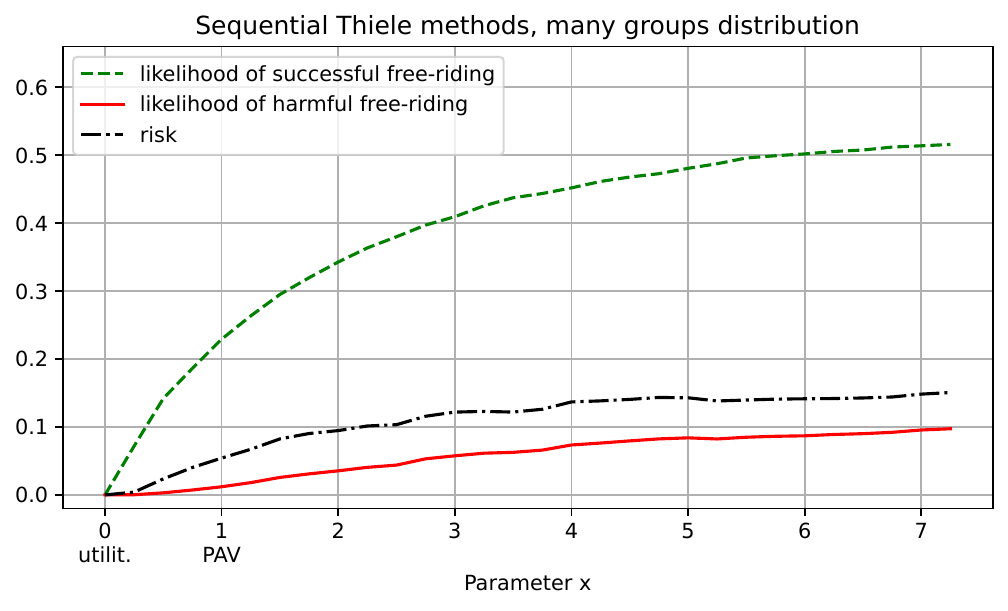}
  \includegraphics[width = 0.49\textwidth]{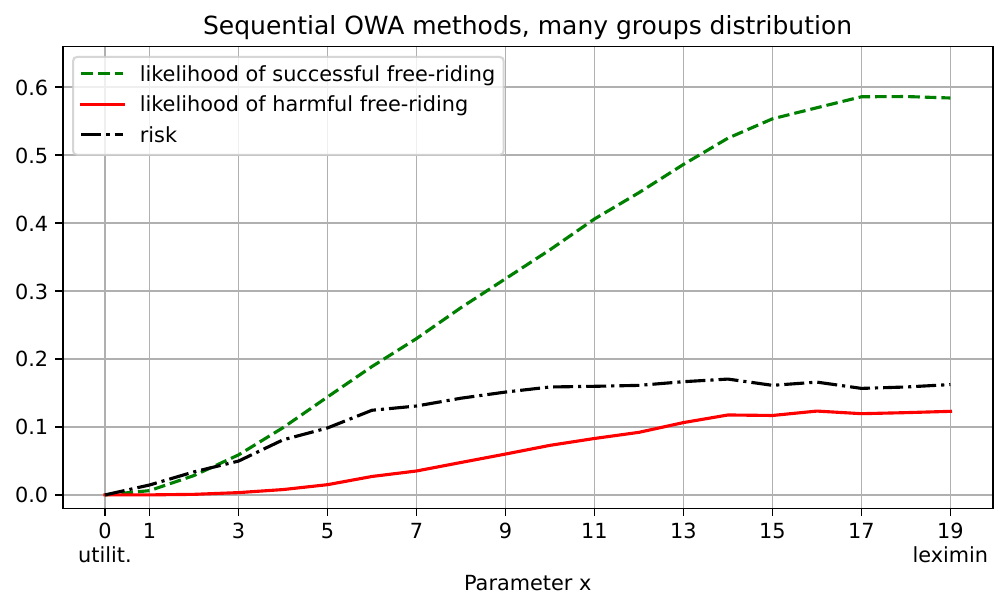} 
  
  \includegraphics[width = 0.49\textwidth]{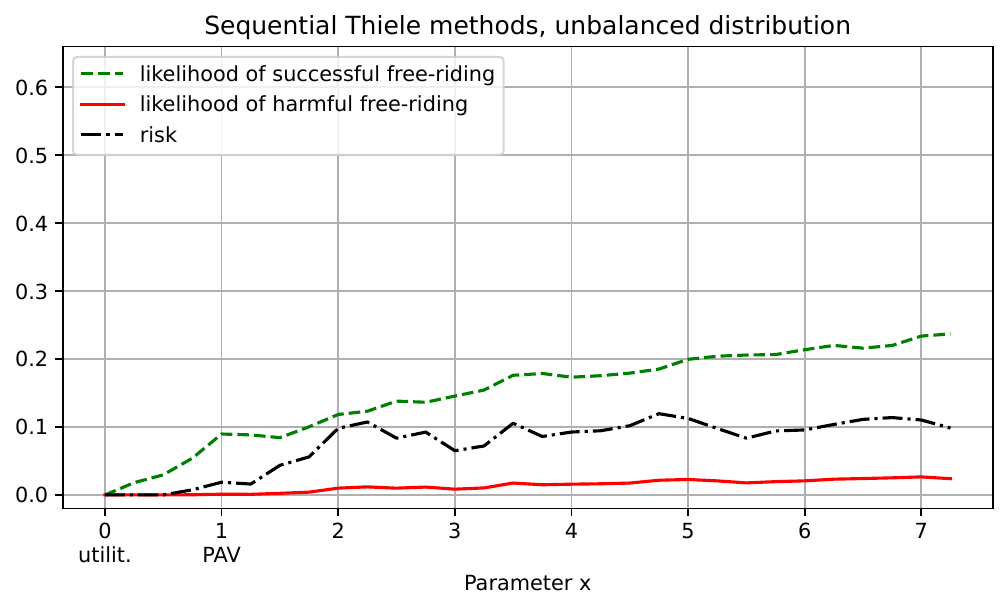}
  \includegraphics[width = 0.49\textwidth]{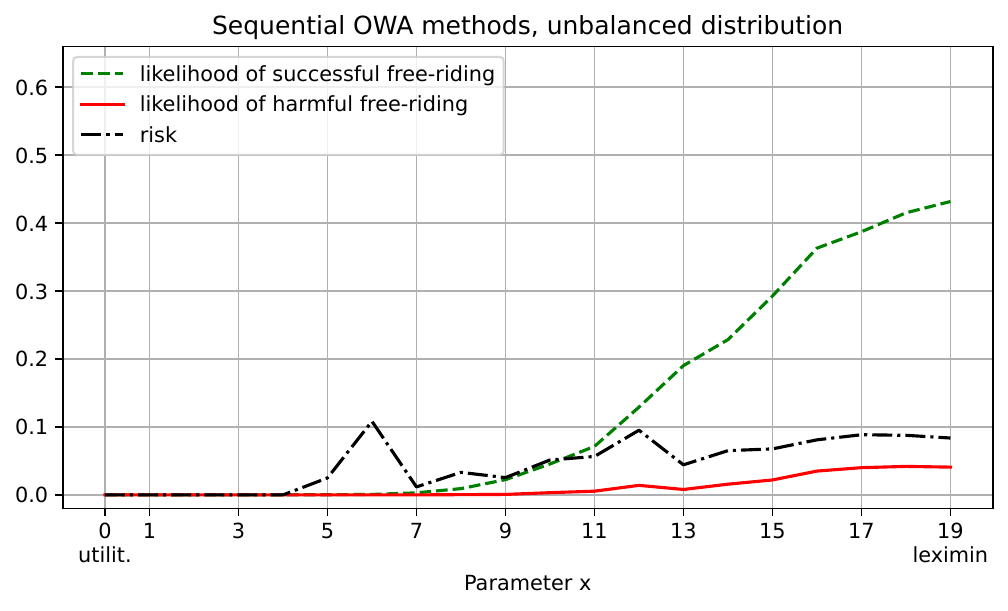} 
  \caption{Possibility and risk of single-issue free-riding. The left column of diagrams show sequential Thiele methods, the right sequential OWA methods. The three rows correspond to the \emph{square}, \emph{many groups}, and \emph{unbalanced} distributions.}
  \label{fig:exp-single}
\end{figure*}

In the single-issue free-riding model, we speak of \emph{successful free-riding} for a given election, voter, and issue, if the voter can free-ride in the given issue and this increases her overall satisfaction (i.e., at the end of the election); we speak of \emph{harmful free-riding} if the voter can free-ride but this decreases her satisfaction.
Note that free-riding can also be neutral (with no change in satisfaction).
In the repeated setting, we speak of \emph{successful free-riding} for a given election and voter if the voter free-rides whenever possible and this increases her overall satisfaction; we speak of \emph{harmful free-riding} if this decreases her satisfaction.

Let us first consider the single-issue free-riding model. That is, for each multi-issue election, we iterate over all voters and all issues and check whether free-riding is possible (Definition~\ref{def:freeriding}). 
Figure~\ref{fig:exp-single} shows the proportion of instances where in at least one issue successful/harmful free-riding is possible (averaged over all voters).
(As successful and harmful free-riding may be possible in the same election -- albeit in different issues -- the proportion of successful and harmful free-riding may be larger than 1.)
In addition, it shows the risk of free-riding: we define the risk of a voter in an election
as the number of issues where harmful free-riding occurs divided by the number of issues where either successful or harmful free-riding occurs.

Let us discuss the results from Figure~\ref{fig:exp-single}. First, we see that the results for the \emph{square} and \emph{many groups} distributions are very similar. In contrast, the \emph{unbalanced} distribution yields different results. 
For all distributions, we clearly see that rules closer to the utilitarian rule ($x=0$) are less susceptible to free-riding than those closer to leximin (larger values of $x$). We also see that -- as expected --
the utilitarian rule is the only rule where free-riding is not possible (cf. Proposition~\ref{prop:util-cannot-manipulated}).
We note that this increase in susceptibility (with distance to the utilitarian rule) has also been observed by \citet{barrot2017manipulation} for arbitrary manipulations.
Both the proportion of voters that can successfully free-ride and those with the possibility of harmful free-riding grow with parameter~$x$.

The most important conclusion from this experiment is that the risk of single-issue free-riding is considerable:  for the square/many groups/unbalanced distributions, the risks are $2.7\%$/$4.0\%$/$0.7\%$ for PAV (Thiele, $x=1$), and $16.5\%$/$16.2\%$/$8.4\%$ for leximin (OWA, $x=19$).
This shows that harmful free-riding is not merely a theoretical possibility. In particular for more egalitarian voting rules (larger x-values in Figure~\ref{fig:exp-single}), this risk could indeed decrease the temptation of free-riding.

If we move to repeated free-riding, the definition of risk changes: here, the risk of a voter in a given election is either 0 (the outcome of repeated free-riding is positive), or 1 (the outcome is negative). Risk is averaged over all voters (for whom successful or harmful free-riding is possible) and over all elections. We see that the risk is small compared to the single-issue model. This is because the probability of successful free-riding in a single issue is generally larger than harmful free-riding (i.e., risk in the single-issue model is $
\leq 0.5$), the risk almost vanishes with sufficient repetitions.
Figure~\ref{fig:exp-rep} shows exemplarily the results for the \emph{square} distribution. We see much higher likelihood of successful free-riding with almost no risk involved.
Recall, however, that this form of free-riding requires the capability to reliably identify free-riding possibilities and to repeatedly exploit them.

\begin{figure*}
  \centering
  \includegraphics[width = 0.49\textwidth]{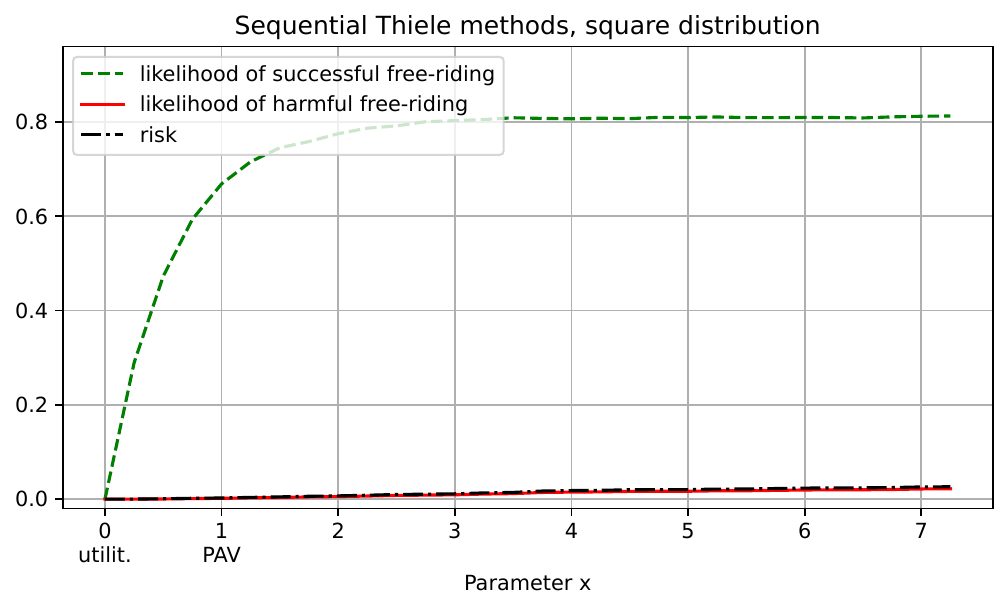}
  \includegraphics[width = 0.49\textwidth]{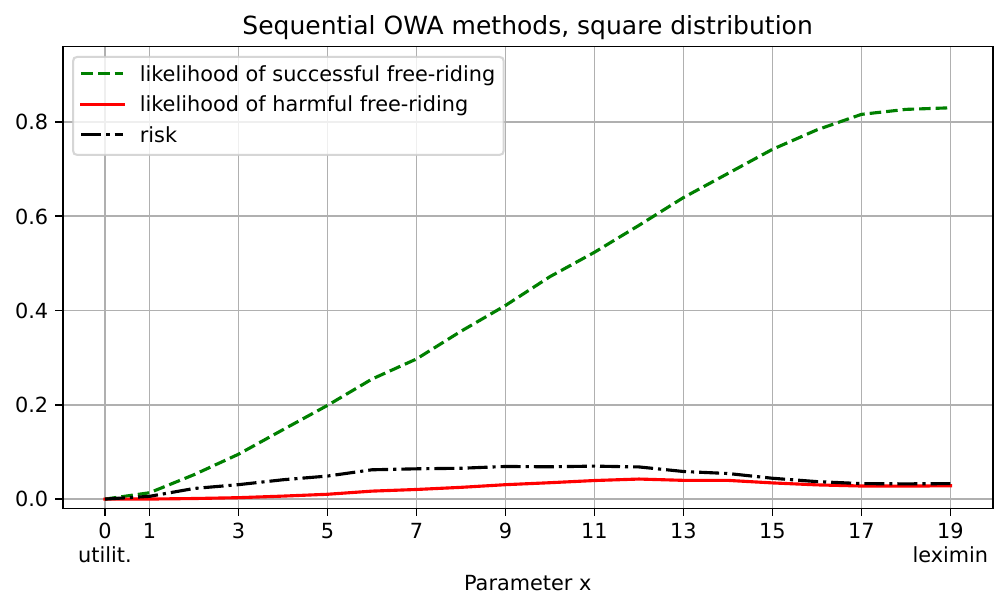} 
  \caption{Results for repeated free-riding with the \emph{square} voter distribution. In contrast to Figure~\ref{fig:exp-single}, we see a much higher likelihood of successful free-riding with very little risk.}
  \label{fig:exp-rep}
\end{figure*}

So far, we have not taken into account how much a voter can benefit from free-riding. To investigate this, we show in Figure~\ref{fig:exp-satisfaction} the expected change in satisfaction for a free-riding voter. As expected, repeated free-riding leads to a more pronounced change in satisfaction. Furthermore, the change in satisfaction grows with parameter $x$, i.e., with increasing distance to the utilitarian rule. This aligns with the increased possibility to free-ride (as we have seen in Figures~\ref{fig:exp-single} and~\ref{fig:exp-rep}).
Finally, note that in all cases, the change in satisfaction is not particularly large in comparison to the number of issues (20). Even with repeated free-riding when using the leximin rule, the expected number issues with an improvement is less than $1.6$.

\begin{figure*}
  \centering
  \includegraphics[width = 0.49\textwidth]{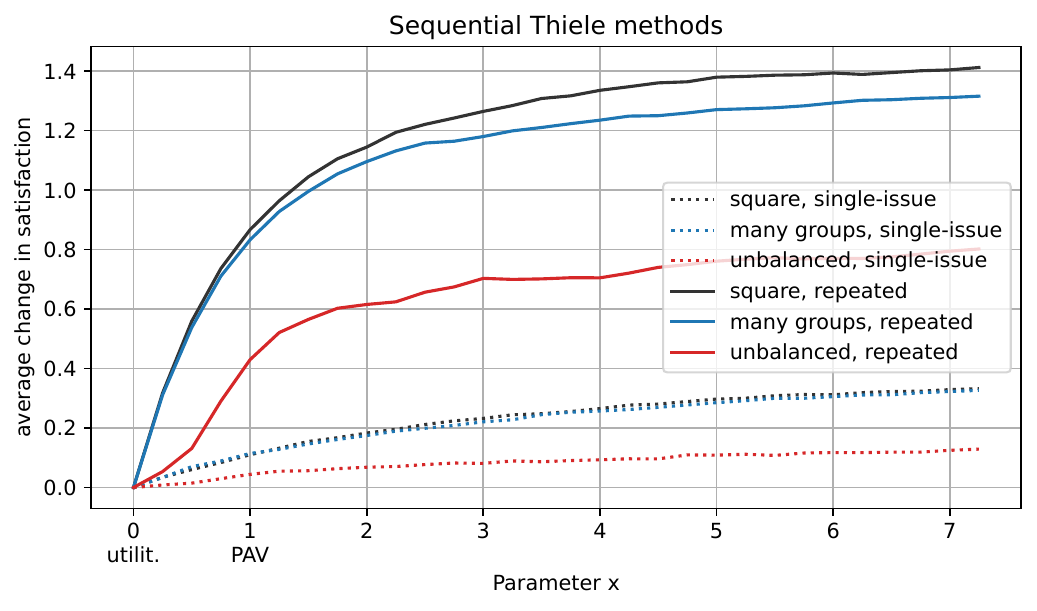}
  \includegraphics[width = 0.49\textwidth]{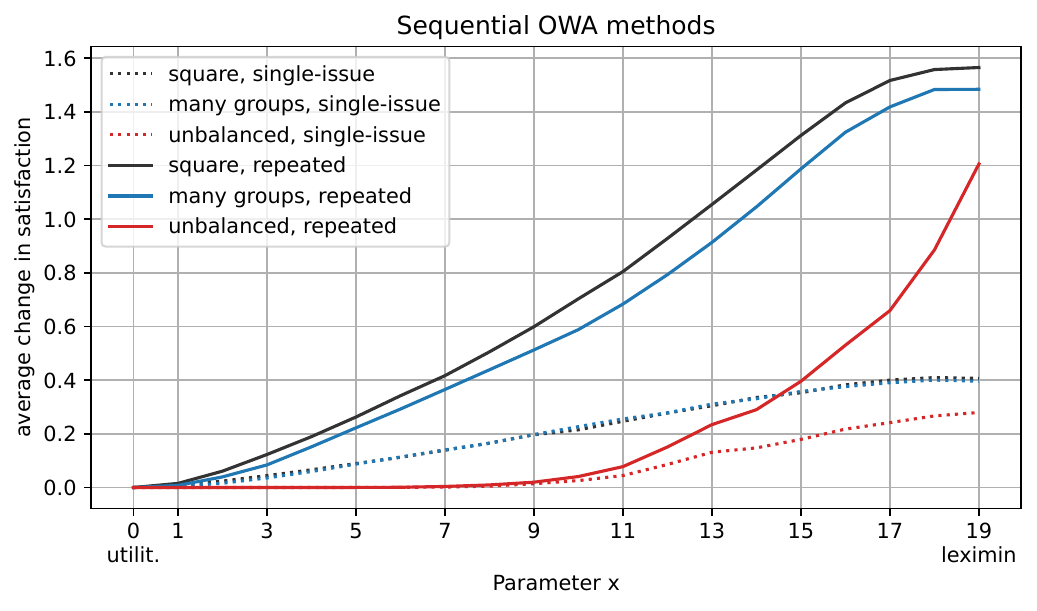} 
  \caption{Average change in satisfaction from free-riding, comparing different voter distributions as well as single-issue vs.\ repeated free-riding.}
  \label{fig:exp-satisfaction}
\end{figure*}

\subsubsection{Impact of model parameters}

Let us now briefly describe the impact of our chosen model parameters. 
In the following, we use our default setting ($n=20$ voters, $k=20$ issues, $5$ candidates per issue, \emph{square} distribution, single-issue free-riding) and vary a single parameter to observe its impact.

Increasing the number 
of voters (Figure~\ref{fig:exp-num_voters}) decreases the chance of voters being pivotal. In line with this observation, we see an overall decrease in both successful and harmful free-riding with an increase in voters. For a larger number of voters,
it would make sense to move to a model where groups of voters free-ride. This requires
additional assumptions about voter coordination (cf.\ the framework of iterative voting~\cite{meir2017iterative}).
\begin{figure*}
  \centering
  \includegraphics[width = 0.49\textwidth]{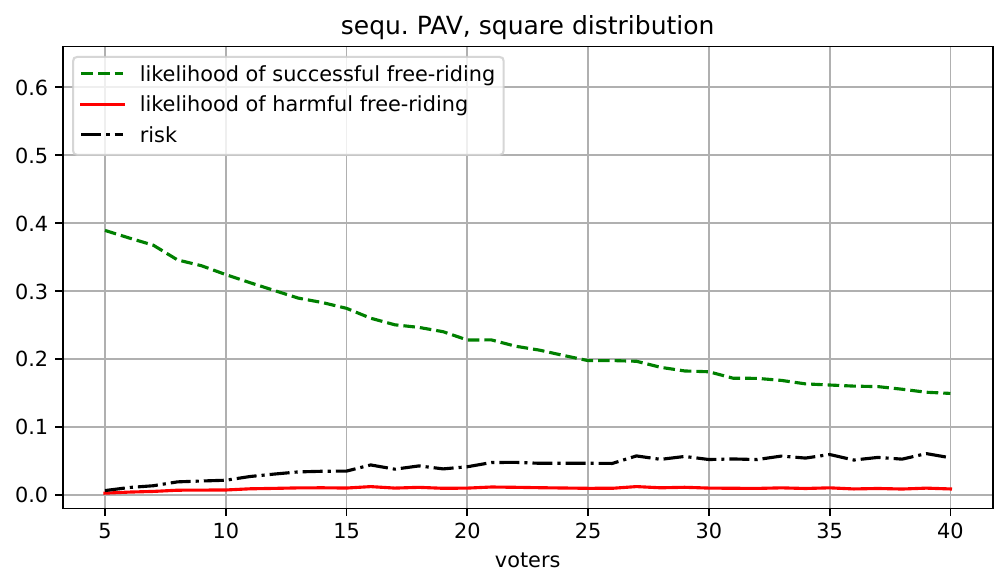}
  \includegraphics[width = 0.49\textwidth]{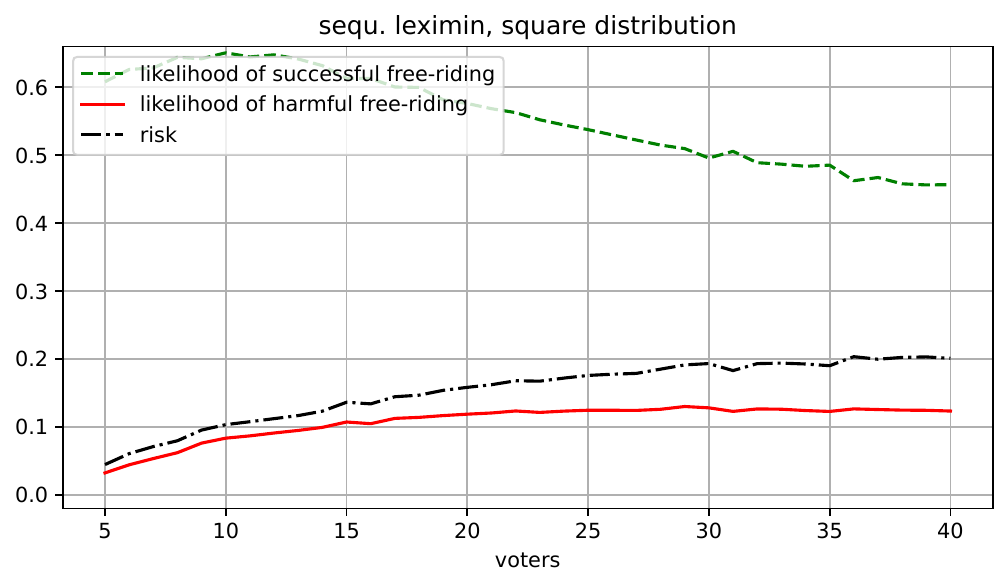} 
  \caption{Varying the number of voters (our default choice is $n=20$): an increase in the number of voters decreases the effectiveness of free-riding by a single voter.}
  \label{fig:exp-num_voters}
\end{figure*}

In Figure~\ref{fig:exp-num_cands}, we see that varying the number of candidates (per issue) has moderate impact. The increase in the likelihood of both successful and harmful free-riding can be explained by the increase in heterogeneity: as voters' preferences become more diverse and like-minded groups become smaller, a single free-riding action is more likely to change the outcome.
\begin{figure*}
  \centering
  \includegraphics[width = 0.49\textwidth]{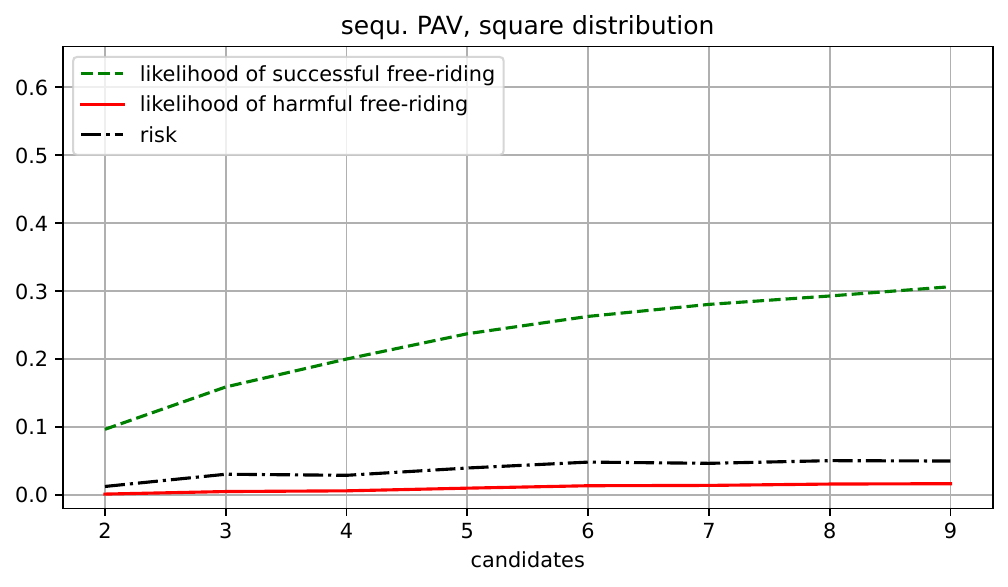}
  \includegraphics[width = 0.49\textwidth]{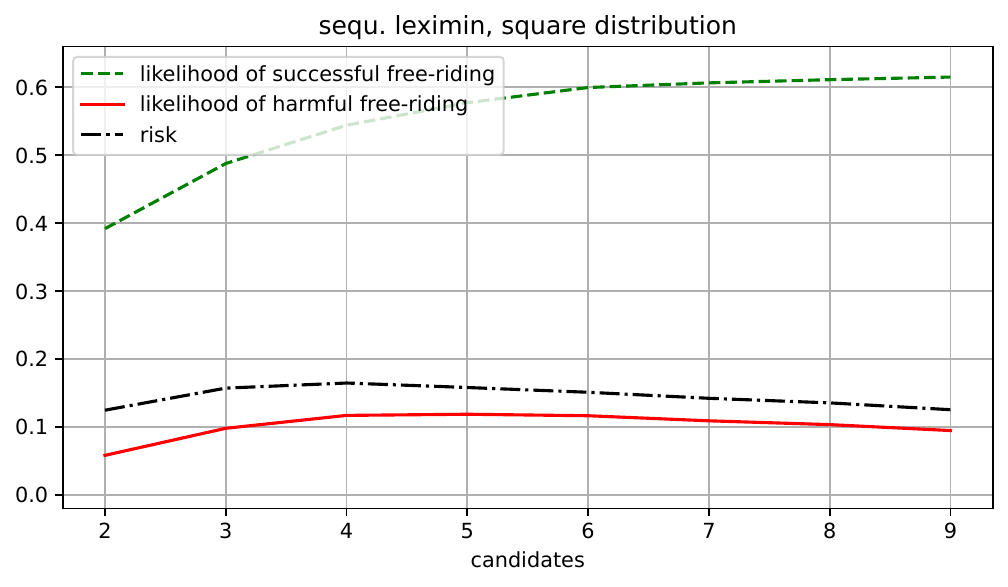}
  \caption{Varying the number of candidates (our default choice is $m=5$): more candidates per issue  increase the likelihood of free-riding.}\label{fig:exp-num_cands}
\end{figure*}

Finally, increasing the number of issues also increases the possibility of both successful and 
harmful free-riding, as some effects may only materialize in the long run. Note that also the risk of free-riding increases in a longer sequence of issues, in particular for sequential OWA methods.
\begin{figure*}
  \centering
  \includegraphics[width = 0.49\textwidth]{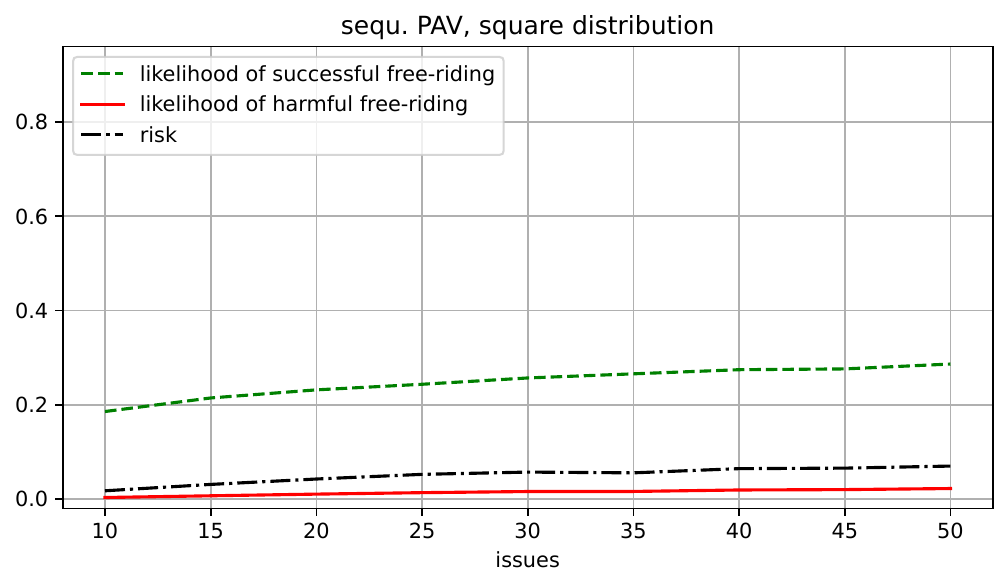}
  \includegraphics[width = 0.49\textwidth]{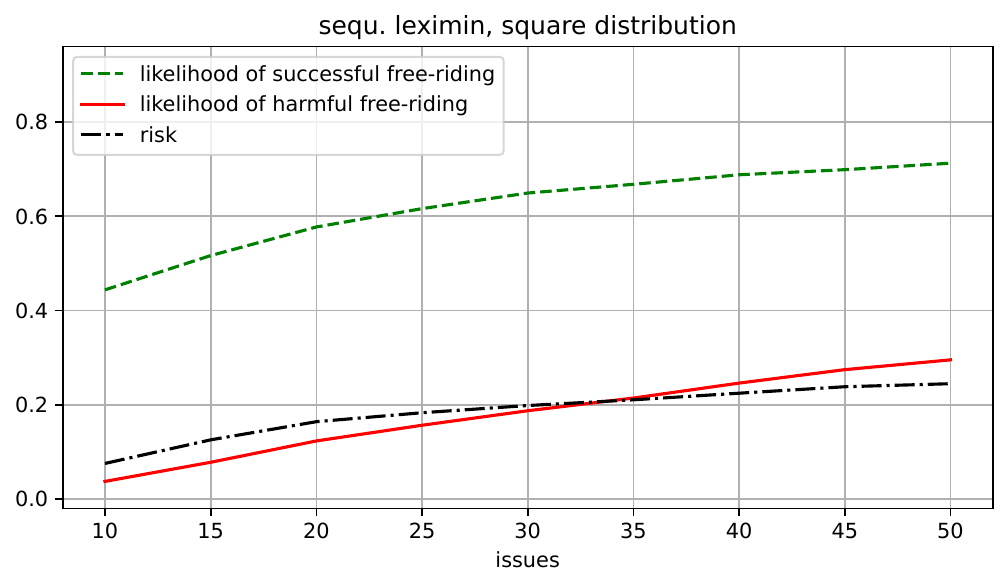} 
  \caption{Varying the number of issues (our default choice is $k=20$): a longer sequence of decisions increases the potential impact of free-riding.}
  \label{fig:exp-num_issues}
\end{figure*}

\subsubsection{Minority vs. majority free-riders}

\begin{figure*}
  \centering
  \includegraphics[width = 0.49\textwidth]{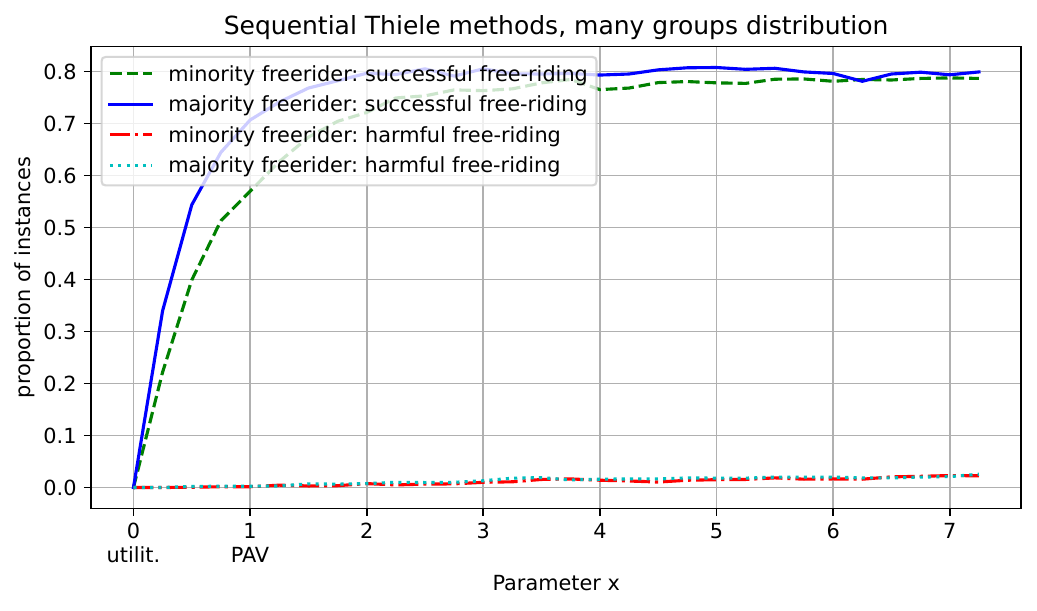}
  \includegraphics[width = 0.49\textwidth]{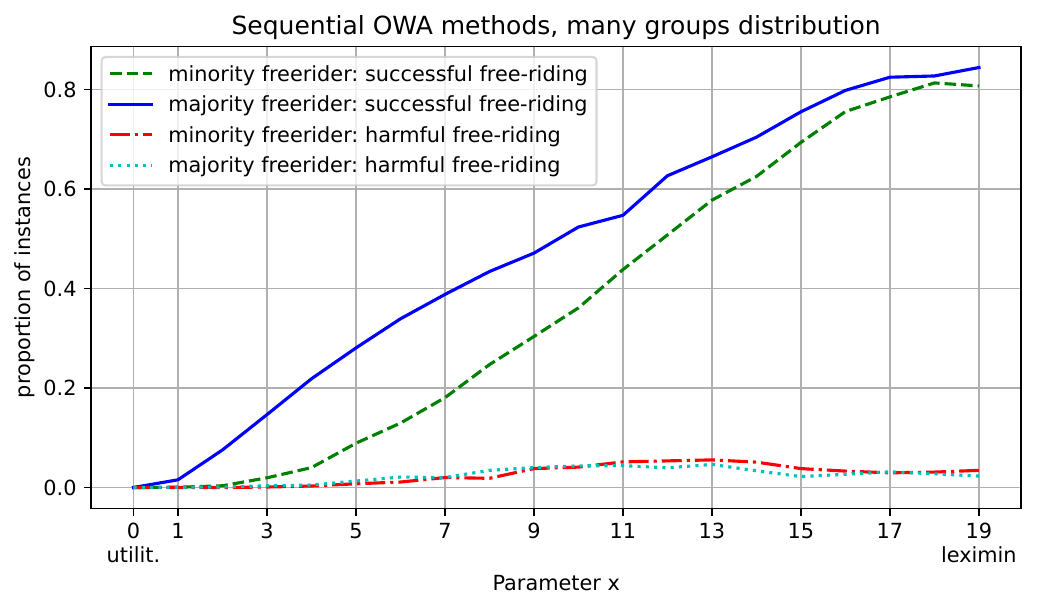} 
    \includegraphics[width = 0.49\textwidth]{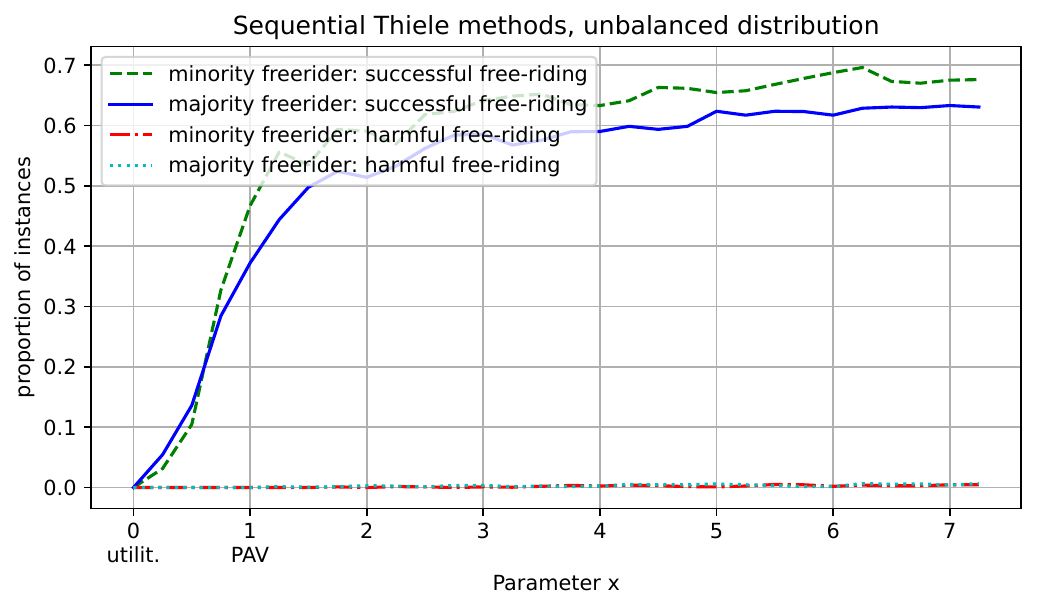}
  \includegraphics[width = 0.49\textwidth]{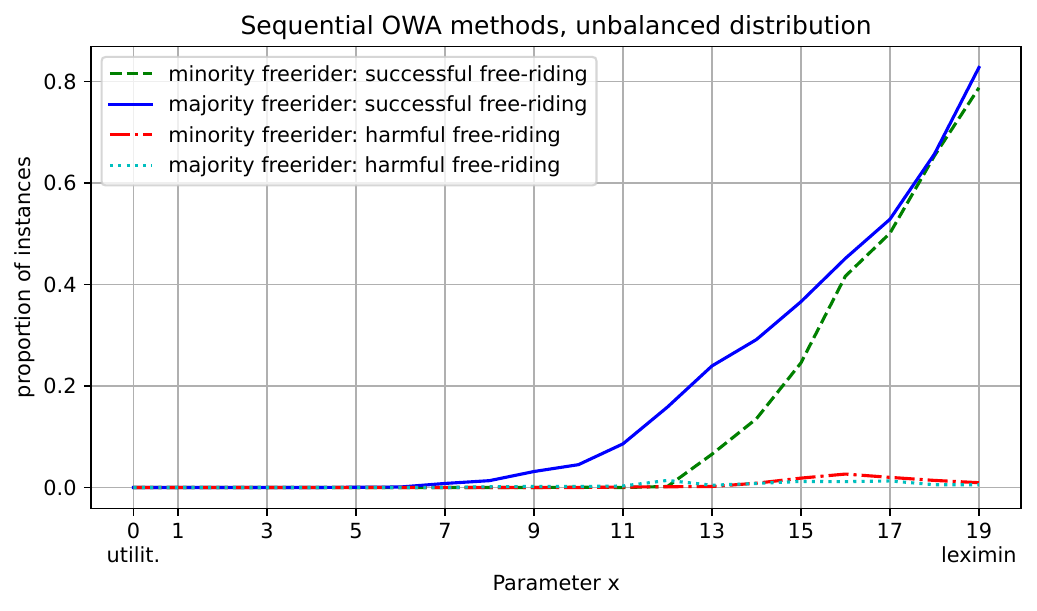} 
  \caption{The possibility of successful and harmful free-riding for minority and majority voters.}
  \label{fig:occ-minority-majority}
\end{figure*}

As a final experiment, we are studying the differences between free-riding voters belonging a minority or majority. Here, we use the \emph{repeated free-riding} model, as it is generally more beneficial for the free-rider (as we have discussed in Section~\ref{sec:results-possibility}) and thus magnifies differences between free-riding and truthful voters.
Let us first consider the many groups distribution. We randomly select one voter from the 40\% group as ``majority'' voter, and one from a 20\% group as ``minority'' voter.
Figure~\ref{fig:occ-minority-majority} (upper part for \emph{many groups}) shows that a majority voter is slightly more likely to successfully free-ride. This can be explained intuitively by the interpretation that free-riding majority voters form their own (one-voter) minority and thus are able to benefit from more egalitarian rules. However, when considering the average change in satisfaction for free-riding voters (Figure~\ref{fig:satisfaction-minority-majority},  upper part for \emph{many groups}), we see that the impact of free-riding is small: the change in satisfaction for a free-riding voter (both minority and minority) is small in comparison to their total satisfaction.
  
\begin{figure*} 
  \includegraphics[width = 0.49\textwidth]{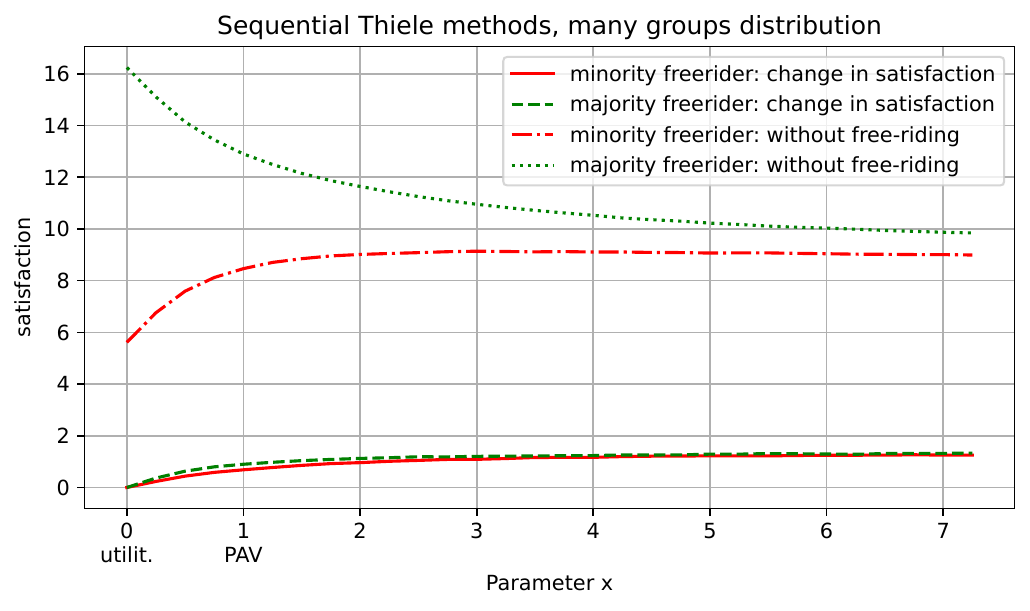}
  \includegraphics[width = 0.49\textwidth]{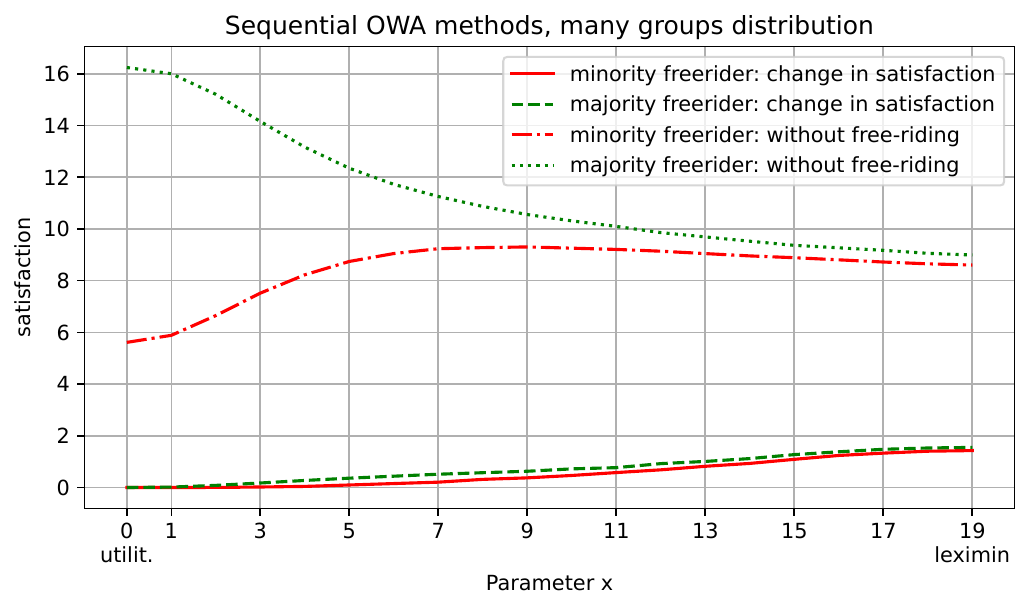}
  
    \includegraphics[width = 0.49\textwidth]{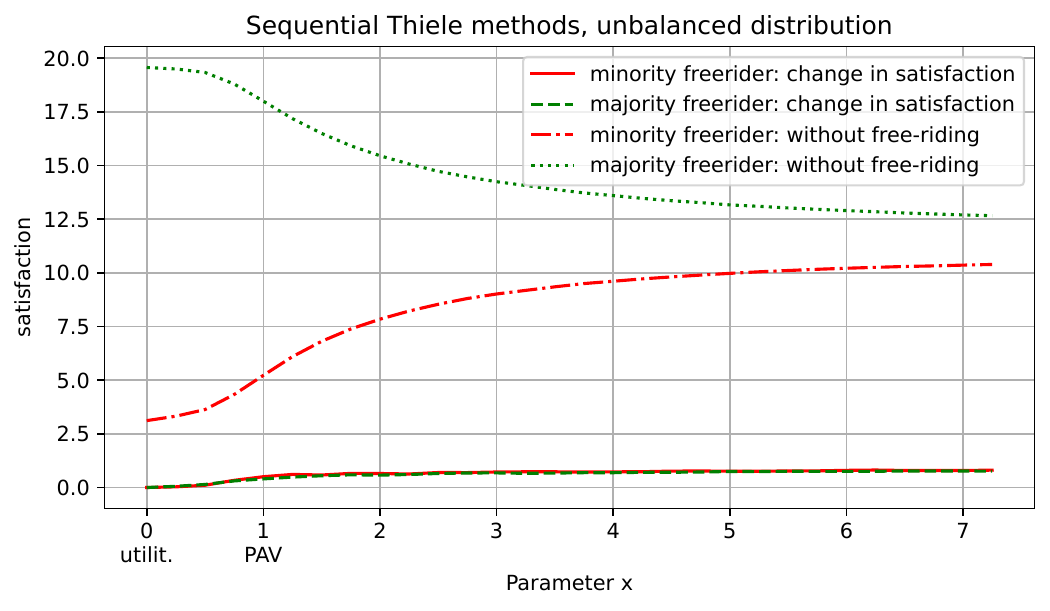}
  \includegraphics[width = 0.49\textwidth]{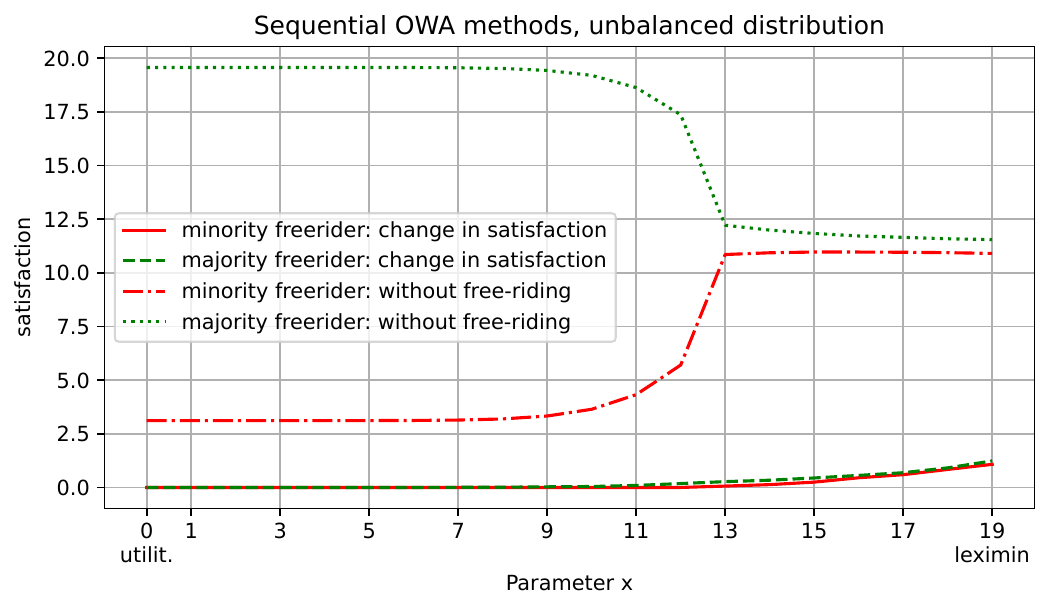}
  \caption{The change in satisfaction for minority and minority voters in relation to the average satisfaction of voters in either groups \emph{without} free-riding. We see that the change in satisfaction is similar for both groups (and relatively small).}
  \label{fig:satisfaction-minority-majority}
\end{figure*}

The results for the unbalanced distribution are very similar in terms of satisfaction (Figure~\ref{fig:satisfaction-minority-majority},  lower part). A noteworthy difference is that -- for Thiele rules -- minority voters have a slightly larger change of successful free-riding. This is probably due to the small weight that a single majority voter has within the majority block (16 voters in total).
A second interesting observation is that free-riding is here almost non-existent for close-to-utilitarian OWA rules (Figure~\ref{fig:occ-minority-majority}, lower part). This is due to the fact that for these rules the majority almost always dictates the decision and consequently free-riding is without effect.

Overall, we can conclude that for highly structured instances (such as the many groups and unbalanced distributions) free-riding is not an overly important phenomenon, due to its very limited impact on satisfaction. 
This statement, however, only holds for single-voter free-riding. In this paper, we are not studying models for groups of free-riders. In contrast, we expect that group free-riding has significant impact in these scenarios (depending on the chosen model), in particular if a majority subgroup forms a ``fake minority''.
Note that a group free-riding model will require additional assumptions on information exchange between free-riders (is group free-riding even possible?) and group cohesiveness (are group preferences identical or merely similar?).

\section{Free-Riding with Optimization-Based Rules}\label{sec:global}

We conclude this paper by briefly looking at the case of optimization-based rules. Recall that, in this setting, voters submit their preferences for all issues at the same time, and the result is computed simultaneously to maximize some global objective. We defer the proofs of this section to Appendix~\ref{app:global}.

We believe that free-riding in this setting, as a phenomenon, is quite distinct from the sequential case. In the sequential case, issues are presented one by one, and at each step, a voter may choose to free-ride \emph{in the current issue}. If a voter believes that a candidate she approves has a strong chance of victory (i.e., it has a large margin of support), she might be tempted to free-ride to immediately lower her (perceived) satisfaction. Our results show that it is computationally easy to free-ride but hard to determine whether it is possible to improve one's satisfaction -- even with full information. Moreover, in many settings, it is risky to do so with no information regarding future issues. 

In contrast, free-riding in the optimization-based setting is more akin to classical, unrestricted manipulation. Given the interconnections among issues, each voter here must, assuming full information, carefully consider how to change her ballots to decide whether free-riding is possible, let alone beneficial. Indeed, as we will see, the main barrier against free-riding in the optimization-based case is the hardness of determining the winner (or the consequences of changing one's ballot). Thus, it seems free-riding here loses its simple, ``greedy'' appeal. Restricting oneself to free-riding in this setting is not advantageous to a manipulator. While it reduces the strategy space, it is not significantly easier than arbitrary manipulation (i.e., to change the full ballot with the goal to increase satisfaction).

We start by studying the possibility and risk of free-riding. Given that our axioms from Section~\ref{sec:strategyproof} also apply to this setting,
our results already indicate that, for essentially all rules that include some form
of fairness towards minorities, free-riding is possible. However, not all the
rules we study satisfy those axioms (again, this is the case of opt-egalitarian). Thus, similarly to Theorem~\ref{thm:freeriding}, we integrate those results with the following.

\begin{restatable}{theorem}{freeridingglobal}
Every opt-Thiele and opt-OWA rule 
except the utilitarian rule can be manipulated by free-riding.\label{thm:freeriding-global}
\end{restatable}

Interestingly, in contrast to the sequential case, for all opt-Thiele rules and for at least opt-leximin
we find that the satisfaction of the free-riding voter can never decrease.

\begin{restatable}{proposition}{safefreeridingthiele}
Free-riding cannot reduce the satisfaction of the free-riding voter when
an opt-Thiele rule is used, but it can increase the satisfaction of the free-riding voter.
\end{restatable}

\begin{restatable}{proposition}{safefreeridingleximin}
Free-riding cannot reduce the satisfaction of the free-riding voter when
opt-leximin is used, but it can increase the satisfaction of the free-riding voter.
\end{restatable}

(It remains an open problem to generalize the latter result to other opt-OWA rules.) The previous two results seem to indicate that free-riding is safe in this setting: however, we will now show that free-riding is hard to perform \emph{at all}, i.e., it is hard to decide whether free-riding is possible in the first place. Hence, a manipulator has no reason to restrict herself to free-riding.

Towards our goal, we start from a more fundamental problem: outcome determination.
Indeed, any hypothetical free-rider needs to decide if, by voting dishonestly,
the outcome would be better than the ``truthful'' outcome.
To do so, she must be able to determine the outcome of an election.
If this step turns out to be intractable, then we already have
a first computational barrier against free-riding. Hence, we study the following problem:



\newcommand{\ODprob}{\textsc{$\calR$-Outcome-Determination}}
\newcommand{\RFprob}{\textsc{$\calR$-Free-Riding-Recognition}}
\newcommand{\genRFprob}{\textsc{Generalized-$\calR$-Free-Riding-Recognition}}
\newcommand{\bothRFprob}{\textsc{(Generalized-)$\calR$-Free-Riding Recognition}}

\begin{problem}
   \problemtitle{\ODprob}
   \probleminput{An election $\elec=\langle N, \bar A,\bar C\rangle$,
   an issue $i$ and a candidate $c\in C_i$.}
   \problemquestion{Does $c$ win in issue $i$ under $\calR$?}
\end{problem}

In the following, we assume that for all opt-$f$-Thiele rules,
$f$ is poly-time computable.\footnote{This is justified by the fact that all relevant values of 
$f$ can be computed ahead of time and stored in a look-up table.} Similarly, we assume that,
for a given opt-$\alpha$-OWA rule and $n$ voters, we can retrieve $\alpha^n$
in polynomial time. Now, we show that outcome determination is hard for both families of rules.

\begin{restatable}{theorem}{WDThiele}
\ODprob{} is \NP-hard for every opt-$f$-Thiele rule distinct from the utilitarian rule.\label{thm:WDThiele}
\end{restatable}

To prove the above we reduce from \textsc{CubicVertexCover} \cite{alimonti2000apx},
similarly as in the proof of Theorem 3 by \citet{skowron2016finding}. Their result
is similar to ours: they show NP-hardness for a large class of ``OWA rules'', which however correspond to what is commonly called ``Thiele rules''
in multi-winner voting~\cite{abcsurvey}.
We note that our result is not a consequence of theirs, since multi-winner voting is not a special case of our model. 

Moreover, we note that the hardness of opt-PAV was already shown by \citet[Theorem 5.1]{brill2022approval}, as the party-approval model used in this paper is a special case of ours. This implies that the hardness of opt-PAV holds even in the special case where the alternatives and preferences of the voters are constant across issues. Our result, in contrast, requires different preferences and candidates for each issue, but holds for all opt-$f$-Thiele rules except the utilitarian rule.

\begin{restatable}{theorem}{WDOWA}
\ODprob{} is \NP-hard for every opt-$\alpha$-OWA rule such that, for all $n$, $\alpha^n$ is nonincreasing and $\alpha_1>\alpha_n$.\label{thm:WDOWA}
\end{restatable}

Note that our result is related to Theorems 2 and 3 by \citet{amanatidis2015multiple}.
While their results hold even for
the special case of binary elections (i.e., in every issue there are exactly two alternatives), 
they only focus on a special subclass of OWA rules, i.e., given by vectors of the form $(1,\ \ldots,\ 1,\ 0,\ \ldots,\ 0)$.

In light of the above,
one could conclude that free-riding is unfeasible for optimization-based rules. Still,
one could argue -- especially since we use
worst-case complexity analysis -- that sometimes the fact that a certain candidate
wins can still be known (or guessed). For example, when a candidate receives
an extremely disproportionate support, or when
some external source (i.e.,
a polling agency having the computational power to solve \ODprob{})
communicates the projected winners. In this case, the manipulator would need to 
solve a slightly different, potentially easier, problem: \emph{Given that some candidate that I approve of wins in this specific issue,
can I deviate from my honest approval ballot, without making this candidate lose}?
Motivated by this, we study the following problem:

\begin{problem}
   \problemtitle{\RFprob}
   \probleminput{An election $\elec=\langle N, \bar A,\bar C\rangle$,
   an issue $i$, a candidate $c\in C_i$ such that $c\in\calR(\elec)_i$, and a voter $v$ such that $c\in A_i(v)$.}
   \problemquestion{Can $v$ free-ride in $\elec$ on issue $i$?}
\end{problem}

We define \genRFprob{} analogously. Luckily, the picture does not change: this problem is still computationally hard for essentially the same families of rules. 

\begin{restatable}{theorem}{RFThiele}
\bothRFprob{} is \NP-hard for every $f$-Thiele rule distinct from the utilitarian rule.\label{thm:RFThiele}
\end{restatable}

Hardness for OWA rules is split in two theorems: one showing \NP-hardness, the other \coNP-hardness.

\begin{restatable}{theorem}{RFOWA}
\bothRFprob{} is \NP-hard for every $\alpha$-OWA rule for which there is a $c\geq 3$ such that, for every $n\in\mathbb{N}$, there is a nonincreasing vector $\alpha$ of size $\ell$ (with $3n\leq\ell\leq cn$) such that $\alpha_1>\alpha_\ell$ and $\alpha_{3n}>0$.\label{thm:RFOWA}
\end{restatable}

\begin{restatable}{theorem}{RFOWAzero}
\bothRFprob{} is \coNP-hard for every $\alpha$-OWA rule for which there is a $c\geq2$ such that, for every $n\in\mathbb{N}$, there is a nonincreasing vector $\alpha$ of size $\ell$ (with $n<\ell\leq cn$) such that $\alpha_1>\alpha_\ell$ and $\alpha_{\ell-n+1}=0$.\label{thm:RFOWAzero}
\end{restatable}

Theorems~\ref{thm:RFThiele},~\ref{thm:RFOWA} and \ref{thm:RFOWAzero} strengthen our previous observations. We conclude that free-riding is generally unattractive for optimization-based rules, since the manipulator cannot even decide efficiently whether free-riding is possible.

\section{Discussion and Research Directions}
\label{sec:discussion}

In this paper, we have demonstrated that free-riding is an essentially unavoidable phenomenon in 
sequential multi-issue voting (cf.~Theorem~\ref{thm:many-issues} and other results in Section~\ref{sec:axiomatic}).
However, we have also identified computational challenges for voters attempting to assess the actual consequences of free-riding.
In addition, our numerical simulations show that the possibility of harmful free-riding is non-negligible, that is, a free-riding voter may end up with fewer satisfactory issues than without free-riding.
This holds especially for voters who will not or cannot free-ride whenever they have the chance to --
but only occasionally. Such reluctance may be caused by the social context in which free-riding occurs:
in small groups, it may be obvious to other group members that free-riding takes place and thus can entail negative social consequences.
Finally, our results show that the expected change in satisfaction of a free-riding voter is rather limited (albeit always positive in expectation).
All in all, we conclude that free-riding in real-world applications of multi-issue decision making may be less attractive to voters than
the theoretical possibility would initially suggest.

We conclude this paper with specific open problems.
First, we would like to point out that many of our hardness proofs use 
several candidates per issue. Do all of these results still hold for binary elections?
Second, our classification of sequential OWA rules with potentially harmful free-riding is not complete. Are there sequential OWA rules where free-riding is never harmful except for the utilitarian rule?
Third, Theorem~\ref{thm:seqFreeridingEgal} shows hardness of \freeridingProb{} for the sequential egalitarian rule. Can this result be extended to a larger class of OWA rules?
Fourth, there are further voting rules to be considered, such as rules based on \phragmen's ideas \cite{Phra95a,Brill2017Phragmens}.
Finally, coordinated free-riding of groups of voters may be much more impactful than single-voter free-riding. Studying this phenomenon provides an attractive opportunity to study group dynamics in voting related to coordination and strategy selection. 


\section*{Acknowledgements}
This research was funded by the Austrian Science Fund (FWF) research grants P31890, P32830, and 10.557766/COE12, by the Austrian Science Fund (FWF) and netidee SCIENCE [10.55776/\allowbreak PAT7221724], the Vienna Science and Technology Fund WWTF (10.47379/ICT23025 and ICT19-065), and the
European Union's Horizon 2020 research and innovation programme (under grant agreement \includegraphics[height=\fontcharht\font`\B]{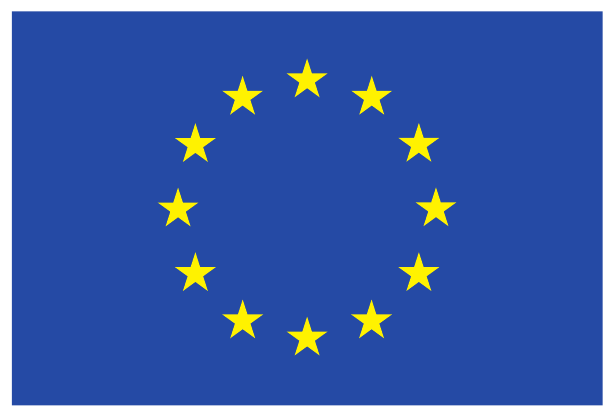}101034440)

\bibliography{literature}

\appendix

\section{Additional results: Possibility of Free-Riding}\label{app:possibility}

In this section, we provide additional results complementing those in Section~\ref{sec:axiomatic}. As mentioned in that section, we consider the following weaker variant of incentive for minorities:
\begin{itemize}
  \item \textbf{Weak incentive for minorities}: For $n\geq3$, there must be a number of issues $k$ such that the complete outcome of the $k$-simple $(k,n)$-election contains both $a$ and $b$ at least once.
\end{itemize}
We get results analogous to Theorems~\ref{thm:many-issues} and~\ref{thm:one-issue} via this weaker variant by strengthening the other axioms:
\begin{itemize}
  \item \textbf{Strong near-unanimity}: For $n\geq3$, the complete outcome of any $\ell$-simple $(k,n)$-election can contain $b$ at most $\ell-1$ times.
  \item \textbf{Strong monotonicity}: For any $\ell\leq k$, the complete outcome of any $\ell$-simple $(k,n)$-election cannot contain $b$ more times than the complete outcome of the $k$-simple $(k,n)$-election.
\end{itemize}

\begin{theorem}
  Any rule that satisfies strong near-unanimity and weak incentive for minorities can be manipulated by free-riding. If it additionally satisfies issue-wise unanimity and strong monotonicity, then it can be manipulated by free-riding in one issue.
\end{theorem}

\begin{proof}
  Consider a rule that satisfies strong near-unanimity and weak incentive for minorities. Consider any $n\geq3$ and let $k$ be the number of issues whose existence is prescribed by weak incentive for minorities. Focus on the $k$-simple $(k,n)$-election $\elec$ and let $I$ be the set of issues where $b$ is the outcome for $\elec$. Consider now a $|I|$-simple $(k,n)$-election $\elec^\prime$ where voter~$n$ approves of $b$ in every issue in $I$. By strong near-unanimity, the outcome here must contain $b$ at most $|I|-1$ times. Let $i$ be some issue where in $\elec^\prime$ voter~$n$ approves of $b$ but $a$ is in the outcome. But then voter~$n$ can free-ride in $\elec^\prime$ via $\elec$ in issues $[k]\setminus I$.

  Next, suppose the rule additionally satisfies issue-wise unanimity and strong monotonicity. Again, fix $n\geq3$ and let $k$ be the number of issues whose existence is prescribed by weak incentive for minorities. Let $s$ be the number of times where $b$ wins in the $k$-simple $(k,n)$-election. We know by weak incentive for minorities that $1\leq s\leq k-1$. Let $\ell$ be the minimal $\ell$ for which there is an $\ell$-simple $(k,n)$-election where $b$ wins exactly $s$ times, and let $\elec$ be such an election. Observe that $\ell>s$: indeed, we get $\ell\neq s$ via strong near-unanimity and $\ell\geq s$ via issue-wise unanimity. Furthermore, for all $(\ell-1)$-simple $(k,n)$-elections, $b$ wins at most $s-1$ times (by strong monotonicity). Let $i$ be some issue where $b$ wins in $\elec$; note that in this issue voter~$n$ must approve of $b$ (by issue-wise unanimity). Next, let $j$ be some issue distinct from $i$ where voter~$n$ approves of $b$ and where $a$ is in the outcome (such a $j$ exists as $\ell>s$). Let $\elec^\prime$ be the $(\ell-1)$-simple $(k,n)$-election obtained by letting voter~$n$ approve of $a$ instead of $b$ in issue $j$. Observe that voter~$n$ can free-ride in $\elec^\prime$ via $\elec$ in issue $j$, completing the proof.
\end{proof}

\section{Omitted proofs: Computational Complexity of Free-Riding}\label{app:complexity}

\seqFreeridingEgal*

\begin{proof}
First, note that, for the egalitarian rule, \freeridingProb{} is in \NP.
Indeed, given an insincere approval ballot for the manipulator, we can check whether it 
improves her satisfaction in polynomial time and whether it is a case of free-riding.

Now, we show hardness by a reduction from \textsc{3-SAT}.
Let $\phi$ be a 3-CNF with $n$ variables and $m$ clauses. We refer to the $j$-th clause as $C_j$.
We assume w.l.o.g.\ that $\phi$ is not satisfied
by setting all variables to false and that each clause contains 
exactly three literals.
We construct an instance of \freeridingProb{} with $2(n+1)$ voters and $5n+1$ rounds. In particular,
we will have two voters $v_i$ and $\bar v_i$ for each variable $x_i$, a voter $u$, and a distinguished voter $v$, the manipulator.

Let us start with the first $4n$ rounds. We subdivide this set of rounds into quadruples; that is,
the first quadruple consists of rounds $1$, $2$, $3$ and $4$, the second are the rounds $5$, $6$, $7$ and $8$, and so on.
We refer to the $j$-th round of quadruple $i$ as round $(i, j)$; for example, round $(2, 3)$ corresponds to round $7$.
In each round of every quadruple, there are two candidates $c$ and $\bar c$. Here (and in all subsequent rounds), we assume that if there is a tie $c$ loses.

Consider a generic quadruple $i$. In round $(i, 1)$, all voters vote for $\bar c$. In round $(i, 2)$, voters $v$ and $\bar v_i$ vote for $c$,
while voters $u$ and $v_i$ vote for $\bar c$. Everyone else approves of both. Furthermore, in round $(i, 3)$, voters $v$ and $u$ approve of $c$ and $\bar c$, respectively;
everyone else approves of both. Finally, in round $(i, 4)$, voter $v$ approves of both $c$ and $\bar c$, while everyone else approves only of $c$.

For each quadruple $i$, we claim that \emph{(i)} $v$ can free-ride (only) in round $(i,1)$ \emph{(ii)} if $v$ does not free-ride, the winners in this quadruple are $(\bar c,\bar c,c,c)$ and \emph{(iii)} if $v$ free-rides, the winners in this quadruple are $(\bar c,c,\bar c,\bar c)$. We show so by induction over the quadruples.

Consider quadruple $1$. Observe that, by tiebreaking, $\bar c$ wins in round $(1,1)$ (irrespectively of what $v$ votes for). Hence, here $v$ can free-ride. Suppose that she votes truthfully. Thus, in the next round, the minimal satisfaction if $c$ or $\bar c$ win is the same (namely, $1$). By tiebreaking, $\bar c$ wins. Now, up to here, every voter has satisfaction $2$, save for $v$ and $\bar v_1$, who have satisfaction $1$. In the next round, then, the minimal satisfaction if $c$ or $\bar c$ win is $2$ and $1$, respectively; hence, $c$ wins. Finally, in round $(1,4)$, the minimal satisfaction of $c$ winning is $3$, whereas the minimal satisfaction in case $\bar c$ wins is $2$ (namely, of voter $u$); hence, $c$ wins. With a similar line of reasoning, one can show that $(\bar c,c,\bar c,\bar c)$ is the result if $v$ does free-ride. Observe that $v$ can indeed free-ride only in round $(1,1)$: in every other round, either she is losing, or her vote would change the outcome. 

Now, suppose this property holds up to quadruple $i$, and consider quadruple $i+1$. If this holds, observe that no voter $v_j$ or $\bar v_j$ can have won fewer rounds than $v$ or $u$, and $v$ and $u$ won the same number of rounds. Therefore, at the beginning of each quadruple, $v$ and $u$ are among the voters with the lowest satisfaction. Let this minimal satisfaction be $s$. 

Again, observe that $v$ can free-ride in round $(i+1, 1)$. Suppose she votes truthfully. Then, she and $u$ will have the same satisfaction of $s+1$ in round $(i+1, 2)$, and $\bar c$ will win by tiebreaking. Next, in round $(i+1,3)$, $v$ will have the minimal satisfaction of $s+1$, and hence $c$ will win. Finally, in round $(i+1,4)$, if $c$ wins the minimal satisfaction will be $s+3$, whereas if $\bar c$ wins it will be $s+2$ (namely, of $u$); hence $c$ wins. With similar arguments, we could show that $(\bar c,c,\bar c,\bar c)$ is the result if $v$ does free-ride. Observe that $v$ can indeed free-ride only in round $(i+1,1)$: in every other round, either she is losing, or her vote would change the outcome.

Now, let us consider round $4n+1$ to round $5n-1$. From the previous discussion we know that, in each quadruple $i$, all voters $v_j$ and $\bar v_j$ (with $j\neq i$) win the same amount of rounds (either $3$ or $4$, depending on whether $v$ free-rides or not). Let this number be $\ell_i$. Furthermore, one voter in $v_i$ and $\bar v_i$ wins $\ell_i$ rounds, while the other wins $\ell_i-1$ rounds. Finally, both $v$ and $u$ won exactly $\ell_i-1$ rounds. Thus, we can partition the voters $v_i$ and $\bar v_i$ into two groups, with a satisfaction differing of exactly $1$ point. Let $s$ be the satisfaction of the voters in the group with the lowest satisfaction. Observe that the satisfaction of $v$ and $u$ will be exactly $s+1-n$, because in each quadruple $i$, both $v$ and $u$ lose exactly one round (compared to the voters $v_j$ and $\bar v_j$ with $i\neq j$). So then, in each of the rounds from $4n+1$ to $5n-1$, there are two candidates: $c$ and $\bar c$. Here, $v$ approves of $\bar c$, $u$ of both candidates, and everyone else only of $c$. Observe that $v$ and $u$ win every such round, as $v$ always has a strictly lower satisfaction than the rest of the voters, and $u$ always approves of all candidates. Furthermore, $v$ can't free-ride here: as she always has the lowest satisfaction, she is always pivotal, and hence her vote decides the outcome. Thus, after these rounds, both $v$ and $u$ will have satisfaction $s$.

In round $5n$, there are two candidates: $\bar c$ and $c$. Here, $v$ votes for $\bar c$, whereas every one else votes for $c$. If either candidate wins, the minimal satisfaction will be $s$, and thus $\bar c$ wins by tiebreaking. Furthermore, $v$ cannot free-ride: if she were to vote for $c$, then $c$ would win.

Finally, in round $5n+1$, we know that $v$ has satisfaction of $s+1$ and $u$ of $s$. Furthermore, if $v_i$ has satisfaction $s$ if $c$ won in round $(i,1)$ and $s+1$ otherwise (and similarly for $\bar v_i$ and $\bar c$). In this round, there are $m+1$ candidates, namely $c,c_1,\dots,c_m$. Here, $u$ approves of all candidates, voter $v_i$ (resp. $\bar v_i$) approves of $c$ and of all candidates $c_j$ such that $x_i\not\in C_j$ (resp. $\bar x_i\not\in C_j$). Finally, voter $v$ approves of $c$.

Observe that we can interpret $c$ winning in round $(i,2)$ as setting $x_i$ to true (and conversely for $\bar c$); we claim that this assignment satisfies $\phi$ if and only if $c$ wins in this final round. To see this, observe that if $c$ wins the minimal satisfaction will be $s+1$ (all voters approve of $c$). Now, consider a clause $C_j$ and its candidate $c_j$. If all literals in $C_j$ are unsatisfied, then the corresponding voters have all (up to this round) satisfaction $s+1$. Hence, if $c_j$ would win, the minimal satisfaction will be at least $s+1$, as all other voters (except $v$) approve of it, and $v$ has satisfaction at least $s+1$. Hence, $c_j$ would win by tiebreaking. Conversely, if at least one literal in $C_j$ is satisfied, there is at least one voter with satisfaction $s$ that does not vote for $c_j$: hence, $c_j$ loses against $c$. Our claim follows.

Now, observe that, if $v$ were to always vote truthfully, her true satisfaction would be $4n$ (she would win three rounds per quadruple, all rounds from $4n+1$ to $5n$, and lose the last round, by the assumption that setting all variables to false does not satisfy $\phi$). Observe also that, as we discussed before, she can only free-ride in the first round of every quadruple. Therefore, the only way she can raise her satisfaction to $4n+1$ is by winning the last round (observe that if she free-rides in some quadruple, she still truly wins three rounds). To do so, she has to force a satisfying assignment for $\phi$ by free-riding. It follows that $v$ can free-ride if and only if $\phi$ is satisfiable, and so we are done.\end{proof}

\seqFreeridingOWAGen*

\begin{proof}
Fix an $\alpha$-OWA rule $\calR$ satisfying the conditions of the theorem. First, notice that \genFreeridingProb{} is in \NP,
as we can guess an insincere approval ballot for the manipulator and check whether it 
improves her satisfaction (and is an instance of generalized free-riding) in polynomial time.

Next, we show hardness by a reduction from \textsc{3-SAT}. 
Let $\phi$ be a 3-CNF with $n$ variables and $m$ clauses. We refer to the $j$-th clause as $C_j$.
We assume w.l.o.g.\ that $\phi$ is not satisfied
by setting all variables to false and that each clause contains 
exactly three literals.
We will construct an instance of \genFreeridingProb{} with $2n+5$ voters.
More specifically, there are two voters $v_i$ and $\bar v_i$ for each variable $x_i$, four 
voters $u_1,\dots,u_4$, plus one distinguished voter $v$ who will try to manipulate. 

Given the weight vector $\alpha=(\alpha_1,\dots,\alpha_{2n+5})$, we distinguish three (not necessarily exclusive) cases: (1) $\alpha_2>\alpha_{n+5}$, (2) $\alpha_{n+1}>\alpha_{2n+5}$, and (3) $\alpha_{2}=\alpha_{2n+5}$. Since $\alpha_1>\alpha_{2n+5}$, at least one case must be true. In the following,
we will give a different reduction for each of the three cases.

\paragraph{First case: $\alpha_2>\alpha_{n+5}$.}
We construct an instance with $n+2$ rounds as follows.
In the first $n+1$ rounds, there are two candidates: $c$ and $\bar c$. In
the last round, there are $m+1$ candidates, namely $c,c_1,\dots,c_m$. We assume that, in the case of ties, $c$ always loses.

In the first round, everybody votes for $\bar c$ except for $v$ and $u_1$, who vote for $c$. Here, $\bar c$ wins by tiebreaking, and $v$ cannot manipulate.

In each round $i$ with $i\in\{2,\dots,n+1\}$, voter $v_i$ votes for candidate $\bar c$, voter $\bar v_i$ for candidate $c$; everyone else votes for both candidates. We claim that $v$ can manipulate in every such round $i$ and force the win of either $c$ or $\bar c$. We show so by induction. In round $2$,
suppose that $v$ votes only for $c$ (the case where she votes for $\bar c$ is analogous). Then, if $c$ were to win, the satisfaction vector would be of form $(1,1,1,2,\dots,2)$ (everyone but $v_1$ wins). On the other hand, if $\bar c$ wins, then it would be of form $(0,1,1,2,\dots,2)$ (everyone but $\bar v_1$ and $v$ win). Hence, $c$ wins. Observe that if $v$ votes truthfully, $\bar c$ wins by tiebreaking. Now, suppose this holds up to a round $i$. Before round $i+1$, $v$ has won $i-1$ rounds (all but the first one), whereas voters $u_1,\dots,u_4$, as well as any pair of voters $v_j,\bar v_j$ (with $j\geq i$) have won $i$ rounds. Furthermore, for every pair $v_j,\bar v_j$ (with $j<i$), exactly one voter won $i-1$ rounds while the other $i$ rounds. Suppose again that $v$ votes for $c$ (the case where she votes for $\bar c$ is similar). Then, observe that, if $c$ or $\bar c$ win,s the satisfaction vectors (excluding $v$) would be completely symmetric (and every voter would have at least a score of $i$); however, if $\bar c$ wins, $v$ would have a satisfaction of $i-1$, whereas if $c$ wins, she would get a satisfaction of $i$. Hence, since $\alpha_1>0$, here $c$ wins. Observe again that if $v$ votes truthfully, then $\bar c$ wins.

Consider the final round. Up to here, $v$ and $u_1$ have won $n$ rounds (they lost the first round), while $u_2,u_3,u_4$ have won $n+1$ rounds. Furthermore, every voter $v_i$ has won $n$ rounds if $c$ won in round $i+1$ and $n+1$ times otherwise (and conversely for $\bar v_i$ and $\bar c$). In this round, voters $u_1,\dots,u_4$ approve of all candidates but $c$, voter $v_i$ (resp. $\bar v_i$) approves of $c$ and every candidate $c_j$ such that $x_i\not\in C_j$ (resp. $\bar x_i\not\in C_j$). Finally, voter $v$ approves of $c$. Observe that we can interpret $c$ winning in round $i+1$ as setting $x_i$ to true, and $\bar c$ winning as setting $x_i$ to false. We claim that $c$ wins in the last round if and only if this assignment satisfies $\phi$. To see this, consider that, if $c$ were to win, the satisfaction vector would be:
\begin{equation*}
  (n,\underbrace{n+1,\ldots,n+1}_{n+4\ \text{times}},\underbrace{n+2,\ldots,n+2}_{n\ \text{times}}).
\end{equation*}
Let us call this vector $s$. Consider now a candidate $c_j$ and its corresponding clause $C_j$. If all three of its literals are unsatisfied, then the corresponding voters all have satisfaction $n+1$. Hence, if $c_j$ were to win in this case, the satisfaction vector would again be exactly $s$. By our assumptions on tiebreaking, here $c_j$ would win against $c$. Furthermore, suppose that either one, two, or three of the literals have been satisfied. Then, the vectors are, respectively:
\begin{align*}
  &(n,n,\underbrace{n+1,\ldots,n+1}_{n+2\ \text{times}},\underbrace{n+2,\ldots,n+2}_{n+1\ \text{times}})\\
  &(n,n,n,\underbrace{n+1,\ldots,n+1}_{n\ \text{times}},\underbrace{n+2,\ldots,n+2}_{n+2\ \text{times}})\\
  &(n,n,n,n,\underbrace{n+1,\ldots,n+1}_{n-2\ \text{times}},\underbrace{n+2,\ldots,n+2}_{n+3\ \text{times}})\\
\end{align*}
Let these vectors be $s_1$, $s_2$ and $s_3$, respectively. One can show that if $\alpha_{2}>\alpha_{n+5}$ the dot product between $\alpha$ and each of these three vectors would be strictly lower than the dot product between $\alpha$ and $s$. For example:
\begin{multline*}
    s\cdot\alpha >  s_1\cdot\alpha
    \iff \alpha_1n+\left(\sum_{i=2}^{n+5}\alpha_i(n+1)\right)+\left(\sum_{i=n+6}^{2n+5}\alpha_i(n+2)\right) >\\
    (\alpha_1+\alpha_2)n+\left(\sum_{i=3}^{n+4}\alpha_i(n+1)\right)+\left(\sum_{i=n+5}^{2n+5}\alpha_i(n+2)\right) \iff \\
    (\alpha_2+\alpha_{n+5})(n+1) > \alpha_2n+\alpha_{n+5}(n+2)
    \iff \alpha_{2} > \alpha_{n+5}.
\end{multline*}
The other two cases are similar. Hence, if $C_j$ is satisfied, candidate $c_j$ cannot win against $c$. Consequently, if all clauses are satisfied, candidate $c$ wins.

Now, if $c$ wins in the last round, then the satisfaction of $v$ would be $n+1$; if $c$ loses, it would be $n$. Notice also that $v$ cannot raise her satisfaction by manipulating in the final round. Furthermore, if $v$ always submits her true preferences, then by tiebreaking $\bar c$ would win in every round $i$ with $i\in\{2,\dots,n+1\}$. By assumption, this would not satisfy $\phi$, and hence $c$ would not win in the last round. Therefore, $v$ has an incentive to manipulate in these rounds to try and choose a satisfying assignment for $\phi$. It follows that $v$ can manipulate via generalized free-riding if and only if $\phi$ is satisfiable, so we are done.

\paragraph{Second case: $\alpha_{n+1}>\alpha_{2n+5}$.}

We construct an instance with $n+2$ rounds as follows.
In the first $n+1$ rounds, there are two candidates, $c$ and $\bar c$. In the last round, there are $m+1$ candidates, namely $c,c_1,\dots,c_m$. We assume that, in the case of ties, $c$ always loses.

In the first round, $v,u_1,\dots,u_4$ approve of $c$, whereas everyone else approves of $\bar c$.

In each round $i$ with $i\in\{2,\dots,n+1\}$, voter $v_i$ votes for candidate $c$, voter $\bar v_i$ for candidate $\bar c$, and everyone else votes for both candidates.

In the final round, voters $v,u_1,u_2$ approve of $c$, voter $v_i$ (resp. $\bar v_i$) approves of candidate $c_j$ if $x_i\in C_j$ (resp. $\bar x_i\in C_j$). Finally, voters $u_3$ and $u_4$ approve of all candidates.

We claim that $\emph{(i)}$ voters $v,u_1,\dots,u_4$ lose in the first round, and that in the following $n$ rounds $\emph{(ii)}$ candidate $\bar c$ wins if $v$ votes truthfully and \emph{(iii)} $v$ can force the win of either $c$ or $\bar c$ by manipulating. The arguments for this are essentially the same as in the first case.

Now, consider the last round. Up to here, $v,u_1,\dots,u_4$ have won $n$ rounds (all but the first). Furthermore, every voter $v_i$ has won $n+1$ rounds if $c$ won in round $i$ and $n$ times otherwise (and conversely for $\bar v_i$ and $\bar c$). We interpret again $c$ winning in round $i$ as setting $x_i$ to true, and $\bar c$ winning as setting $x_i$ to false. We claim that $c$ wins in the last round if and only if this assignment satisfies $\phi$. To see this, consider that, if $c$ were to win, the satisfaction vector would be:
\begin{equation*}
  (\underbrace{n,\ldots,n}_{n\ \text{times}},\underbrace{n+1,\ldots,n+1}_{n+5\ \text{times}}).
\end{equation*}
Let us call this vector $s$. Consider now a candidate $c_j$ and its corresponding clause $C_j$. If all three of its literals are unsatisfied, then the corresponding voters all have satisfaction $n$. Hence, if $c_j$ were to win in this case, the satisfaction vector would again be exactly $s$. By our assumptions on tiebreaking, here $c_j$ would win against $c$. Furthermore, suppose that either one, two, or three of the literals have been satisfied. Then, the vectors are, respectively:
\begin{align*}
  &(\underbrace{n,\ldots,n}_{n+1\ \text{times}},\underbrace{n+1,\ldots,n+1}_{n+3\ \text{times}},n+2)\\
  &(\underbrace{n,\ldots,n}_{n+2\ \text{times}},\underbrace{n+1,\ldots,n+1}_{n+1\ \text{times}},n+2,n+2)\\
  &(\underbrace{n,\ldots,n}_{n+3\ \text{times}},\underbrace{n+1,\ldots,n+1}_{n-1\ \text{times}},n+2,n+2,n+2)\\
\end{align*}
One can show that if $\alpha_{n+1}>\alpha_{2n+5}$ the dot product between $\alpha$ and each of these three vectors would be strictly lower than the dot product between $\alpha$ and $s$. The computation is analogous to the one we did in the first case. Hence, if $C_j$ is satisfied, candidate $c_j$ cannot win against $c$. Consequently, if all clauses are satisfied, candidate $c$ wins. With similar arguments as before, we conclude that here $v$ can manipulate if and only if $\phi$ is satisfiable.

\paragraph{Third case: $\alpha_{2}=\alpha_{2n+5}$.}

We construct an instance with $n+3$ rounds as follows. The rounds $2,\dots,n+1$ and the last round are equal to the first case. The first round is almost identical, save for the fact that $u_1$ also votes for $\bar c$. 

Hence, we know that $v$ loses in the first round and that she can manipulate in all the following rounds to force a win of either $c$ or $\bar c$. Let us focus on round $n+2$. We will design this round to make sure that only $v$ wins, and that she cannot manipulate via generalized free-riding. We distinguish two sub-cases:

\begin{enumerate}
  \item $\alpha_2=\dots=\alpha_{2n+5}=0$. Here, everyone votes for $c$, save for $v$, who votes for $\bar c$. There are no other candidates. Observe that, in case either $c$ or $\bar c$ wins here, the minimal satisfaction will be $n$ in both cases; $\bar c$ wins by tiebreaking. Furthermore, were $v$ to vote for $c$, then $c$ would win (as the minimal satisfaction for $c$ winning would raise to $n+1$).
  \item $\alpha_2=\dots=\alpha_{2n+5}>0$. Here, there is one candidate $c_{v^*}$ for every voter $v^*\in N$, and we assume that in case of ties $c_v$ wins. Furthermore, we assume that all voters vote for their voter-candidate. We show that here all candidates receive the same score. Consider any two candidates $c_y$ and $c_z$. Let $y=(y_1,\dots,y_{2n+5})$ and $z=(z_1,\dots,z_{2n+5})$ be the satisfaction vectors corresponding to $c_y$ and $c_z$ winning, respectively. For both $c_y$ and $c_z$, there is surely at least one voter with satisfaction $n$ that disapproves of them; hence, $y_1=z_1=n$. Furthermore, there is exactly one voter approving each candidate, and hence $\sum_{i=2}^{2n+5}y_i=\sum_{i=2}^{2n+5}z_i$. These two facts, together with the fact that $\alpha_2=\dots=\alpha_{2n+5}$, imply that $\alpha\cdot y = \alpha\cdot z$. Hence, every two candidates receive the same score: by tiebreaking, $c_v$ wins. Now, notice that, if $v$ approves of any other candidate $c_{v^*}$ distinct from $c_v$, then $c_{v^*}$ will receive a strictly greater score than any other candidate (as now two voters approve of it).
\end{enumerate}

Now, consider the last round. Up to here, $v,u_1,\dots,u_4$ won $n+1$ rounds. Furthermore, every voter $v_i$ has won $n+1$ rounds if $c$ won in round $i$ and $n$ times otherwise (and conversely for $\bar v_i$ and $\bar c$). We interpret again $c$ winning in round $i$ as setting $x_i$ to true, and $\bar c$ winning as setting $x_i$ to false. We claim that $c$ wins in the last round if and only if this assignment satisfies $\phi$. To see this, consider that, if $c$ were to win, the satisfaction vector would be:
\begin{equation*}
  (\underbrace{n+1,\ldots,n+1}_{n+4\ \text{times}},\underbrace{n+2,\ldots,n+2}_{n+1\ \text{times}}).
\end{equation*}
Let us call this vector $s$. Consider now a candidate $c_j$ and its corresponding clause $C_j$. If all three of its literals are unsatisfied, then the corresponding voters all have satisfaction $n+1$. Hence, if $c_j$ were to win in this case, the satisfaction vector would again be exactly $s$. By our assumptions on tiebreaking, here $c_j$ would win against $c$. Furthermore, suppose that either one, two, or three of the literals have been satisfied. Then, the vectors are, respectively:
\begin{align*}
  &(n,\underbrace{n+1,\ldots,n+1}_{n+2\ \text{times}},\underbrace{n+2,\ldots,n+2}_{n+2\ \text{times}})\\
  &(n,n,\underbrace{n+1,\ldots,n+1}_{n\ \text{times}},\underbrace{n+2,\ldots,n+2}_{n+3\ \text{times}})\\
  &(n,n,n,\underbrace{n+1,\ldots,n+1}_{n-2\ \text{times}},\underbrace{n+2,\ldots,n+2}_{n+4\ \text{times}})\\
\end{align*}
One can show that if $\alpha_{1}>\alpha_{n+4}$ (which is implied by $\alpha_1>\alpha_{2n+5}$ and $\alpha_{2}=\alpha_{4}=\alpha_{2n+5}$) the dot product between $\alpha$ and each of these three vectors would be strictly lower than the dot product between $\alpha$ and $s$. The computation is analogous to the one we did in the first case. Hence, if $C_j$ is satisfied, candidate $c_j$ cannot win against $c$. Consequently, if all clauses are satisfied, candidate $c$ wins. With similar arguments as the first case, we conclude that here $v$ can manipulate if and only if $\phi$ is satisfiable.\\

This concludes the proof.\end{proof}

\section{Omitted proofs: Free-Riding with Optimization-Based Rules}\label{app:global}

\freeridingglobal*

\begin{proof}
Let $\calR$ be an opt-$f$-Thiele Rule different from the utilitarian rule.
Then, there exists a $k$ such that $f(k-1)> f(k)$. 
Consider a $k+1$ issue election with four voters and two candidates $a$ and $b$ 
such that for the first $k$ issues all voters only approve candidate $a$.
Moreover, on issue $k+1$ voters $1$ and $2$ approve $b$ while voter $3$ and $4$
approve $a$. Assume further that $a$ is preferred to $b$ in the tiebreaking order.
Clearly, selecting $b$ in one of the first $k$ rounds just reduces the score of all voters,
hence in any optimal outcome $a$ wins in the first $k$ issues. Letting $a$ or $b$ win on
issue $k+1$ increases the score of the outcome by $2f(k)$ for both candidates. We can assume that
$a$ wins by tiebreaking.
We claim that voter $1$ can manipulate by changing her vote in one of the first $k$ issues
to $\{b\}$. Assume that $1$ manipulates on issue $k$. Then, it is still clearly best 
to let $a$ win in the first $k-1$ issues. This leads to score of $S:= 4\sum_{i=1}^{k-1} f(i)$. 
Let us now consider the score of four possible outcomes on issue $k$ and $k+1$.
The outcome $(a, \dots, a,a)$ has score of $S + 3f(k-1) + 2f(k)$,
the outcome $(a, \dots, a,b)$ has score of $S + 3f(k-1) + f(k) + f(k-1)$,
the outcome $(a, \dots, b,a)$ has score of $S + f(k-1) + 2f(k-1)$
and the outcome $(a, \dots, b,b)$ has score of $S + f(k-1) + f(k) +f(k-1)$.
As, $f(i-1) > f(k)$ this implies that $(a, \dots, a,b)$ is the winning outcome.
Therefore, voter $1$ did free-ride successfully.

Now, let $\calR$ be an opt-OWA rule that is not the utilitarian rule. Then there exists a $n$
for which the vector $\alpha$ for $k$ voters satisfies $\alpha_{1} > \alpha_{n}$.
Clearly, $n\geq 2$.
Consider an election with $2$ issues and $n$ voters. In each issue there are $n$ candidates
$a_1, \dots a_n$. In the first issue, voters~$1$ and $2$ approve $a_1$. Every other voter $i\in\{3,\dots, n\}$
approves $a_i$. In the second issue voter $1$ approves $a_1$, voter $2$ approves $a_2$ and
all other voters approve both $a_1$ and $a_2$.
We assume that candidates with a lower index are preferred by the tiebreaking,
which is applied lexicographically.
Selecting a candidate other than $a_1$ in the first issue leads to satisfaction vector $(0,1,\ldots,1,2)$,
independently of whether $a_1$ or $a_2$ is selected in issue $2$.
On the other hand, selecting $a_1$ in issue $1$ leads to satisfaction vector $(1,1,\ldots,1,2)$
independently of whether $a_1$ or $a_2$ is selected in issue $2$.
This means $(a_1,a_1)$ and $(a_1,a_2)$ lead to the highest OWA score.
By tiebreaking, $(a_1,a_1)$ wins.
Now, we claim that voter $2$ can free-ride by approving $a_2$ instead of $a_1$ in the first issue.
Assume first, that a candidate other than $a_1$ or $a_2$ is selected in the first issue.
This still leads to a satisfaction vector of $(0,1,\ldots,1,,2)$,
independently of whether $a_1$ or $a_2$ is selected in issue $2$.
Choosing $a_1$ in both issues leads to a satisfaction vector of $(0,1,\ldots,1,2)$.
Choosing $a_1$ in issue $1$ and $a_2$ in issue $2$ leads to vector
$(1,\ldots,1)$.
Choosing $a_2$ both times or first $a_2$ and then $a_1$ is symmetric.
As $\alpha_{1} > \alpha_{n}$ we know that 
\[
  \alpha \cdot (\underbrace{1,\ldots,1}_{n\ \text{times}})= \sum_{i=1}^n \alpha_{i} > 
  \alpha_{n} - \alpha_{1} + \sum_{i=1}^n \alpha_{i} =
  \alpha \cdot (0,\underbrace{1,\ldots,1}_{n-2\ \text{times}}, 2).
\]
It follows that $(a_1,a_2)$ and $(a_2,a_1)$ are the outcomes maximizing the OWA score.
By tiebreaking, $(a_1,a_2)$ is the winning outcome. It follows that $2$ did successfully 
free-ride.
\end{proof}

\safefreeridingthiele*

\begin{proof}

Fix the opt-$f$-Thiele rule. It follows directly from Theorem~\ref{thm:freeriding-global} that free-riding can increase
the satisfaction of the free-riding voter.
Let us show that it can never decrease the satisfaction of the free-riding voter.
Let $\elec$ be an election, $\bar w$ be the outcome
of $\elec$ under the $f$-Thiele rule, and consider a voter $v^*$ such that $v^*$ can 
free-ride in issues $I\subset[k]$. Finally, let $\elec^*$ be the election after $v^*$ free-rides 
and $\bar w^*$ the outcome of $\elec^*$ under the same $f$-Thiele rule.
Now, as $v^*$ free-rides, i.e., the winners in issues~$I$ are the same in $\bar w$ and $\bar w^*$,
we know that $v^*$ approves the winners of issues $I$ in her honest ballot in $\elec$ and 
does not approve the winners of issues $I$ in her manipulated ballot in $\elec^*$.
It follows the satisfaction of $v^*$ with $\bar w^*$ resp.\ $\bar w$ in $\elec$
is higher by one than in $\elec^*$, i.e.,
\begin{align}
\sat_{\elec^*}(v^*, \bar w^*) &= \sat_{\elec}(v^*, \bar w^*) - |I| \quad\text{as well as }\nonumber\\
\sat_{\elec^*}(v^*, \bar w) &= \sat_{\elec}(v^*, \bar w) - |I|.\label{eq:3}
\end{align}
All other voters submit the same ballot in $\elec$ and $\elec^*$.
Hence, for all $v \neq v^*$ we have 
\begin{align}
\sat_{\elec^*}(v, \bar w^*) &= \sat_{\elec}(v, \bar w^*) \quad\text{as well as  }\nonumber\\
\sat_{\elec^*}(v, \bar w) &= \sat_{\elec}(v, \bar w). \label{eq:4}
\end{align}
Now assume for the sake of a contradiction that 
$\sat_{\elec}(v^*, \bar w) > \sat_{\elec}(v^*, \bar w^*)$, i.e.,
free-riding led to a lower satisfaction for $v^*$ with respect to her honest ballot.

As $\bar w^*$ is the winning outcome of $\elec^*$, we know that
\[
  \sum_{v\in N} \smashoperator[r]{\sum_{i=1}^{\sat_{\elec^*}(v, \bar w^*)}} f(i) \quad>\quad \sum_{v\in N} \smashoperator[r]{\sum_{i=1}^{\sat_{\elec^*}(v, \bar w)}} f(i).
\]
Note that the inequality is strict because, if the two sides were equal, it would mean that $w^*$ is preferred over $w$ by the tiebreaking and that, by Equations~\eqref{eq:3} and~\eqref{eq:4}, that $w$ and $w^*$ would have the same score also in election $\elec$. Hence, $w^*$ would win in $\elec$, contradicting our assumptions. Next, by the previous inequality and Equations~\eqref{eq:3} and~\eqref{eq:4},
\[
  \smashoperator{\sum_{i=1}^{\sat_{\elec}(v^*, \bar w^*)-|I|}} f(i) + \smashoperator{\sum_{v\in N\setminus\{v^*\}}} \smashoperator[r]{\sum_{i=1}^{\sat_{\elec}(v, \bar w^*)}} f(i) \quad > \smashoperator[r]{\sum_{i=1}^{\sat_{\elec}(v^*, \bar w)-|I|}} f(i) + \smashoperator{\sum_{v\in N\setminus\{v^*\}}} \smashoperator[r]{\sum_{i=1}^{\sat_{\elec}(v, \bar w)\ }} f(i).
\]
This can be rewritten as
\[
  \sum_{v\in N} \smashoperator[r]{\sum_{i=1}^{\sat_{\elec}(v, \bar w^*)}} f(i) \quad - \smashoperator[r]{\sum_{i=\sat_{\elec}(v^*,\bar w^*)-|I|+1}^{\sat_{\elec}(v^*,\bar w^*)}}f(i) \quad > \quad \sum_{v\in N} \smashoperator[r]{\sum_{i=1}^{\sat_{\elec}(v, \bar w)}} f(i) \quad - \smashoperator[r]{\sum_{i=\sat_{\elec}(v^*,\bar w)-|I|+1}^{\sat_{\elec}(v^*,\bar w)}}f(i),
\]
which gives
\begin{equation}
  \sum_{v\in N} \smashoperator[r]{\sum_{i=1}^{\sat_{\elec}(v, \bar w^*)}} f(i) - \sum_{v\in N} \smashoperator[r]{\sum_{i=1}^{\sat_{\elec}(v, \bar w)}} f(i) \quad >  \smashoperator[r]{\sum_{i=\sat_{\elec}(v^*,\bar w^*)-|I|+1}^{\sat_{\elec}(v^*,\bar w^*)}}f(i) \quad - \smashoperator[r]{\sum_{i=\sat_{\elec}(v^*,\bar w)-|I|+1}^{\sat_{\elec}(v^*,\bar w)}}f(i).\label{eq:5}
\end{equation}
Next, we know that the left-hand side of~\eqref{eq:5} must be non-positive, as $\bar w$, and not $\bar w^*$, is the winner of $\elec$. Moreover, we assumed that $\sat_{\elec}(v^*, \bar w) > \sat_{\elec}(v^*, \bar w^*)$, which together with the non-increasingness of $f$ implies that the right-hand side of~\eqref{eq:5} must be non-negative: a contradiction. Hence, we get $\sat_{\elec}(v^*, \bar w) \leq \sat_{\elec}(v^*, \bar w^*)$, completing the proof.
\end{proof}

\safefreeridingleximin*

\begin{proof}

It follows directly from Theorem~\ref{thm:freeriding} that free-riding can increase
the satisfaction of the free-riding voter.
Let us show that it can never decrease the satisfaction of the free-riding voter.
Let $\elec$ be an election, $\bar w$ be the outcome
of $\elec$ under the leximin rule, and consider a voter $v^*$ such that $v^*$ can 
free-ride in issues $I\subset[k]$. Finally, let $\elec^*$ be the election after $v^*$ free-rides 
and $\bar w^*$ the outcome of $\elec^*$ under opt-leximin.
In the following we write $N^{\elec}_i(\bar w) = \{v \in N\mid\sat_{\elec}(v, \bar w) = i\}$.  
Now, as $v^*$ free-rides, i.e., the winners in issues~$I$ are the same in $\bar w$ and $\bar w^*$,
we know that $v^*$ approves the winners of issues $I$ in her honest ballot in $\elec$ and 
does not approve the winners of issues $I$ in her free-riding ballot in $\elec^*$.
It follows the satisfaction of $v^*$ with $\bar w^*$ resp.\ $\bar w$ in $\elec$
is higher by one than in $\elec^*$, i.e.,
\begin{align}
\sat_{\elec^*}(v^*, \bar w^*) &= \sat_{\elec}(v^*, \bar w^*) -|I| \quad\text{as well as }\nonumber\\
\sat_{\elec^*}(v^*, \bar w) &= \sat_{\elec}(v^*, \bar w) -|I|.\label{eq:1}
\end{align}
All other voters submit the same ballot in $\elec$ and $\elec^*$.
Hence, for all $v \neq v^*$ we have 
\begin{align}
\sat_{\elec^*}(v, \bar w^*) &= \sat_{\elec}(v, \bar w^*) \quad\text{as well as  }\nonumber\\
\sat_{\elec^*}(v, \bar w) &= \sat_{\elec}(v, \bar w). \label{eq:2}
\end{align}
Now assume for the sake of a contradiction that 
$\sat_{\elec}(v^*, \bar w) > \sat_{\elec}(v^*, \bar w^*)$, i.e.,
free-riding led to a lower satisfaction for $v^*$ with respect to her honest ballot.

First of all, observe that $\bar w^*$ and $\bar w$ cannot tie in the leximin order in $\elec^*$
(i.e., the victory of $\bar w^*$ over $\bar w$ cannot be due to tiebreaking).
Were this not the case, then we would have $|N^{\elec^*}_j(\bar w^*)| = |N^{\elec^*}_j(\bar w)|$
for every index $j$. However, $\sat_{\elec}(v^*, \bar w) > \sat_{\elec}(v^*, \bar w^*)$ and \eqref{eq:1} imply that
$\sat_{\elec^*}(v^*, \bar w) > \sat_{\elec^*}(v^*, \bar w^*)$. This,
together with \eqref{eq:2}, implies that $|N^{\elec^*}_j(\bar w^*)| \neq |N^{\elec^*}_j(\bar w)|$
for some index $j$.

Consequently, we know that $\bar w^* \succ \bar w$
according to the leximin order in $\elec^*$. In other words, there is a $j$
such that $|N^{\elec^*}_j(\bar w^*)| < |N^{\elec^*}_j(\bar w)|$ and 
$|N^{\elec^*}_\ell(\bar w^*)| = |N^{\elec^*}_\ell(\bar w)|$ for all $\ell < j$.
We claim that the deciding index $j$ cannot be smaller than $\sat_{\elec^*}(v^*, \bar w^*)$
as for all smaller indices $\ell < \sat_{\elec^*}(v^*, \bar w^*)$ it follows from
$\sat_{\elec}(v^*, \bar w) > \sat_{\elec}(v^*, \bar w^*)$ that
$v^*$ is not in $N^{\elec}_\ell(\bar w^*)$, $N^{\elec^*}_\ell(\bar w^*)$, $N^{\elec}_\ell(\bar w)$ 
and $N^{\elec^*}_\ell(\bar w)$.
Therefore, it follows from \eqref{eq:2} that $|N^{\elec}_\ell(\bar w^*)| =
|N^{\elec^*}_\ell(\bar w^*)|$ and
$|N^{\elec}_\ell(\bar w)| = |N^{\elec^*}_\ell(\bar w)|$. Hence,
$j < \sat_{\elec^*}(v^*, \bar w^*)$ would be a contradiction to the assumption that
$\bar w$ is the leximin outcome of $\elec$ and hence leximin preferred to $\bar w^*$ in $\elec$. 

Therefore, we know that $|N^{\elec^*}_\ell(\bar w^*)| = |N^{\elec^*}_\ell(\bar w)|$ for all 
$\ell < \sat_{\elec^*}(v^*, \bar w^*) \leq j$ and 
\[|N^{\elec^*}_{\sat_{\elec^*}(v^*, \bar w^*)}(\bar w^*)| \leq
|N^{\elec^*}_{\sat_{\elec^*}(v^*, \bar w^*)}(\bar w)|.\]
It follows that also $|N^{\elec}_\ell(\bar w^*)| = |N^{\elec^*}_\ell(\bar w^*)|
= |N^{\elec^*}_\ell(\bar w)| = |N^{\elec}_\ell(\bar w)|$
for all $\ell < \sat_{\elec^*}(v^*, \bar w^*) \leq j$.
Finally, it follows from \eqref{eq:1} that $v^*$ is in
$N^{\elec^*}_{\sat_{\elec^*}(v^*, \bar w^*)}(\bar w^*)$ but not in
$N^{\elec}_{\sat_{\elec^*}(v^*, \bar w^*)}(\bar w^*)$.
Moreover, because we assumed $\sat_{\elec}(v^*, \bar w) > \sat_{\elec}(v^*, \bar w^*)$,
$v^*$ is neither in
$N^{\elec^*}_{\sat_{\elec^*}(v^*, \bar w^*)}(\bar w)$ nor in
$N^{\elec}_{\sat_{\elec^*}(v^*, \bar w^*)}(\bar w)$.
Therefore, we have
\[
|N^{\elec}_{\sat_{\elec^*}(v^*, \bar w^*)}(\bar w^*)| +1 =
|N^{\elec^*}_{\sat_{\elec^*}(v^*, \bar w^*)}(\bar w^*)| \leq
|N^{\elec^*}_{\sat_{\elec^*}(v^*, \bar w^*)}(\bar w)| =
|N^{\elec}_{\sat_{\elec^*}(v^*, \bar w^*)}(\bar w)|.
\]
However, that means that $\bar w^*$ is leximin preferred to $\bar w$ in $\elec$, which is
a contradiction to the assumption that $\bar w$ is the outcome of $\elec$. 
\end{proof}

\WDThiele*

\begin{proof}
Fix an opt-$f$-Thiele rule $\calR$ distinct from the utilitarian rule. We show hardness by a reduction from \textsc{CubicVertexCover}, a variant of \textsc{VertexCover} where every node has a degree of exactly three \cite{alimonti2000apx}. In the following, let $\ell$ be the smallest $\ell$ such that $f(\ell)>f(\ell+1)$ (such an $\ell$ must exist by virtue of $\calR$ not being the utilitarian rule).

Consider an instance $(G, k)$ of \textsc{CubicVertexCover}. Here, $k$ is a natural number and $G=(V,E)$ is an undirected graph with $n$ nodes and $m$ edges where every node has a degree of $3$. We assume w.l.o.g.\ that $k<m$.
We will construct an instance $(\elec, \ell+k, c_{d_1})$ of \ODprob{}, where $\elec=\langle N,\bar A, \bar C\rangle$ is an election with $(\ell+k)$ issues and $2m$ voters. Observe in particular that $\ell$ does not depend on $(G, k)$. 

We construct the instance as follows. We have one voter $v_e$ for every edge $e\in E$, plus $m$ extra dummy voters $\{d_1,\dots,d_m\}$. In the first $\ell-1$ issues, there are two candidates $c$ and $c'$, and all voters approve of both. In the next $k$ issues, there is one candidate $c_\eta$ for every node $\eta\in V$, plus one candidate $c_{d_i}$ for every dummy voter $d_i$. In all of these issues, every edge voter $v_e$ approves of the two candidates $c_\eta$ and $c_{\eta'}$ such that $e=\{\eta, \eta'\}$. Furthermore, every dummy candidate $d_i$ approves of only $c_{d_i}$. Finally, in the last issue, there is one candidate $c_v$ for every voter $v\in N$, and every such voter only approves of $c_v$.

To deal with ties, we assume that each issue $i$ is associated with a total ordering $\succ_i$ such that:
\begin{enumerate}
  \item If $i\in\{\ell,\dots,\ell+k-1\}$, then node-candidates are preferred over other candidates, and $c_{d_n} \succ_i \dots \succ_i c_{d_1}$;
  \item If $i=\ell+k$, then all candidates $c_{v_e}$ (with $e\in E$) are preferred over other candidates, and $c_{d_1} \succ_i \dots \succ_i c_{d_n}$;
\end{enumerate}

We compare outcomes $\bar w$ and $\bar w'$ with $\bar w=(w_1,\dots,w_{\ell+k})$ and $\bar w'=(w_1',\dots,w_{\ell+k}')$ lexicographically, that is $\bar w \succ \bar w'$ if there exists an index $j\in [\ell+k]$ such that $w_1=w_1'$, $\dots$, $w_{j-1}=w_{j-1}'$ and $w_{j} > w_{j}'$. Among the outcomes with maximal scores, we return the maximal outcome according to $\succ$.

We want to show that $(G, k)$ is a yes-instance if and only if $(\elec, \ell+k, c_{d_1})$ is. First, note that all voters have a satisfaction of at least $\ell-1$ (because of the first $\ell-1$ issues). Next, let us show the following, useful claim:

\begin{claim}
Let $\bar w$ be an outcome of the election $\elec$, let $\elec_{[-1]}$ be the election
that only differs from $\elec$ in that issue $\ell +k$ is missing and let $\bar w_{[-1]}$
be $\bar w$ restricted to $\elec_{[-1]}$. Then 
\[\thielescore_f(\bar w) =  \thielescore_f(\bar w_{[-1]}) + f(\ell).\] 
\end{claim}

\begin{claimproof}
First, note that at most one dummy voter can win in each issue in $\{\ell, \dots,\ell+k-1\}$. As there are $m$ dummy voters and $k<m$, at least one voter will win no issue in $\{\ell,\dots,\ell+k-1\}$. Thus, whatever outcome we fix for issue $1$ to $\ell +k-1$, there will be at least one voter with satisfaction $\ell-1$. Now, in issue $\ell+k$ every outcome increases the satisfaction of one voter by one. As $f(\ell) > f(\ell+1) \geq f(\ell^*)$ for $\ell^* \geq \ell +1$, it is always optimal to pick a candidate corresponding to a voter with satisfaction $\ell-1$ in $\elec_{[-1]}$.\end{claimproof}\\

Using this fact, we show that $c_{d_1}$ wins in the last issue if and only if
the candidates selected in issues $\ell$ to $\ell +k-1$ correspond to a vertex cover of $G$.

Let $\bar w$ be the winning outcome and assume that the winners of issue
$\ell$ to $\ell +k-1$ correspond to a vertex cover of $G$, i.e., that
$V[\bar w] := \{\eta \in V \mid \exists i \leq k-1 \text{ s.t.\ } w_{\ell + i} = c_{\eta}\}$
is a vertex cover. Then, clearly, every edge voter $v_e$ has a satisfaction of at least $\ell$
in $\elec_{[-1]}$. As we observed above, this means no candidate $c_{v_e}$
can be winning in the last issue. Moreover, $c_{d_i}$ does not win in issues $\ell$
to $\ell +k-1$: choosing this candidate cannot give a higher score than choosing another
dummy candidate $c_{d_i}$ (with $i>1$) for a voter $d_i$ with satisfaction $\ell -1$,
as every such candidate is approved by exactly one dummy voter.
Moreover, the outcome where we replace $c_{d_1}$ by $c_{d_i}$ will always be preferred by our 
tiebreaking. Therefore, $d_1$ has satisfaction $\ell -1$ in $\elec_{[-1]}$.
It follows that $c_{d_1}$ must win in the last issue: Selecting a candidate $c_{v_e}$ leads
to a worse score, and selecting a candidate $c_{d_i}$ for $i > 1$ does not lead to a higher score 
but to an outcome that is less preferred by the tiebreaking.

Now assume that the winners of issues $\ell$ to $\ell +k-1$ do not correspond to
a vertex cover of $G$. Then there is one edge voter $v_e$ with satisfaction $\ell-1$ 
in $\elec_{[-1]}$. Hence, selecting $c_{d_1}$ in the last issue cannot lead to a higher score
than selecting $c_{v_e}$, and the latter is preferred lexicographically.
Thus, $c_{d_1}$ does not win in the last issue.

It remains to show that if there is a vertex cover of $G$ with at most $k$ vertices, then 
$V[\bar w]$ is a vertex cover for the winning outcome $\bar w$. 
Assume a vertex cover of $G$ of size at most $k$ exists. Further, assume for the sake of
a contradiction that $V[\bar w]$ is not a vertex cover. 

If $V[\bar w]$ is not a vertex cover, then
for every issue $i \in \{\ell, \dots, \ell +k-1\}$ the winner must be a vertex candidate
$c_{\eta}$.
Assume otherwise that there is an issue $i \in \{\ell, \dots, \ell +k-1\}$
where a candidate $c_{d_j}$ wins. We observe that because $\bar w$ does not correspond
to a vertex cover, there is at least one voter $v_e$ that has satisfaction $\ell-1$ in $\elec_{[-1]}$.
Then $w_i$ contributes at most $f(\ell)$ to the score of $\bar w$.
If the winner in the last round is not $c_{v_e}$ we can replace 
$c_{d_j}$ by $c_{\eta}$ which would contribute at least $f(\ell)$ to the score and be preferable 
by tiebreaking. If the winner in the last round is $c_{v_e}$, then we can replace $c_{d_j}$
in issue $i$ by $c_{\eta}$ and $c_{v_e}$ in the last issue by any other candidate 
corresponding to a voter with satisfaction $\ell -1$ without the last issue. The score
of the resulting outcome is at least as good as the score $\bar w$ and it is preferred by 
tiebreaking.


Now, let $E_{\bar w}^\ell := \{v_e \in N \mid e\in E \land \sat_{\elec_{[-1]}}(v_e,\bar w_{[-1]})\geq\ell\}$
be the set of edge-voters with satisfaction at least $\ell$ in $\bar w_{[-1]}$.
We define $E_{\bar w}^{\ell+1}\subseteq E_{\bar w}^\ell$ analogously.
Now, we observe that
\begin{multline*}
  \sum_{v_e\in E_{\bar w}^{\ell}}\left(\sat_{\elec_{[-1]}}(v_e,\bar w_{[-1]})-(\ell-1)\right) = \\
  |E_{\bar w}^\ell|+\sum_{v_e\in E_{\bar w}^{\ell+1}}\left(\sat_{\elec_{[-1]}}(v_e,\bar w_{[-1]})-\ell\right) = 3|V[\bar w]| = 3k
\end{multline*}
because the set contains $k$ nodes and each node has degree 3.

Finally, any outcome $\bar w^*$ on $\elec_{[-1]}$ in which on every issue
in $\{\ell, \dots, \ell+k-1\}$ a vertex candidate $c_{\eta}$ wins has the following Thiele score:
\begin{equation*}
\sum_{v\in V} \left(\sum_{i=1}^{\ell-1} f(i)\right) + |E_{\bar w}^\ell| f(\ell) + 
\underbrace{\sum_{v_e\in E_{\bar w}^{\ell+1}}\sum_{i=\ell+1}^{\sat_{\elec_{[-1]}}(v_e,\bar w_{[-1]})} f(i)}_{3k-|E_{\bar w}^\ell|\text{ addends}}.
\end{equation*}
As $f(\ell) > f(\ell +1)$, this function is maximized by maximizing $|E_{\bar w}^\ell|$.
As a vertex cover exists, we know that we can reach $|E_{\bar w}^\ell|= m$. Hence, the outcome
maximizing the Thiele score must do so, which means that it must be a vertex cover.
\end{proof}

\WDOWA*

\begin{proof}

Fix a rule $\calR$ satisfying the condition of the theorem. We show hardness by a reduction from \textsc{CubicVertexCover}. Consider an instance $(G, k)$ of this problem. Here, $G=(V,E)$ is a graph with $n$ nodes and $m$ edges where each node has a degree of exactly three, and $k\in\mathbb{N}$. We assume w.l.o.g. that $k<n$.
We construct an instance of \ODprob{} with $(k+1)$ issues and $3m$ voters.
As $\alpha_1>\alpha_{3m}$, there are two cases:

\begin{enumerate}
  \item There is a $p\in[2m]$ such that $\alpha_p>\alpha_{p+1}$, or
  \item There is a $p>2m$ with $p<3m$ such that $\alpha_1=\alpha_p>\alpha_{p+1}$.
\end{enumerate}

In the following, we treat these cases separately.

\paragraph{First case.}

We construct an instance $(\elec, k+1, c_{d_1})$ of \ODprob{}. Here, we have one voter $v_e$
for each edge $e\in E$, and two sets of dummy voters, $\{d_1,\dots,d_p\}$ and
$\{w_1,\dots,w_{2m-p}\}$. In the first $k$ issues, there is one candidate $c_\eta$ for each node $\eta\in V$,
plus one dummy candidate $c_{d_i}$ for each dummy voter $d_i$. Here, each edge-voter $v_e$ approves of the two node-candidates $v_\eta$ and $v_{\eta'}$
such that $e=\{\eta,\eta'\}$. Moreover, each dummy voter $d_i$ approves only of dummy candidate $c_{d_i}$, and all dummy candidates $w_i$ approve of all candidates.
In the last issue, there is one candidate $c_v$ for all voters $v\in N\setminus\{w_i\}_{i\in[2m-p]}$, and every such $v$ only approves of $c_v$. Finally, here, all voters in $\{w_i\}_{i\in[2m-p]}$ approve of all candidates.

We use the following tiebreaking mechanism, which is essentially identical to the one used in the proof of Theorem~\ref{thm:WDThiele}.
We assume that each issue $i$ is associated with a total ordering $\succ_i$ such that:
\begin{enumerate}
  \item If $i\in\{1,\dots,k\}$, then node-candidates are preferred over other candidates, and $c_{d_n} \succ_i \dots \succ_i c_{d_1}$;
  \item If $i=k+1$, then all candidates $c_{v_e}$ (with $e\in E$) are preferred over other candidates, and $c_{d_1} \succ_i \dots \succ_i c_{d_n}$;
\end{enumerate}

We compare outcomes $\bar w$ and $\bar w'$ with $\bar w=(w_1,\dots,w_{\ell+k})$ and $\bar w'=(w_1',\dots,w_{\ell+k}')$ lexicographically, that is $\bar w \succ \bar w'$ if there exists an index $j\in [\ell+k]$ such that $w_1=w_1'$, $\dots$, $w_{j-1}=w_{j-1}'$ and $w_{j} > w_{j}'$. Among the outcomes with maximal scores, we return the maximal outcome according to $\succ$.

We want to show that $(G, k)$ is a yes-instance if and only if $(\elec, k+1, c_{d_1})$ is. 
Suppose that there exists a vertex cover for $G$ with size at most $k$. First, we show that all edge-voters must win at least one issue in $[k]$. Then, we show that, if all edge-voters win at least one issue in $[k]$, then $c_{d_1}$ wins in issue $k+1$.

Let us show that all edge-voters win at least one issue in $[k]$. Let $\bar w=\calR(\elec)$, and assume towards a contradiction that some edge-voter $v_e$ never win any issue in $[k]$. Assume that some dummy candidate $c_{d_j}$ wins some issue $i\in[k+1]$. If $v_e$ never wins at all, we can make $c_\eta$ (for some $\eta\in e$) win in issue $i$ and obtain an outcome that has a greater or equal score (as $\alpha$ is nonincreasing) and is preferred lexicographically. If $v_e$ wins in issue $k+1$, we can make a similar argument by making $c_\eta$ win issue $i$ and $c_{d_j}$ win issue $k+1$. Hence, in the following, we assume w.l.o.g. that no dummy candidate wins in $\bar w$.

Next, let $\bar w^*$ be some outcome where all edge-voters win at least one issue in $[k]$ (which is possible, because $(G,k)$ is a yes-instance), no dummy voter wins any issue in $[k]$ while each node-candidate is chosen at most once (which is possible, since $k<n$), and some dummy voter wins in issue $k+1$. We will show that $\bar w^*$ leads to a strictly higher score than $\bar w$.

First, observe that, in both outcomes, since each time a node-candidate is selected exactly three edge-voters approve of it, the total satisfaction (ignoring the dummy voters $w_i$) will be $3k+1$ (the extra $1$ comes from the last issue). Next, let $s=(s_1,\dots,s_{m+p})$ and $s^*=(s^*_1,\dots,s^*_{m+p})$ be the sorted satisfaction vectors (ignoring the dummy voters in $\{w_i\}_{i\in[2m-p]}$) when $\bar w$ and $\bar w^*$ are the outcomes, respectively. Furthermore, let $i_1$, $i_2$ and $i_3$ be the three smallest indexes such that $s_{i_1}=1$, $s_{i_2}=2$, and $s_{i_3}=3$ hold. If any of these indexes is undefined, we set it to $m+p+1$. Moreover, we define $i_1^*$ and $i_2^*$ analogously for $s^*$ (observe that no voter here can have satisfaction greater than $2$). Clearly $i_1^*=p<i_1$, as $\bar w$ satisfies once at most $3m-p$ voters, whereas $\bar w^*$ satisfies once exactly $3m-p+1$ voters. We get that:
\begin{multline*}
    \owascore_\alpha(\bar w)<\owascore_\alpha(\bar w^*) \implies \alpha\cdot s < \alpha\cdot s^*\implies\\
    \sum_{i=i_1}^{m+p}\alpha_i+\sum_{i=i_2}^{m+p}\alpha_i+\sum_{i=i_3}^{m+p}(s_i-2)\alpha_i < \sum_{i=p}^{m+p}\alpha_i+\sum_{i=i^*_2}^{m+p}\alpha_i\implies\\
    \sum_{i=i_2}^{m+p}\alpha_i+\sum_{i=i_3}^{m+p}\sum_{j=3}^{s_i}\alpha_i < \sum_{i=p}^{i_1-1}\alpha_i+\sum_{i=i^*_2}^{m+p}\alpha_i.
\end{multline*}
If $i_2^*<i_2$, we obtain:
\begin{equation*}
  \sum_{i=i_3}^{m+p}\sum_{j=3}^{s_i}\alpha_i < \sum_{i=p}^{i_1-1}\alpha_i+\sum_{i=i^*_2}^{i_2-1}\alpha_i.\\
\end{equation*}
Since $i_1-1\leq i_2-1<i_3$, every addend occurring on the left-hand side is smaller or equal to every addend occurring on the right side. In particular, $\alpha_p$ is positive and strictly greater than all the addends on the left side (as $p<i_3$). Furthermore, since $\sum_i s_i=\sum_i s_i^*=3k+1$, there is the same number of addends being summed on both sides. It follows that $\owascore_\alpha(\bar w)<\owascore_\alpha(\bar w^*)$. If, on the other hand, $i_2^*\geq i_2$, we obtain:
\begin{equation*}
  \sum_{i=i_2}^{i_2^*-1}\alpha_i+\sum_{i=i_3}^{m+p}\sum_{j=3}^{s_i}\alpha_i < \sum_{i=p}^{i_1-1}\alpha_i.\\
\end{equation*}
Since $i_1-1<i_2$ and $i_1-1<i_3$, by similar arguments as above, we conclude that $\owascore_\alpha(\bar w)<\owascore_\alpha(\bar w^*)$. But this is impossible, as we assumed that $\bar w$ is the outcome. We have finally reached the required contradiction: it cannot be that some edge-voter $v_e$ loses all issues in $[k]$.

Now, let us show that $c_{d_1}$ wins in $k+1$ if all edge-voters win at least once in issue in $[k]$. If voter $d_1$ never won an issue in $[k]$, then it means she has a satisfaction of $0$. Since all edge-voters and all the $w_i$ won at least once, there are at least $m+2m-p=3m-p$ voters with a satisfaction of at least $1$. Therefore, $d_1$
occupies a position within the first $p$ entries of the satisfaction vector, whereas all edge-voters occupy a position within the last $4m-p$ entries. Since $\alpha_p>\alpha_{p+1}$, in this case choosing in issue $k+1$ candidate $c_{d_1}$ will yield a greater score than choosing a voter-candidate $c_{v_e}$ for any edge $e\in E$. Finally, since $c_{d_1}$ dominates in the tiebreaking every other candidate $c_{d_j}$ in issue $k+1$, here we must choose $c_{d_1}$. On the other hand, suppose that $d_1$ wins at least one issue $i\in[k]$. Suppose -- towards a contradiction -- that $c_{d_1}$ is not selected in issue $k+1$. Let $c_v$ (for some voter $v\in N\setminus\{w_i\}_{i\in[2m-p]}$ distinct from $d_1$) be the candidate winning issue $k+1$. Observe that if we make $c_{d_1}$ win in issue $k+1$ and make some candidate approved by $v$ win in issue $i$, we would obtain a score that is higher or equal than before, and this would surely be preferred by tiebreaking: contradiction. We conclude that $c_{d_1}$ must win in the final issue.

Finally, suppose that there exists no vertex cover for $G$ with size at most $k$. Then, surely there is one edge-voter that never wins an issue in $[k]$ (otherwise, some vertex cover would exist). By tiebreaking, this edge-voter would decide the last issue, i.e., $c_{d_1}$ would not win.

\paragraph{Second case.}

Here, we can assume that $\alpha_1=\dots=\alpha_p=1>\alpha_{p+1}$. We construct an instance $(\elec, k+1, c)$ of \ODprob{}. Here, we have one voter $v_e$
for each edge $e\in E$, and three sets of dummy voters: $\{d_1,\dots,d_{p-2m+1}\}$,
$\{a_1,\dots,a_m\}$, and $\{w_1,\dots,w_{3m-p-1}\}$. In the first $k$ issues, there is one candidate $c_\eta$ for each node $\eta\in V$,
plus one dummy candidate $c_{d_i}$ for each dummy voter $d_i$. Here, each edge-voter $v_e$ (with $e=\{\eta,\eta'\}$) approves of every node-candidate $v_\eta$ where $\eta\not\in e$. Furthermore, each dummy voter $d_i$ approves only of dummy candidate $c_{d_i}$. Every other dummy candidate approves of all candidates.
In the last issue, there are two candidates $c$ and $c'$. All dummy candidates $d_i$ and $w_i$ approve of both, every edge-voters approves only of $c$, while every dummy voter $a_i$ approves only of $c'$.

We assume a tiebreaking mechanism almost identical to the one used in the proof of Theorem~\ref{thm:WDThiele}. However, here, in the last issue, $c$ loses against $c'$.

First, suppose that there exists a vertex cover for $G$ with size at most $k$. Since every dummy candidate $c_{d_i}$ is always approved only by one voter and $\alpha_1=\alpha_p=1$, it is easy to see that any outcome where at least one such dummy candidate wins in the first $k$ issues cannot have maximal score. Now, observe that, in total, the edge-voters will receive exactly a score of $k(m-3)$ for the first $k$ issues (as every time we select some node-candidate, $m-3$ voters approve of it). This does not depend on which node-candidates we select; thus, to determine the outcome with the greatest score, we can focus on the last issue. First, observe that if the node-candidates selected in the first $k$ issues correspond to a vertex cover, no edge-voter will have won more than $k-1$ issues within the first $k$ issues (as every edge-voter loses at least once). Now, focusing on the last issue, note that the last $4m-p-1$ positions of the satisfaction vector will be occupied by dummy voters in $\{a_1,\dots,a_m\}$, and $\{w_1,\dots,w_{3m-p-1}\}$ (as they all have a satisfaction of at least $k$). Thus, if $c$ wins in the last round, we get an extra score of $\sum_{i=p-2m+2}^{p-m+1}\alpha_i=m$; if $c'$ wins, we get a score of $\sum_{i=p-m+2}^{p+1}\alpha_i=(m-1)+\alpha_{p+1}$. As $\alpha_{p+1}<1$, here $c$ wins. Similarly, if the node-candidates selected in the first $k$ issues \emph{do not} correspond to a vertex cover, we would still get a score of $(m-1)+\alpha_{p+1}$ for $c'$ winning in the last round. So we have that $c$ wins in the last issue if a vertex cover of size $k$ exists.

Now, suppose that there exists no vertex cover for $G$ with size at most $k$. Consider any outcome $\bar w$. Then, surely there is one edge-voter that wins all issues in $[k]$ (otherwise, if every edge-voter loses at least once, then some vertex cover would exist). In issue $k+1$, ignoring the voters supporting both candidates, both $c$ and $c'$ have exactly $m$ voters supporting them, and in case either $c$ or $c'$ wins, the last $3m-p-1$ positions of the satisfaction vectors would be occupied by the dummy voters $w_i$ (that have all satisfaction $k+1$). If $c'$ wins, then we get an extra score of $\sum_{i=p-m+2}^{p+1}\alpha_i=(m-1)+\alpha_{p+1}$ (recall that all voters approving only of $c'$ have a satisfaction of at least $k$, if we ignore the last issue). If $c$ wins, we get at most the same score (as at least one voter has satisfaction $k$, she will occupy the $(p+1)$-position in the satisfaction vector). By tiebreaking, $c$ cannot win in $\bar w$.\\

This concludes the proof.\end{proof}

\RFThiele*

\begin{proof}
We show the hardness of \RFprob{} by a reduction from \textsc{CubicVertexCover}. Again, let $\ell$ be the 
smallest $\ell$ where $f(\ell)>f(\ell+1)$ holds, and consider an instance $(G, k)$ of \textsc{CubicVertexCover}. Here, $k\in[m-1]$, and $G=(V,E)$ is an undirected graph with $n$ nodes and $m$ edges where every node has degree of $3$.
We construct an instance $(\elec, \ell+k, c_{d_1}, d_2)$ of \RFprob{}, where $\elec=\langle N,\bar A, \bar C\rangle$ is an election with $(\ell+k)$ issues and $2m$ voters.
We use a construction similar to the one in the proof of Theorem~\ref{thm:WDThiele},
except for the fact that, on issue $\ell+k$, voter $d_2$ approves only of $c_{d_1}$.

First, suppose that $(G, k)$ is a yes-instance. By the same arguments used in the proof of Theorem~\ref{thm:WDThiele}, we know that $c_{d_1}$ wins in issue $\ell+k$ (the fact that now $d_2$ also approves of it is irrelevant). Moreover, if $d_2$ votes for $c_{d_2}$ in the last issue, we obtain the same election constructed in the proof of Theorem~\ref{thm:WDThiele}. We have already shown that here $c_{d_1}$ wins the final issue: hence, $d_2$ can free-ride.

Now, suppose that $(G, k)$ is a no-instance. Let us first show that $c_{d_1}$ is still selected for issue $\ell+k$. Clearly, at least one edge-voter does not win any issue in $\ell,\dots,\ell+k$ (otherwise, a vertex cover would exist). Towards a contradiction, suppose that in $\calR(\elec)$ some dummy candidate $c_{d_j}$ is winning some issue $i\in\{\ell,\dots,\ell+k-1\}$. Then, at least one edge-voter must have satisfaction $\ell-1$ (for if all edge-voters were to have satisfaction at least $\ell$, we could cover all but one edge with $k-1$ nodes). So let $v_e$ be some edge-voter with satisfaction $\ell-1$ in $\calR(\elec)$. If we select $c_\eta$ (for some $\eta\in e$) in $i$ instead of $c_{d_j}$, we would increase the total score by at least $f(\ell)$ (contributed by $v_e$) and decrease it by $f(\ell^*)$ for some $\ell^*\geq \ell$ (contributed by $d_j$). Since $f(\ell^*)\leq f(\ell)$, this new outcome cannot have a lower score than $\calR(\elec)$, and would be preferred to it by the tiebreaking. Contradiction: we conclude that no dummy voter $d_j$ can win in $\ell,\dots,\ell+k-1$. This implies, in particular, that neither $d_1$ nor $d_2$ win any issue in $\ell,\dots,\ell+k-1$. Therefore, selecting $c_{d_1}$ for issue $\ell+k$ contributes $2f(\ell)$ to the total score, whereas selecting any other candidate can contribute at most $f(\ell)$. Since $f(\ell)>f(\ell+1)\geq 0$, $c_{d_1}$ must win in the final issue.

It remains to show that $d_2$ cannot free-ride. Suppose that $d_2$ does not approve of $c_{d_1}$. Consider some edge-voter $v_e$ that never wins in $\ell,\dots,\ell+k-1$ (which, as argued above, must exist). Choosing $c_{v_e}$ in the final issue contributes at least $f(\ell)$ to the total score, whereas choosing $c_{d_1}$ can contribute at most $f(\ell)$. By tiebreaking, $c_{d_1}$ does not win in the final issue, that is, $d_2$ cannot free-ride.

To conclude, observe that the same construction can be used to show hardness for \genRFprob{}. Indeed, here $d_2$ only approves of $c_{d_1}$, and hence free-riding and generalized free-riding coincide.\end{proof}

\RFOWA*

\begin{proof}

We show the hardness of \RFprob{} by a reduction from \textsc{CubicVertexCover}. Consider an instance $(G, k)$ of this problem. Here, $G=(V,E)$ is a graph with $n$ nodes and $m$ edges where each node has a degree of exactly three, and $k\in\mathbb{N}$. By the condition of the theorem, we know there is an $\ell\geq 3m$ (polynomial in the size of $m$) such that $\alpha=(\alpha_1,\dots,\alpha_\ell)$ contains at least $3m$ non-zero entries and $\alpha_1>\alpha_\ell$.
We will construct an instance of \RFprob{} with $(k+1)$ issues and $\ell$ voters.
Since $\alpha_1>\alpha_\ell$ and $\alpha_{3m}>0$, there are two cases:

\begin{enumerate}
  \item There is a $p\in[2m]$ such that $\alpha_p>\alpha_{p+1}$ and $\alpha_{p+m}>0$, or
  \item There is a $p\in\{2m+1,\dots,\ell-1\}$ such that $\alpha_1=\alpha_p>\alpha_{p+1}$.
\end{enumerate}

We treat them separately.

\paragraph{First case.}

We construct an instance $(\elec, k+1, v_{e^*}, c_{d_1})$ of \RFprob{} (here, $e^*\in E$ is some edge, it does not matter which). The construction is similar to the one shown in the first case of the proof for Theorem~\ref{thm:WDOWA}. However, here, in issue $k+1$ voter $v_{e^*}$ approves only of $c_{d_1}$, and we have $\ell-m-p$ dummy voters $w_i$ instead of $3m-p$. The latter change makes no difference in our construction.

First, note that $(\elec, k+1, v_{e^*}, c_{d_1})$ is indeed a legal instance of \RFprob{}, as surely $c_{d_1}$ wins in issue $k+1$. If $(G, k)$ is a yes-instance then we have already shown that this candidate wins, and here it is only receiving increased support. If it is a no-instance, then $c_{d_1}$ will be supported by one voter that never won in the first $k$ issues (namely, $d_1$), as well as by $v_{e^*}$. Since $\alpha_{p+m}>0$ and since the edge-voters together with the dummy voters $d_i$ occupy at most the first $p+m$ positions of the satisfaction vector, $v_{e^*}$ will break the tie in favour of $c_{d_1}$.

Now, if $(G, k)$ is a yes-instance of \textsc{CubicVertexCover}, then $v_{e^*}$ can free-ride in the last issue: if she votes for her voter-candidate, then we have an election identical to the one constructed in the first case of the proof of Theorem~\ref{thm:WDOWA}, and we have already shown there that $c_{d_1}$ wins if $(G, k)$ has a vertex cover.

If $(G, k)$ is a no-instance, then there are two cases: either $v_{e^*}$ won in some issue in $[k]$ or not. If she did, there will be at least one voter $v_e$ (with $e\in E\setminus\{e^*\}$) that never did, whose voter-candidate will get at least the same score as $c_{d_1}$ (since $v_{e^*}$ does not approve of the latter when she free-rides): $c_{d_1}$ cannot win here. If she did not, there are again two cases: either $v_{e^*}$ approves of some dummy candidate $c_{d_i}$ (with $i>1$) or of some $c_{v_e}$ (where $e\in E$). In the first case, $c_{d_i}$ would get a strictly higher score than $c_{d_1}$, while in the second case $c_{v_e}$ would get at least the same score as $c_{d_1}$ (and win by tiebreaking). In all cases, $c_{d_1}$ loses: no free-riding is possible.
\paragraph{Second case.}

We construct another instance $(\elec, k+1, a_1, c)$ of the problem \RFprob{}. The construction is similar to the one shown in the second case of the proof
for Theorem~\ref{thm:WDOWA}. However, here, in issue $k+1$ voter $a_1$ approves only of $c$, and we have $\ell-p-1$ dummy voters $w_i$ instead of $3m-p-1$. The latter change makes no difference in our construction.

First, note that $(\elec, k+1, a_1, c)$ is a legal instance of \RFprob{}. Consider how the election is constructed, and recall that no dummy voter $d_i$ can ever win here. In the last issue, $c'$ receives the support of $m-1$ voters, whereas $c$ receives the support of $m+1$ voters. Since the last $\ell-p-1$ positions are occupied by voters $w_i$ (who approve of all candidates and have a satisfaction of $k+1$), and since $\alpha_1=\alpha_p=1$, $c$ wins here.

If $a_1$ free-rides in $k+1$, she must vote only for $c'$. Here, we obtain a construction identical to the one shown in the second case of the proof Theorem~\ref{thm:WDOWA}, and we have already shown there $c$ wins if and only if $(G, k)$ is a yes-instance.\\

Finally, observe that, in both cases, the same construction can be used to show hardness for \genRFprob{}. Indeed, here the manipulator only approves of one alternative, and hence free-riding and generalized free-riding coincide. This concludes the proof.\end{proof}

\RFOWAzero*

\begin{proof}

We show the hardness of \RFprob{} by a reduction from \textsc{VertexCover} \cite{garey1979computers}. Consider an instance of this problem, $(G, k)$, where $G$ has $n$ nodes and $m$ edges. By the condition of the theorem, we know there is an $\ell>m$ (polynomial in the size of $m$) such that $\alpha=(\alpha_1,\dots,\alpha_\ell)$ contains at least $m$ zeros and at least one non-zero value.
We will construct an instance $(\elec, k+1, d_1, c_{d_1})$ of \RFprob{} with $\ell$ voters and $k+1$ issues. Here, let $p$ be the unique value such that $\alpha_{p}>\alpha_{p+1}=0$.

In $\elec$, there is one voter $v_e$ for each edge $e\in E$, $p$ dummy voters $d_1,\dots,d_{p}$, and $\ell-m-p$ dummy voters $w_1,\dots,w_{\ell-m-p}$. In all issues, all voters $w_i$ approve of all candidates (so they always have satisfaction $k+1$, and occupy the last $\ell-m-p$ positions of the satisfaction vector). In the first $k$ issues, there is one candidate $c_\eta$ for each node $\eta\in V$, plus one candidate $c_{d_i}$ for every dummy voter $d_i$. Here, each edge-voter $v_e$ approves of the two node-candidates $v_\eta$ and $v_{\eta'}$ such that $e=\{\eta,\eta'\}$, while every $d_i$ approves of $c_{d_i}$. In the last issue, there is one candidate $c$, plus one candidate $c_v$ for all voters $v\in N$, and any such $v$ approves only of $c_v$. In the case of ties, we assume that in the last issue $c_{d_1}$ dominates all candidates and that $c_{d_1}$ is dominated by all other candidates in all other issues.

We will show that $(\elec, k+1, d_1, c_{d_1})$ is a legal instance of \RFprob{} (i.e., that $c_{d_1}$ wins in issue $k+1$). Furthermore, we will show that $(G, k)$ is a yes-instance if and only if $(\elec, k+1, d_1, c_{d_1})$ is a no-instance.

Suppose $(G, k)$ is a no-instance of \textsc{VertexCover}. We show that $c_{d_1}$ wins in issue $k+1$ and that $d_1$ can free-ride here. First, observe that if any dummy candidate $c_{d_i}$ wins in issue $k+1$, then the total satisfaction will be $0$ (as we cannot give a satisfaction of at least $1$ to all edge-voters in the other $k$ issues, since $(G,k)$ is a no-instance). On the other hand, if any candidate $c_{v_e}$ (for some edge $e$ wins), then the total satisfaction will still be $0$: otherwise, that would mean that we could cover the remaining edges in $E\setminus\{e\}$ with $k-p$ nodes, but that is impossible (otherwise, we could cover all edges with $k$ nodes). By tiebreaking, $c_{d_1}$ wins in the final issue. Observe that if $d_1$ votes for $c$, we obtain the same effects: $d_1$ can free-ride.

Suppose $(G, k)$ is a yes-instance of \textsc{VertexCover}. We show that $c_{d_1}$ wins in issue $k+1$, but $d_1$ cannot free-ride here. Clearly, at least $m+1$ voters (ignoring the dummy voters $w_i$) need to win here: was not this the case, the total score would be zero, but the outcome where all edge-voters win in the first $k$ issues and some dummy candidate wins in the last issue has a greater satisfaction (regardless of whether $d_1$ free-rides or not). If $d_1$ never wins in the first $k$ issues, then it is clear that she must win in issue $k+1$: satisfying in this issue some voter that has never won surely will maximize the score (since this voter will be within the first $p$ entries of the vector), and $c_{d_1}$ is preferred in the tiebreaking mechanism. If, on the other hand, $c_{d_1}$ wins in some issue $i\in[k]$, but loses to some candidate $c_v$ (with $v\neq d_1$) in issue $k+1$, then we could obtain an outcome with a score greater or equal by making $c_{d_1}$ win in issue $k+1$, and making some candidate approved by $v$ win in issue $i$. Since this would be preferred in the tiebreaking, $c_{d_1}$ wins in $k+1$. Now, suppose that $d_1$ does not approve of $c_{d_1}$ in issue $k+1$ (i.e., she attempts to free-ride). If $c_{d_1}$ never wins in any issue in $[k]$, then clearly it cannot win in $k+1$: picking some candidate that now $d_1$ approves for in the last issue would give a greater score. If, on the other hand, $c_{d_1}$ wins in some issue in $i\in[k]$ and also in $k+1$, we can obtain an outcome with a greater or equal score (and preferred in the tiebreaking) by making some node-candidate win in issue $i$ and some candidate approved by $d_1$ in issue $k+1$. Therefore, $d_1$ cannot free-ride here.

Finally, observe that the same construction can be used to show hardness for \genRFprob{}. Indeed, here $d_1$ only approves of $c_{d_1}$, and hence free-riding and generalized free-riding coincide. This concludes the proof.\end{proof}

\end{document}